\documentclass{new_tlp}

\makeatletter
\let\O@argtabularcr\@argtabularcr
\def\O@xtabularcr{\@ifnextchar[\O@argtabularcr{\ifnum 0=`{\fi}\cr}}
\let\O@tabacol\@tabacol
\let\O@tabclassiv\@tabclassiv
\let\O@tabclassz\@tabclassz
\let\O@tabarray\@tabarray
\def\author@tabular{\authorsize\def\@halignto{}\@authortable}
\let\endauthor@tabular=\endtabular
\def\author@tabcrone{{\ifnum0=`}\fi\O@xtabularcr\affilsize\itshape
 \let\\=\author@tabcrtwo\ignorespaces}
\def\author@tabcrtwo{{\ifnum0=`}\fi\O@xtabularcr[-3\p@]\affilsize\itshape
 \let\\=\author@tabcrtwo\ignorespaces}
\def\@authortable{\leavevmode \hbox \bgroup $\let\@acol\O@tabacol
 \let\@classz\O@tabclassz \let\@classiv\O@tabclassiv
 \let\\=\author@tabcrone \ignorespaces \O@tabarray}
\makeatother

\DeclareMathAlphabet\mathbfcal{OMS}{cmsy}{b}{n}

\usepackage[ruled]{algorithm2e}
\usepackage{thmtools}
\usepackage{fancybox}
\usepackage{times}
\usepackage{amssymb}
\usepackage{amsfonts}
\usepackage{algorithmic}
\usepackage{color}
\usepackage{listings} 

\newcommand{\eat}[1]{}

\usepackage{latexsym}
\usepackage{amsfonts}
\usepackage{amsmath}
\usepackage{amssymb}
\usepackage{color}
\usepackage{colortbl}
\usepackage{epsfig}
\usepackage{xspace}
\usepackage{graphicx}
\usepackage{subfigure}
\usepackage{enumerate}
\usepackage{comment}
\usepackage[all]{xy}
\usepackage{epsfig}
\usepackage{multirow}
\usepackage{hhline}
\usepackage{url}
\usepackage[table]{xcolor}
\newcommand{\bi}{\begin{itemize}}
\newcommand{\ei}{\end{itemize}}

\newcommand{\be}{\begin{enumerate}}
\newcommand{\ee}{\end{enumerate}}
\newcommand{\beqn}{\begin{eqnarray*}}
\newcommand{\eeqn}{\end{eqnarray*}}

\newcommand{\stitle}[1]{\vspace{0.6ex}\noindent{\bf #1}}

\newcommand{\ie}{\emph{i.e.,}\xspace}
\newcommand{\eg}{\emph{e.g.,}\xspace}

\newcommand{\wrt}{\emph{w.r.t.}\xspace}

\newcommand{\kwlog}{\emph{w.l.o.g.}\xspace}
\newcommand{\kWlog}{\emph{W.l.o.g.}\xspace}

\newlength\sindent
\newcommand{\kw}[1]{{\ensuremath {\mathsf{#1}}}\xspace}

\newcommand{\cf}{\kw{cf.}}

\newcounter{ccc}

\newcommand{\eop}{\hspace*{\fill}\mbox{$\Box$}\vspace{1ex}}     % End of proof
\newcommand{\nthesection}{\arabic{section}}
\newcounter{alg}[section]
\renewcommand{\thealg}{\nthesection.\arabic{alg}}

\newcounter{arule}
\renewcommand{\thearule}{\arabic{arule}}

\newcounter{claim}
\renewcommand{\theclaim}{\arabic{claim}}

\newtheorem{definition}{Definition} % [section]
\newtheorem{example}{Example} % [section]
\newtheorem{proposition}{Proposition} % [section]
\newtheorem{lemma}{Lemma} % [section]
\newtheorem{conjecture}{Conjecture} % [section]
\newtheorem{corollary}{Corollary} % [section]
\newtheorem{theorem}{Theorem} % [section]
\begin{document}
\bibliographystyle{acmtrans}

\newcommand{\todo} [1]{\textcolor{blue}{{\sf TODO}: #1}}

%\long\def\comment#1{}

%\title{Towards Coordination-free,  Data-Parallel Systems}
\title{A Datalog-based Computational Model  for Coordination-free,  Data-Parallel Systems}

\author[M.Interlandi and L.Tanca]
{Matteo Interlandi\thanks{Work partially done while at University of California, Los Angeles.}\\\textit{Microsoft}
\\\textit{E-mail: mainterl@microsoft.com}
\and Letizia Tanca
\\\textit{Politecnico di Milano}
\\\textit{E-mail: letizia.tanca@polimi.it}
}

\pagerange{\pageref{firstpage}--\pageref{lastpage}}
\volume{\textbf{10} (3):}
\jdate{March 2002}
\setcounter{page}{1}
\pubyear{2002}

\maketitle

\label{firstpage}

\begin{abstract}

\emph{Cloud computing} refers to maximizing efficiency by sharing computational and storage resources, while \emph{data-parallel systems} exploit the  resources available in the cloud to perform parallel transformations over large amounts of data. In the same line, considerable emphasis has been recently given to two apparently disjoint research topics: \emph{data-parallel}, and \emph{eventually consistent, distributed} systems.

\emph{Declarative networking} has been recently proposed to ease the task of programming in the cloud, by allowing the programmer to express only the desired result and leave the implementation details to the responsibility of the run-time system. In this context,  we deem it appropriate to propose a study on a \emph{logic-programming-based computational model} for  eventually consistent, data-parallel systems, the keystone of which is provided by the recent finding that the class of programs  that can be computed in an eventually consistent, coordination-free way is that of \emph{monotonic programs}. This principle is called CALM and has been proven by Ameloot et al. [2013] for distributed, asynchronous settings.

We advocate that CALM should be employed as a basic theoretical tool also for data-parallel systems, wherein computation usually proceeds synchronously in rounds and where communication is assumed to be reliable.
We deem this problem relevant and interesting, especially for what concerns \emph{parallel dataflow optimizations}. Nowadays we are in fact witnessing an increasing concern about  understanding which properties distinguish synchronous from asynchronous parallel processing, and when the latter can replace the former. 
It is general opinion that coordination-freedom can be seen as a major discriminant factor.

In this work we make the case that the current form of CALM does not hold in general for data-parallel systems,
% if the techniques developed by Ameloot et al. are directly used, 
 and show how, using novel techniques, the satisfiability of the CALM principle can still be obtained although  just for the subclass of programs called \emph{connected monotonic queries}.
We complete the study with considerations on the relationships between our model and the one employed by Ameloot et al., showing 
 that our techniques subsume the latter when the synchronization constraints imposed on the system are loosened. 
 
 Under consideration in Theory and Practice of Logic Programming (TPLP).

\end{abstract}

\section{Introduction}
%!TEX root = TPLP.tex

Recent research has explored ways to exploit different levels of \emph{consistency} in order to improve the performance of distributed systems \wrt specific tasks and network configurations while preserving correctness \cite{Vogels:2009:EC:1435417.1435432,DeCandia:2007:DAH:1323293.1294281,Brewer:2000:TRD:343477.343502}.
A topic strictly related to consistency is \emph{coordination}, usually informally interpreted as a mechanism to accomplish a distributed agreement on some system property \cite{FaginHMV03}.
%If we consider consistency as such, 
Indeed, coordination can be used to enforce consistency when, in the natural execution of a system, the latter is not guaranteed. %in general. 
%For instance, in MapReduce, reduce tasks start their computation strictly after map tasks have completed (coordination). This permits to avoid situations in which late-arriving values can jeopardize the final outcome (consistency).
%

In this paper we set forth a logic-programming-based framework to express database queries and study some theoretical problems springing from the use of eventually consistent, coordination-free computation over \emph{synchronous systems with reliable communication} (\emph{rsync} in short). %Eventual consistency is a form of weak consistency in which the storage system guarantees that, if no new updates are made to the object, eventually all accesses will return the last updated value. %and we will introduce other interesting results which are not emerging in the purely asynchronous case.
\emph{Rsync} is a common setting in modern data-parallel frameworks such as MapReduce \cite{DeanG20}, Pregel \cite{Malewicz:2010:PSL:1807167.1807184}, and Apache Spark \cite{ZahariaCD12}, where
 computation is commonly performed in \emph{rounds}, and each task is blocked and cannot start the new round until a \emph{synchronization barrier} is reached, \ie every other task has completed its local computation. %//TOCHECK
 %In this context,  we deem it appropriate to propose a study on a \emph{logic-programming-based computational model} for  eventually consistent, data-parallel systems, the keystone of which is provided by the recent finding that a class of programs exists that can be computed in an eventually consistent, coordination-free way: \emph{monotonic programs}. This principle is called CALM and has been proven by Ameloot et al. [2013] for distributed, asynchronous settings.
\eat{In this work 
%, differently then commonly assumed, 
we consider synchronization (barrier) and coordination as two different, although related entities: the former is a \emph{mechanism} enforcing the \emph{rsync} model; the latter a \emph{property of executions}.}
%For instance, it is known that in \emph{rsync} settings, for \emph{Graph Coloring} and \emph{Gibbs Sampling} \cite{GonzalezLGG11} to converge a proper coordination logic must be applied \cite{LowBGG12,XieCGZ15}.
%Nevertheless, as the intuition tells us, and as we will see in the following, coordination come almost for free in \emph{rsync} systems.
%Hence it is not the synchronization barrier per se, but the need of coordination that produces ``computation-breaker" points, where all tasks must wait and input records must be fully materialized before distributed computation can safely proceed.
%As a consequence, 
Identifying under what circumstances eventually consistent, coordination-free computation can be performed over \emph{rsync} systems would enable the introduction of  novel execution plans, no longer restricted
by predefined (synchronous) patterns. For example, coordination-free programs can be divided into independent sub-units that can be run concurrently: a property known as \emph{decomposability}~\cite{DBLP:conf/sigmod/WolfsonS88}. 
While our recent work~\cite{Shkapsky:2016:BDA:2882903.2915229} implements the generalized pivoting technique~\cite{generalizedpivoting} by which a decomposable plan can be identified from  simple syntactic analysis of the program(s), still no semantics study exists on the matter.

%This suggest that a tradeoff exists between the simplicity of physical query plans and coordination patterns, and amenably to optimizations.

Our aim is therefore to understand \emph{in what generic circumstances a synchronous ``blocking'' computation is actually required by the program semantics -- and therefore must be strictly enforced by the system -- and when, instead, an asynchronous execution can be performed as optimization}.
Recently, the class of programs that can be computed in an eventually consistent, coordination-free way has been identified: \emph{monotonic programs} \cite{Hellerstein10};
this property is called \emph{CALM} (Consistency and Logical Monotonicity) and has been proven in \cite{AmelootNB13}. %in the case of asynchronous systems.
While CALM was originally proposed to simplify the specification of distributed (asynchronous) data management systems, %-- also known as \emph{declarative networking} \cite{Ameloot:2014} --  
in this paper we advocate that CALM should be employed as a basic theoretical tool also for the declarative specification of data-parallel (synchronous) systems. 
As a matter of fact, CALM permits to link a property of the execution (coordination-freedom) to a class of  programs, i.e. monotonic queries.
But to which extent does CALM  hold over data-parallel systems?
%In this paper we study the theoretical results springing from the application of the CALM principle over \emph{synchronous systems with reliable communication} (\emph{rsync}).
Surprisingly enough, %the demonstration of the CALM principle in \emph{rsync} systems is not trivial and, 
with the communication model and the notion of coordination as defined in \cite{AmelootNB13}, the CALM principle does not hold in general in \emph{rsync} settings, the main reason being that the proposed definition of coordination is too weak to capture the type of coordination ``baked" into the synchronization barrier of \emph{rsync} systems  (\cf Example \ref{ex:false_calm}). In this paper we first characterize such type of coordination, then we study \emph{to which extent the ``synchronization barrier" creates coordination}, and  finally we devise additional forms of coordination patterns.

%In order to extend CALM over data-parallel synchronous computation, %similarly to MapReduce, we will introduce a 
%interpretation of coordination exploiting the \emph{syncausality relation} \cite{DBLP:conf/wdag/Ben-ZviM10} to generalize the previous one. 
%Under the new definition of coordination, not only the CALM Theorem indeed also holds in $rsync$ systems, but actually \emph{two different levels} of coordination can be isolated: one which can be evaluated in a \emph{communication-free} way and whose efficiency is directly influenced by the non-functional properties of the system, and one instead that is system-independent, communication is necessarily required, and is influenced by the given input data.
%\noindent Following other parallel programming models such as MapReduce, in this section we will introduce a
 %content-based communication model founded on hashing.
% More concretely, 
To reach our goal, we develop a new \emph{generic parallel computation model}, leveraging previous works on logic-based  \emph{relational transducers} \cite{AbiteboulVFY00} and \emph{transducer networks} \cite{AmelootNB13}, and grounding \emph{rsync} computation on the well-known \emph{Bulk Synchronous Parallel (BSP)} model \cite{Valiant90}. 
With BSP,  
computation proceeds in a series of global rounds, each comprising three phases: $(i)$ a \emph{computation phase}, in which nodes concurrently perform local computations; $(ii)$ a \emph{communication phase}, in which data is exchanged among the nodes; and $(iii)$ the synchronization barrier.
%Many of today's data-parallel frameworks  not only assume the underlying system to be \emph{rsync}, but they actually directly \cite{Malewicz:2010:PSL:1807167.1807184} or indirectly \cite{Pace2012246} implement the BSP model.
% since the major parallel data-processing framework  %model for relational parallel data-processing
Exploiting this new type of transducer network, equipped with a content-based addressing model, we then show that the CALM principle is in general satisfied for BSP-style systems \emph{under a new  definition of coordination-freedom}, although, surprisingly enough, just for a subclass of monotonic queries, \ie the \emph{connected monotonic queries} (\cf Definition \ref{df:connected}).
When defining coordination-freedom, we will take advantage of recent results %born in the field of knowledge representation and reasoning over  groups of nodes \cite{FaginHMV03},
describing how knowledge can be acquired in synchronous systems \cite{Ben-Zvi:2014:BLH:2605175.2542181}.
As a final outcome, a series of coordination patterns -- and related classes of characteristic queries -- is identified,
and we will discuss how these coordination patterns behave under BSP and weaker synchronous settings.
As a corollary, we show that the new  definition of coordination-freedom subsumes the one employed in \cite{AmelootNB13}.
%The results discussed in this paper are the theoretical foundations of a new declarative data-parallel system that we are currently developing \cite{Interlandi:vision:2014}. By exploiting CALM, in fact, we are able to break the cage of modern parallel computation models, and provide asynchronous executions when allowed by the program logic.

\stitle{Contributions}:
Summarizing, the contributions of the paper are as follows:
\begin{enumerate}
\item %Using the techniques developed in \cite{Ameloot:2011:RTD:1989284.1989321}, 
The only-if direction of the CALM principle (namely that only monotonic queries can be computed in a coordination-free way) is proven not to hold in general for \emph{rsync} systems  (\cf Example \ref{ex:false_calm} in Section~\ref{sec:co_co});
% if the techniques developed by Ameloot \cite{Ameloot:2011:RTD:1989284.1989321} are used;
\item A novel, logic-programming-based computational model is introduced that emulates common patterns found in modern data-parallel frameworks (Section~\ref{sec:network_independence});

\item A new definition of coordination is proposed, leveraging recent results on knowledge acquisition in \emph{rsync} systems (Section~\ref{sec:refine_coordination}); 

\item Exploiting the new techniques, the CALM principle is proven to hold for  \emph{connected monotonic queries} in \emph{rsync} systems with bounded delay and deterministic data delivery (Theorem 1 in Section~\ref{sec:sync_calm}); 

\item A complete taxonomy of queries is provided that permits the identification of different types of coordination patterns (Section~\ref{sec:sync_calm}); and

\item The definition of coordination previously introduced in \cite{AmelootNB13} is shown to collapse into the one we propose, when the synchronization constraints assumed on the system model are loosened (Section~\ref{sec:generic_snapshot_delay}).
\end{enumerate}

\stitle{Applications}:
Data-parallel programs such as the one implemented on top of MapReduce and Apache Spark relies on the assumption that computation is executed in rounds. This assumption makes the implementation of distributed programs easier because hides to the developers details on coordination and on how tasks are executed by the systems.
Nevertheless, coordination-free execution is shown to be faster when allowed.
In fact, certain algorithms are shown to converge faster when (a bounded amount of) asynchrony is permitted~\cite{Niu:2011:HLA:2986459.2986537,Cui:2014:EBS:2643634.2643639}.
In general, it is well known that coordination-free (asynchronous) computations are amenable to \emph{pipelining}, i.e., one-record-at-a-time executions of sequences of computations requiring no intermediate materialization of data;
overall pipelining is highly desirable in the Big Data context, where full materialization is often problematic because of the limits on the available main memory. % --  \eg it is known that in Spark \cite{ZahariaCD12} scalability is limited by the amount of  main memory available.
Finally, as previously mentioned, coordination-free programs are decomposable~\cite{DBLP:conf/sigmod/WolfsonS88}. 
%\footnote{\Eg Hive \cite{Thusoo:2009:HWS:1687553.1687609} and Pig \cite{Olston:2008:PLN:1376616.1376726} sacrifice pipelining  to fit query plans into MapReduce workflows.}.

Currently, all high-level data-parallel languages are compiled into synchronous (blocking) plans;
for instance both Hive \cite{ThusooSJS09} and Pig \cite{OlstonRSK08} sacrifice efficiency in order to fit query plans into rounds of MapReduce jobs. Similarly, Spark SQL statically splits programs into stages separated by ad-hoc coordination logic.
Other more sophisticated systems such as Hyracks \cite{BorkarCF11} and Apache Flink~\cite{AlexandrovBEF14} do provide the ability to pipeline operators, but it is the programmer's task to manually select the proper strategy.
To our knowledge, only BigDatalog~\cite{Shkapsky:2016:BDA:2882903.2915229} and few other systems \cite{XieCGZ15,HanD15,Cui:2014:EBS:2643634.2643639,Niu:2011:HLA:2986459.2986537} have started to explore when asynchronous executions can be delivered as optimizations of parallel, synchronous programs, and these are mainly in the graph-processing and machine learning domain.
%Notably, BigDatalog implements the generalized pivoting technique~\cite{generalizedpivoting} by which a decomposable plan can be identified from  simple  syntactic analysis of the program(s). Unfortunately, decomposability is undecidable in general~\cite{DBLP:conf/sigmod/WolfsonO90}.

\stitle{Organization}:
The rest of the paper is organized as follows: Section \ref{preliminaries} introduces some preliminary notation. %and concepts. 
Section \ref{sec:synchronous_transducer_network} defines our model of synchronous and reliable parallel system, and shows that the CALM principle is not satisfied for  systems of this type. 
Section \ref{sec:network_independence} proposes a new computational model based on hashing, while Section \ref{sec:refine_coordination} introduces the new definition of coordination. Finally, Section \ref{sec:sync_calm} discusses CALM under the new setting. %and describes how the previous definition of coordination used in \cite{Ameloot:2011:RTD:1989284.1989321} can be reached by weakening the assumptions on the system.
The paper ends with a comparison with other work and concluding remarks.
%Due to space limitations, we will only provide proof sketches of relevant results.%\footnote{We refer the interested readers to our technical report \url{http://arxiv.org/abs/1405.7264}.}.
The extended version of this paper~\cite{DBLP:journals/corr/InterlandiT14} contains some additional material: $(i)$  result on the decidability of independent specifications; and $(ii)$ a complete study on the expressive power of the bulk synchronous transducer network model.
%An extended version of this work~\cite{extended} contains some additional material: $(i)$  result on the decidability of independent specifications; and $(ii)$ a complete study on the expressive power of the bulk synchronous transducer network model.
 %\footnote{For a more detailed discussion we refer the interested readers to our technical report \url{http://arxiv.org/abs/1405.7264}.}. %The Appendix contains some further examples and a reference table for the notation.%if interested in digging into the full proofs.
%novel definition of coordination subsumes the previous one introduced in \cite{Ameloot:2011:RTD:1989284.1989321}
%Finally, we will show that although non-monotonic queries cannot be expressed without resorting coordination, certain subclass of such queries can be specified using communication-free transducers.
%We name this new theorem, \emph{Embarrassingly Parallel is Inflationary oblivious and Communication-free} (\emph{EPIC}).
%The CALM and the EPIC theorems provide the theoretical background which will be exploited in Chapter \ref{chap6} where we will discuss an approach to efficiently parallelly compute generic Datalog$^\neg$ queries.

\vspace{-1ex}
\section{Preliminaries}
\label{preliminaries}
%!TEX root = TPLP.tex

The ultimate goal of this paper is to understand to which extent the execution of high-level data-parallel SQL-like languages such as Hive~\cite{ThusooSJS09}, Pig~\cite{OlstonRSK08}, and Spark SQL~\cite{ArmbrustXL15} can be optimized by understanding coordination patterns.  
In this section we therefore recall  the basic notions of database theory whereby queries (expressed as logic programs) are evaluated in a bottom-up fashion. We also  set forth  our notation, %and to briefly introduce some formal background. %in order to make this paper as self-contained as possible.
%The notation we follow hereafter
which is close to that of \cite{AbiteboulHV95} and \cite{AmelootNB13}.

\vspace{-1ex}
\subsection{Basic Database Notation}
\label{sec:basic}

%Let start with some preliminaries on database theory.
%Given a finite set of \emph{relation names} \textbf{relname}, 
We denote by $\mathcal{D}$ an arbitrary \textit{database schema} composed by a non empty set of \emph{relation schemas} (or simply \emph{relations}). %\in$ \textbf{relname}. %Often, with a slight abuse of notation, and when clear from the context, we will  use the word \emph{relation} instead of  \emph{relation schema}. %In the following we will use interchangeably the terms \textit{relations} and \textit{predicates}. 
%We associate to each relation name $R$ a function $arity : R \rightarrow \mathbb{N}_0$. 
In the following we will use the notation $R^{(a)}$ %or equivalently $(R, a)$ 
to denote a relation name together with its \emph{arity} $a$.
With \textbf{dom} we indicate a countably infinite set of \textit{constants}. %disjoined from \textbf{relname}. %, and with \textbf{var} an infinite set of \textit{variables} used to range over the elements of \textbf{dom}. We consider \textbf{dom}, \textbf{var} and \textbf{relname} as disjoined from one another. 
Given a relation $R^{(a)}$, a \emph{fact} $R(\bar{\texttt{u}})$  is an ordered $a$-tuple over $R$  composed by constants only. %or, more precisely, an element of the Cartesian product $(\textbf{dom})^a$. 
%For a fact $R(\bar{\texttt{u}})$, we refer to $R$ as the \emph{predicate}. %or \emph{relation name}.
%With $(\texttt{u}_1, \texttt{u}_2, \ldots, \texttt{u}_a)$ we denote a tuple of arity $a$, %we use $\bar{\texttt{u}}(i)$ to refer to the $i$-th \emph{term} of $\bar{\texttt{u}}$, with $i \le a$.
%while we write $( )$ to denote a tuple of arity zero, which we name \emph{nullary} tuple.
%Given a tuple $\bar{u}$ defined over $R$, if $S$ is a relation such that $S \subseteq R$, $\bar{u} \mid_{S}$ denotes the \emph{restriction} of $\bar{u}$ to $S$. 
%A \textit{fact} over $R$ is an expression in the form $R(u_1, \ldots, u_k)$ where each $u_i$ is referred to as a \emph{term}. 
%We sometimes write $R(\bar{u})$ to refer to a fact and $\bar{u}(i)$ to refer to the $i$-th term. 
%A \textit{literal} is an atom -- in this case we refer to it as \textit{positive} -- or the negation of an atom. 
%Now, a \textit{ground atom} %or equivalently a \textit{fact},
%is an atom containing only constant terms -- \ie $u_i \in$ \textbf{dom}.
A \textit{relation instance} $I_R$ is a set of facts defined over $R \in \mathcal{D}$, while a \textit{database instance} \textbf{I} is the union $\bigcup_{R \in \mathcal{D}} I_R$.
In general we write $I_\mathcal{D^\prime}$ to denote an instance over the relations $\mathcal{D^\prime} \subseteq \mathcal{D}$.
%In general, given an instance \textbf{I} defined over $\mathcal{D}$, if $\mathcal{D}^\prime$ is a set of relations such that $\mathcal{D}^\prime \subseteq \mathcal{D}$, $I_{\mathcal{D}^\prime}$ denotes the set of fact defined over the relations in $\mathcal{D}^\prime$. 
 %In the following we will use the relation instance signature $I_{R^\prime}$, to denote a subset of \textbf{I} defined over the set of relations $R^\prime \subseteq \textbf{R}$, \ie $I_R^\prime = \bigcup_{R \in R^\prime}I_R$.
The set $adom(\textbf{I}) \subseteq \textbf{dom}$ of all constants appearing in a given database instance \textbf{I} is called \textit{active domain} of \textbf{I}, while $inst(\mathcal{D})$ denotes the set of all the possible database instances defined over $\mathcal{D}$.
%Similarly, \textit{adom}($I_R$) is used to denote the active domain of a relation instance $I_R$.
%
Given a database schema $\mathcal{D}$ and a relation $R \in \mathcal{D}$, %we considers of interest one type of \emph{database transformations}, \ie functions mapping instances over $\mathcal{D}_{in}$ to instances over $R$: \emph{queries}.
%we considers two types of \emph{database transformations} \ie  functions mapping an -- given a database schema \textbf{R}, if $in \subseteq \textbf{R}$ and $out \subseteq \textbf{R}$ are two schemas with $in \cap out = \emptyset$, 
%If we assume $\mathcal{D}_{in}$ and $R$ to be disjoined, 
a \emph{query} $q_R$ is a total function such that $q_R : inst(\mathcal{D}) \rightarrow inst(R)$ and $adom(q(\textbf{I})) \subseteq adom(\textbf{I})$.
%We will sometime omit to specify the input and the output schema of the query if unimportant for the discussion.
In practice, we will only consider \emph{generic} queries, \ie if $p$ is a permutation of \textbf{dom}, and $\textbf{I}$ an input instance, then $q(p(\textbf{I})) = p(q(\textbf{I}))$.
Finally, we say that a query is \emph{monotonic} when given two instances $\textbf{I}$, $\textbf{J}$, if $\textbf{J} \subseteq \textbf{I}$, then $q(\textbf{J}) \subseteq q(\textbf{I})$.
%Similarly, we define an update $u_R$ to be a database transformation such that $\mathcal{D}_{in} \supseteq R$. Also in this case we will consider only generic updates.
%Hereafter when we refer to a query or an update we always actually indicate a generic query or a generic update.
Note that in the above definitions we have considered queries with a single output relation; this is not a limitation since queries with multiple output relations can be expressed as collections thereof:
given an input and output schema $\mathcal{D}_{in}$ and $\mathcal{D}_{out}$ respectivelly, we will write 
$\mathcal{Q} = \{q_R \mid R \in \mathcal{D}_{out}\}$. %with $\mathcal{D}^\prime \subseteq \mathcal{D}$.
In this paper we will consider the following query languages all expressible using a rule-based formalism: unions of conjunctive queries {\sc ucq}, first order queries {\sc fo}, {\sc datalog}, and {\sc datalog} with negation {\sc datalog}$^\neg$. 
%The full connection can be at the physical level (nodes are actually physically connected through some communication mean) or at logical level (a routing level is responsible for making the graph fully connected).
%This is a natural assumption in the parallel computing context, and is indeed perpetrated by many models, such as \cite{Dean:2008:MSD:1327452.1327492,Koutris:2011:PEC:1989284.1989310,Valiant:1990:BMP:79173.79181} just to mention a few.
%Assumed this, we specify a distributed system as being composed by a network graph $\textbf{Net} = (N, E)$ where $N$ is the set of nodes composing the network, with $|N|$ the total number of nodes, while $E$ are the edges connecting them.
%Each node identifier $i$ has a value in the set of natural numbers, \ie $N = \{1, \ldots, n\}$.
%For simplicity, in the following we will moreover consider a distributed system to be composed just by the set of nodes $N$ instead of the complete network graph $\textbf{Net}$.
%We can make this assumption w.l.o.g. since the network graph is granted to be fully connected.
%
Next we briefly introduce the syntax of the above query languages.

\vspace{-1ex}
\subsection{Query Languages}

Let \textbf{var} be an infinite set of \textit{variables} ranging over the elements of \textbf{dom}. %We consider \textbf{dom}, \textbf{var} and \textbf{relname} as disjoined from one another.
Given a relation $R$, an \emph{atom} $R(\bar{u})$ is a tuple in which both constants and variables are permitted as terms.
A \emph{literal} is an atom -- in this case we refer to it as \emph{positive} -- or the negation of an atom. 
%A \emph{ground atom} is a positive atom containing only constant terms.

A \emph{conjunctive query with negation},  is an expression in the form of:
\begin{equation}
H(\bar{w}) \leftarrow B_1(\bar{u}_1), \ldots, B_n(\bar{u}_n), \neg C_1(\bar{v}_1), \ldots, \neg C_m(\bar{v}_m).
\end{equation}
where $H(\bar{w})$, $B_i(\bar{u}_i)$, and $C_j(\bar{v}_j)$ are atoms. As usual $H(\bar{w})$ is referred to as the \textit{head}, and $B_1(\bar{u}_1)$, \ldots, $\neg C_m(\bar{v}_m)$ as the \textit{body}. If $m = 0$ the rule is called \textit{positive} while if $m=n=0$ the rule is expressing a fact. 
For simplicity, in this paper we assume each query to be \textit{safe}, \ie every variable, occurring either in a rule head or in a negative literal, appears in at least one positive literal of the rule body.
%As a consequence, we assume all facts to be ground because all non-ground facts are unsafe.

A {\sc ucq} query is a union of positive conjunctive queries, represented as a set of positive rules. %Therefore,  when working within {\sc ucq},  we will name \emph{rule} each single disjunct composing a {\sc ucq} query.  
%A {\sc fo} query is a union of conjunctive queries, possibly with negation. %and a \emph{rule} is  each single conjunct composing a {\sc fo} query .
In a {\sc datalog} query all predicates appearing in the body of a rule must be positive, and can also be used in the heads of rules to eventually produce \emph{recursive} computation.
%We will use {\sc nrdatalog} to denote a non-recursive {\sc datalog} query.
 A {\sc datalog}$^\neg$ query is a set of safe rules where both recursion and negation are allowed.
For {\sc datalog}$^\neg$ we will assume the \emph{stratified semantics}~\footnote{A {\sc datalog}$^\neg$ program is said stratifiable if it can be  partitioned into sub-programs  (\ie \emph{strata}), each defining one or more negated predicates, and where no cycle of recursion contains a negated predicate~\cite{AbiteboulHV95}. According to the stratified semantics, these sub-program are  then evaluated in order, following the dependencies among the negated predicates: initially the sub-programs having no dependency on negated predicates are fired, followed then by the strata depending on those that have  just been executed, and so forth.}.
Finally, since non-recursive {\sc datalog}$^\neg$ queries are equivalent to first-order logic~\cite{AbiteboulHV95}, we will use {\sc fo} to denote such class of queries.
In this paper we will only consider languages belonging to the above introduced set, with {\sc Datalog}$^\neg$ being the most expressive. Therefore, for each language $\mathcal{L}$, unless otherwise specified, we will assume $\mathcal{L} \subseteq$ {\sc Datalog}$^\neg$.

We will sometimes employ the function $sch$ to return from a query its schema, \ie $\mathcal{D} = sch(\mathcal{Q}).$
The \textit{intensional} ($idb$) part of the database schema is the subset of the database schema containing all the relations that appear in at least one non-fact rule head, while we refer to all the other relations in $sch(\mathcal{Q}) \setminus idb$ as \textit{extensional} ($edb$). 

\vspace{-1ex}
\subsection{Distributed Systems}
We define a \textit{distributed system} as a \emph{fully connected} graph of communicating nodes $N = \{1, \ldots, n\}$. 
We assign to each node $i$ a \emph{node configuration} denoted by the pair $(N, i)$.
We will in general assume all nodes to share the same \emph{global} database schema. % to be \emph{global}, meaning that every node in the network has the same schema.
%Despite this, given a relation $R$, we will use $R^i$ to refer to the \emph{local view} over $R$ for node $i \in N$. 
%We employ this notation because every node maintains its \emph{local instance} $I^i$.
We will use the notation $I^i_{R}$ to denote a \emph{local instance} for node $i$ over a relation $R$, while
%Intuitively, a \emph{global instance} can be reconstructed as the union of all the local instances.
a \emph{global instance} over $R$ is defined as $I_{R} = \bigcup_{i \in N} I^i_{R}$. 
Given an initial  database instance \textbf{I} $\in inst(\mathcal{D}^\prime)$ defined over a subset of the global schema $\mathcal{D}$, we assume that a \emph{distribution function} $D$ exists mapping each node $i$ to a (potentially overlapping) portion of the initial instance, this is $D : inst(\mathcal{D}^\prime) \times N \rightarrow inst(\mathcal{D}^\prime)$.
For correctness, we assume that $D$ is such that each fact composing the input instance is mapped to at least one node, i.e., $\bigcup_{i \in N} D(\textbf{I}, i) = \textbf{I}$.
Finally, a \emph{network configuration} is identified with the pair $(N, D)$.
%We will then use the notation $I^i_{R}$ to denote a \emph{local} instance for node $i$ over $R^i$, so, we have that $P(i) = I^i_{db}$.
%In the opposite way we have that the initial global instance $\textbf{Init}$ is the union of all the partitions, \ie $\textbf{Init} = \bigcup_{i \in N} P(i)$.
%Note that since each partitioned instance is a portion of a database instance, which by assumption is finite, the former are also finite.
%
%We used this two different notation because the schema is one for all the nodes, while the instances can be different node to node. In addition, with this formalism a distributed instance can be implemented by an hash function
%

%Finally, we assume to have a (possibly infinite) \emph{input tape} containing an ordered sequence of consecutive natural numbers starting from 0.
%The sequence of values stored in the input tape will be used for providing to the network the clock driving the computation. This is necessary because of the synchronous assumption, \ie the computation is driven by a global clock accessible by every node.

\vspace{-1ex}
\section{Computation in rsync}
\label{sec:synchronous_transducer_network}
%!TEX root = TPLP.tex

Query computability is usually defined using the classical model of computation: the Turing machine.
%The basic idea is that a query, to be computable, must be implementable by a TT.
However in this paper we are interested in %distributed computation.
%In the previous Sections we have recalled the definition of query computability in not distributed settings by employing TT. But what is 
the meaning of \emph{computing a query in parallel settings}.
To this end, in the following we introduce a novel kind of \emph{transducer network} \cite{AmelootNB13}, where computation is \emph{synchronous} and communication is \emph{reliable},
%in the following we will introduce the notion of \emph{relational transducer} \cite{Abiteboul:1998:RTE:275487.275507} and \emph{transducer network} \cite{Ameloot:2011:RTD:1989284.1989321}.
%In particular, we will introduce a special version of transducer and transducer network where computation is \emph{synchronous}, and communication is \emph{reliable}. 
%The notation introduced hereafter mostly follows the one of \cite{Ameloot:2013:RTD:2450142.2450151}.
%This permits us to define how a set of relational transducers can be assembled to 
thus obtaining an abstract computational model for \emph{distributed data-parallel systems}.
As a first step, next we describe how \emph{relational transducers} \cite{AbiteboulVFY00,AmelootNB13} (hereafter simply transducer) can be used to model local computations.

\subsection{Relational Transducers}
\label{sec:relational_transducer}

We employ transducers as an abstraction modeling the behavior of each single computing node composing a computer cluster: this abstract computational model permits us to make our results as general as possible without having to rely on a specific framework, since
transducers and \emph{transducer networks} (introduced in the next Section) can be easily used to describe any modern data-parallel system. %such as MapReduce and Spark.

We consider each node to be equipped with an immutable \emph{database}, and a \emph{memory}-store used to maintain useful data between consecutive computation steps. 
%Given a database schema $\mathca{D}$, and three new extensional relation names \texttt{id}, \texttt{all} and \texttt{time}.
%We now partition $\mathca{D}$ in several schemas: the \emph{(extensional) database} edb
%Informally, a transducer is a transition system programmed as a deductive database.
In addition, a node can produce an \emph{output} for the user and can also \emph{communicate} some data with other nodes (data communication will be clarified in Section \ref{sec:transducer_network} with the concept of transducer network).
An internal \emph{time}, and \emph{system} data are kept mainly for configuration purposes.
Every node executes a \emph{program} that translates a set of input instances  (from the database, the memory and the communication channel), to a new set of instances that are either saved to memory, or directly output to the user, or addressed to other nodes. Programs are expressed in one of the languages of Section~\ref{sec:basic}.

Formally, each node is modeled as a transducer $\mathcal{T}$ defined by the pair $(\mathcal{P}, \Upsilon)$ where $\mathcal{P}$ and $\Upsilon$ respectively denote the \emph{transducer program} and the \emph{transducer schema}.
A transducer schema is  a 6-tuple ($\Upsilon_{db}$, $\Upsilon_{mem}$, $\Upsilon_{com}$, $\Upsilon_{out}$, $\Upsilon_{time}$, $\Upsilon_{sys}$) of disjoint relational schemas, respectively called $data\-base$, $memory$, $communication$, $output$, $time$ and $system$ sche\-mas. As default, we consider $\Upsilon_{sys}$ to contain two unary relations \texttt{Id}, \texttt{All}, while $\Upsilon_{time}$ includes just the unary relation \texttt{Time},
employed to store the current transducer local \emph{clock} value~\footnote{The semantics of $\Upsilon_{time}$ will become clearer in Section \ref{sec:synchronous_systems} when we will describe the synchronous model.}. %issued by the input tape, therefore its unique term is ranging over the set of natural numbers $\mathbb{N}_0$.
%Hereafter, w.l.o.g. and to ease the notation, we will consider not compulsorily to have schemas $\Upsilon_{mem}$, $\Upsilon_{sync}$ and $\Upsilon_{sys}$ disjoined. 
%A transducer \emph{configuration} is a tuple $(N, i)$ where %$\textbf{I}$ is an initial instance over $\Upsilon_{db}$, 
%$N$ is the set of identifiers defining the relation $\texttt{all}$, while $i \in N$ is the identifier of the transducer, \ie the value stored in \texttt{id}. 
%A transducer configuration basically provides the \emph{context} in which the transducer program is run. 
%For this reason we sometime refer to a configuration as context.
A \emph{transducer local state} over the schema $\Upsilon $ is then an instance \emph{I} over $\Upsilon_{db} \cup \Upsilon_{mem} \cup \Upsilon_{out}\cup \Upsilon_{sys}$.
The transducer program $\mathcal{P}$ %= Q_{ins}$ $\bigcup$ $Q_{del}$ $\bigcup Q_{out}$ $\bigcup Q_{snd}$ 
is composed by a collection of \emph{insertion}, \emph{deletion}, \emph{output} and \emph{send} queries  $Q_{ins} = \{q^{ins}_{R} | R \in \Upsilon_{mem}\}$, $Q_{del} = \{q^{del}_{R} | R \in \Upsilon_{mem}\}$, $Q_{out} = \{q^{out}_{R} | R \in \Upsilon_{out}\}$, and $Q_{snd} = \{q^{snd}_{R} | R \in \Upsilon_{com}\}$, all taking as input an instance over the schema $\Upsilon$.
%\begin{itemize}
%\vspace{-1mm}
%\item a set of \emph{insertion} queries $Q_{ins} = \bigcup_{R \in \Upsilon_{mem}} q^{ins}_{R}$
%\vspace{-1mm}
%\item a set of \emph{deletion} queries $Q_{del} = \bigcup_{R \in \Upsilon_{mem}} q^{del}_{R}$
%\vspace{-1mm}
%\item a set of \emph{output} queries $Q_{out} = \bigcup_{R \in \Upsilon_{out}} q^{out}_{R}$
%\vspace{-1mm}
%\item a set of $retraction$ queries $Q_{ret} = \bigcup_{R \in \mathcal{RT}_{ret}} q^{R}_{ret}$
%\item a set of \emph{emission} queries $Q_{snd} = \bigcup_{R \in \Upsilon_{sync}} q^{snd}_{R}$
%\vspace{-1mm}
%\end{itemize}
%
%A \emph{timed relational transducer} $\mathcal{T}$ (hereafter simply transducer) is defined by the pair $(\mathcal{P}, \Upsilon)$. 

Starting from a relational transducer $\mathcal{T} = (\mathcal{P}, \Upsilon)$ and %and a partitioning function $P$. 
a node configuration $(N, i)$, we can construct a \emph{configured transducer}, denoted by $\mathcal{T}^i_{N}$, by setting $I_{Id} = \{\texttt{Id}(\texttt{i})\}$ and $I_{All} = \{\texttt{All}(\texttt{j}) | \texttt{j} \in N\}$. %, $I_{time} = \texttt{time}(0)$.
Given a configured transducer and an instance $\textbf{I}$ defined over $\Upsilon_{db}$, we can create a \emph{transducer initial local state} by setting $I_{db} = \textbf{I}$.
This basically models the starting status of a computing node, before the actual program execution starts: a node has received a program and a read-only instance (e.g., stored in a distributed file system such as HDFS where data is immutable) over which the computation must be performed, and has global knowledge of the other nodes composing the network.
Indeed this is exactly how working nodes are set up, for instance in MapReduce or Spark.
%{\bf Letizia: questa parte va bene? Questi $I_{rcv}$ non dovevano essere modificati???}

Now, given a configured transducer $\mathcal{T}^i_{N}$, let $I_{rcv}$, $J_{snd}$ denote two instances over $\Upsilon_{com}$ -- and hence disjoint from $I$ -- with the former identifying a set of facts that have been previously sent to $\mathcal{T}^i_{N}$. 
If $I$ is a local state, a \emph{transducer transition}, denoted by $I, I_{rcv}{\Rightarrow}J, J_{snd}$ is such that $J$ is the updated local state, while $J_{snd}$ contains a set of facts that must be addressed to other transducers.
The semantics for updates leaves the $database$ and the $system$ instances unchanged, while %$time$ is updated with the successive clock value, and 
 the facts produced by the insertion query $Q_{ins}$ are inserted into the $memory$ relations and all the facts returned by the deletion query $Q_{del}$ are removed from them. In case of conflicts -- \ie a fact is simultaneously added and removed -- we adopt the no-op semantics. %deletion to always win. 
%Similarly, in each $output$ relations are inserted the facts generated by the insertion query $Q_{ins}$ and removed all the facts returned by the retraction query $Q_{ret}$. In case of conflict, \ie a fact is simultaneously added and removed, the no-op semantics is employed.
As a result for the user, the set of tuples derived by the query $Q_{out}$ are output. %Note that output facts, once produced, they cannot be retracted. %while insertion/deletion are effective just in the successive state. 
As regards $J_{snd}$, this is the set of facts returned by query $Q_{snd}$ and sent by the transducer towards the other nodes. %where $I^\prime = I \cup I_{rcv}$.
%Also in this case 
We assume that, once sent or output, \emph{facts cannot be retracted}.
%$J = \mathcal{P}_{st}(I \cup I_{rcv})$, and $J_{snd} = \mathcal{P}_{sync}(I \cup I_{rcv})$.
%
%as follows: if we denote %with $t$ a clock time value taken from the tape, 
%with $I$ a transducer state and $I_{rcv}$ an input instance over $\Upsilon_{sync}$, %and $I_{time} = \texttt{time}(t)$, where $t$ is a time value provided by the input tape, 
%in addition, if we partition $\mathcal{P}$ in two subprograms $\mathcal{P}_{st}$ (where $st$ is for \emph{state}) and $\mathcal{P}_{sync}$ (for \emph{synchronous}), and $I^\prime = I \cup I_{rcv}$, we have that the partial function $\mathcal{P}_{st}$ is such that $J = \mathcal{P}_{st}(I^\prime)$ is the new transducer state. 
%Similarly, $\mathcal{P}_{sync}$ takes as input an instance $I^\prime$ and produce a set of facts $J_{snd}$ over $\Upsilon_{sync}$.
%More precisely, the $\mathcal{P}_{st}$ is such that: %(recall that $J = \mathcal{P}_{up}(I^\prime)$, with $I^\prime = I \cup I_{rcv}$):

If $I^\prime = I \cup I_{rcv}$, a transducer transition $I, I_{rcv}{\Rightarrow}J, J_{snd}$  %{\bf Letizia: controllare i pedici !!!} 
is formally defined by the laws:
\begin{itemize}
\item $J$ and $I$ agree on $\Upsilon_{db}$ and $\Upsilon_{sys}$;
%J_{time}& = \texttt{time}(\texttt{t}+1) \text{, with }\texttt{time}(\texttt{t}) \in I_{time}\\
\item $J_{mem} = (I_{mem} \cup I^+_{ins}) \setminus I^-_{del}$, where $I^+_{ins} = Q_{ins}(I^\prime) \setminus Q_{del}(I^\prime)$ and $I^-_{del} = Q_{del}(I ^\prime) \setminus Q_{ins}(I^\prime)$;
\item $J_{out} = I_{out} \cup Q_{out}(I^\prime)$; and
\item $J_{snd} = Q_{snd}(I^\prime)$.
\end{itemize}
%while $\mathcal{P}_{sync}$ is defined as:
%
%Informally, $J$ is the updated local state.
%
%As for datalog programs, we employ the function $sch$ to derive the schema of a transducer given a program. 
%
%\subsubsection{Local Transitions}
%\label{sec:local_transitions}

\noindent Finally, note that transitions are deterministic, \ie if $I, I_{rcv}{\Rightarrow}J,$ $J_{snd}$ and $I, I_{rcv} {\Rightarrow} J^\prime, J^\prime_{snd}$, then $J = J^{\prime}$ and $J_{snd} = J^{\prime}_{snd}$. 

%\paragraph{Transducer Traces}
%\label{sec:transducer_trace}
%\input{IV/transducer_trace}

%\subsubsection{Transducer Types}
%\label{sec:transducer_types}

%The following terminology is a revisited and extended version of the one introduced in \cite{Ameloot:2013:RTD:2450142.2450151}.
Many different versions of transducers can be obtained by constraining the type of queries or the transducer schema.
A transducer is \emph{oblivious} if its queries do not use any $system$ and $time$ relations.
%A transducer is \emph{space oblivious} if its queries do not use any $system$ relation. %in its queries. 
Intuitively, this means that each query is unaware of the configuration, because independent of $(i)$ the node it is running on, $(ii)$ the other nodes in the network, and $(iii)$ the time point at which the computation is.
%Instead, a transducer is \emph{time oblivious} if its queries do not use the relation \texttt{Time}. 
%This means that each query is unaware of the time point in which the computation is.
%In general a transducer is \emph{oblivious} if it is both time and space oblivious.
A transducer is called \emph{monotonic} if all its queries are monotonic. %while we refer to a transducer as \emph{communication-free} if $Q_{snd}$ is empty.
Finally, we say that a transducer is \emph{inflationary} if $memory$ facts are never deleted -- \ie $Q_{del}$ is empty. %-- and $J_{out} = I_{out}$ $\cup$ $Q_{out}(I)$, 
%while we name a transducer \emph{semi-inflationary} if its output is inflationary but $Q_{del}$ is not empty, \ie deletion are allowed in memory.
%Finally, a transducer is \emph{non-inflationary} if $J_{mem} = Q_{ins}(I)$, and $J_{out} = Q_{out}(I)$.
%Inflationary transducer can be obtained by adding, for each $memory$ relation $R$, the following rule:
%\begin{equation}
%\begin{split}
%R_{ins}(\bar{u}) \leftarrow R(\bar{u}).
%\end{split} 
%\end{equation}
%On the other hand we name a transducer \emph{forgetting} if $J_{out} = Q_{out} / Q_{ret}$.
%Finally if both $Q_{ins}$ and $Q_{del}$ are empty, the transducer is called \emph{memoryless}. 
%A transducer is instead called \emph{dbless} if the database instance $I_{db}$ is empty, and, in general, a transducer is called \emph{stateless} if both memoryless and dbless.
%Basically a stateless transducer evaluate the queries $Q_{out}$ only on the input instance $I_{rcv}$ and is not storing any results between adjacent transitions. 

\stitle{Remark}: The relational transducer model we have just defined  is general, however it can be instantiated using a specific query language $\mathcal{L}$. We shall then write $\mathcal{L}$-transducer to denote that the program is actually implemented in $\mathcal{L}$. 
%(2) Just for the sake of readability, in the following we will use subscripts in the rule heads to indicate the transducer query the corresponding relation belongs to.

\begin{example}
\label{ex:transducer}
A first example of single-node relational transducer is the {\sc ucq}-transducer $\mathcal{T}$  below, computing an equi-join~\footnote{The readers who are not familiar with database notation should consider that  an equi-join operation between two relations  is specified  by sharing one or more variables. For example in our case the variable $v$, shared by   $R$ and $T$, indicates that the $Q$ relation is obtained by imposing that the second term of $R$  be equal to the first term of $T$.} between relations $R$ and $T$.\footnote{Note that, throughout the paper, we add a subscript to the relations in the rule heads to denote to which query -- among the queries  $Q_{ins}, Q_{del}, Q_{out}$, and $Q_{snd}$ of the transducer program -- the rule belongs.} %(which can  be  seen as a \emph{view}).}

\noindent\hrulefill
\vspace{-1mm}
\begin{equation*}
\begin{split}
&\text{Schema: }\Upsilon_{db} = \{R^{(2)}, T^{(2)}\}, \Upsilon_{mem} = \emptyset, \Upsilon_{com} = \emptyset, \Upsilon_{out} = \{Q^{(3)}\}\\
&\text{Program: } Q_{out}(u, v, w) \leftarrow R(u, v), T(v, w).
\end{split}
\end{equation*}

\noindent\hrulefill
\vspace{2mm}

\noindent Let $\mathcal{T}^i_N$ be a configured version of $\mathcal{T}$, and  \textbf{I} an initial instance over which we want to compute the join. Then, let $I_{db} = \textbf{I}$. %i.e., a node has access to the full input instance. %be an instance over $\Upsilon_{db}$. 
A transition for $\mathcal{T}^i_N$ is defined by setting $I = I_{db} \cup I_{sys}$, $I_{rcv} = J_{snd} = \emptyset$ (no communication query exists), and $J = I_{db} \cup I_{out} \cup I_{sys}$, where $I_{out}$ is the result of the query on $Q$, \ie the join between $R$ and $T$.

\end{example}
\vspace{-1mm}

\stitle{Remark}: Note that, to simplify  the notation, in the example above we omitted the schemas  $\Upsilon_{sys} $ and $\Upsilon_{time}$ because they are always the same; for the same reason, henceforth we will also omit all the empty schemas, as in the case of $\Upsilon_{mem}$ and $\Upsilon_{com}$  in the example. %-- from transducer definitions.
\vspace{1mm}

\subsection{Transducer Networks}
\label{sec:transducer_network}

%We have already defined how local computations can be expressed by introducing the notion of relational transducer. 
%\textbf{Letizia: in questa sezione ogni nodo contiene lo stesso transducer (non solo lo schema ma tutto?)}
We have already defined how local computations can be expressed by introducing the notion of relational transducer.
In this section we will model the behavior of a networked set of computing nodes by means of \emph{specifications}.
A \emph{transducer network specification} (henceforth simply \emph{transducer network}, or \emph{specification}) $\mathcal{N}$ is a tuple $(\mathcal{T}, \mathcal{T}^e, \gamma)$ where %$P$ is a partition function, %$N$ is the set of nodes composing the distributed system, 
$\mathcal{T} = (\mathcal{P}, \Upsilon)$ is a transducer, $\mathcal{T}^e$ is a transducer defining the \emph{environment}%, $\sigma$ is a \emph{state} function mapping each node $i$ to a local state $I^i$
, and $\gamma : N \rightarrow inst(\Upsilon_{com})$ is a \emph{communication} function mapping each node to a set of received facts.
%We will sometime use the function $sch$ to return the schema definition of the transducer network, \ie the schema $\Upsilon$ of the related transducer $\mathcal{T}$.
%We call the specification homogeneous because 
Transducer networks are defined such that all the nodes employ the same transducer $\mathcal{T}$, while the only thing that can be different from node to node is their state.
Such an abstraction is thus appropriate for modelling data-parallel computation, where each node applies the same set of transformations in parallel over a portion of the initial  data. \footnote{Homogeneous transducers model what is known in the parallel computing world as Single Program, Multiple Data (SPMD) computations.}

For the moment we will consider two types of specifications: \emph{broadcasting} and \emph{communication-free}. The former are specifications in which the communication function $\gamma$ is such that every fact emitted by a transducer is sent to all the other transducers composing the network. %\ie $I_{rcv}^i = \bigcup_{j \in N} J^j_{snd}$. %(but not to itself). 
%We will then use the notation $\gamma^b$, and $\mathcal{N}^b$ when we want to emphasize that the transducer network is broadcasting.
In the latter, instead, every fact is delivered just locally to the sending node. %\ie $I_{rcv}^i = J^i_{snd}$.
%We dubbed these networks as communication-free because, intuitively, no tuple is actually communicated to the other nodes.
In the remainder of this section, we will assume the network to be broadcasting. %We will use the 
%label $f$ when we want to point out that a particular specification is communication-free.
%We will use the labels $b$, $f$ when we want to point out that a particular specification is respectively broadcasting or communication-free.
%For the momemnt, however, we will assume every specification to be broadcasting.
%In general, we employ $\gamma$ basically to abstract from all the low level details about data communication, and hence $\gamma$ can be considered as enforcing the fully connected network assumption.
%In addition, we consider an input tape for each transducer. This because nodes could compute at different rates. In the next Section we will see how we can restrict the model adding the synchronous assumption.
%We are now going to explain what is the purpose of the environment.

The environment $\mathcal{T}^e$ is a ``special'' relational transducer. %collecting all the facts emitted by the transducers composing the network. %\textbf{Letizia: non ho capito: dovrebbe servire a formulate l'output complessivo verso l'utente? }
%More precisely, we define $\mathcal{T}^e = (\mathcal{P}^e, \Upsilon^e)$,  where $\Upsilon^e$ is composed only by the $memory$ and $communication$ relation schemas ($database$ and $output$ schemas are empty, while $time$ and $system$ are as usual). 
To give an intuition, following a common practice of multi-agent systems \cite{FaginHMV03}, we use the environment for modeling all the non-functional concerns related to the system; 
in our specific case: data communication and synchronization.
More precisely, we define $\mathcal{T}^e = (\mathcal{P}^e, \Upsilon^e)$ as a transducer where $\Upsilon^e$ is composed only by the $memory$ and $communication$ relation schemas, %($database$, $output$ and $time$ schemas are empty, while $system$ is as usual)
where $\Upsilon ^e_{com} = \Upsilon_{com}$, and $\Upsilon ^e_{mem}$ is the primed copy of $\Upsilon^e_{com}$ (i.e., $\forall R \in \Upsilon^e_{com}, \exists R^\prime$ s.t. $R^\prime \in \Upsilon ^e_{mem}$). 
%\footnote{We designed $\Upsilon ^e_{mem}$ relations to store all the emitted facts. 
%We set $\Upsilon ^e_{mem}$ to be the primed version of $\Upsilon ^e_{sync}$ just because schemas must be disjoined set.
The transducer program $\mathcal{P}^e$ contains, for each $R \in \Upsilon_{com}$, a set of queries %used to $(i)$ store in memory a copy of all the tuples sent by each node in the network, and $(ii)$ to emit every received fact.
%where $\Upsilon ^e_{snd} = \Upsilon_{com}$, and $\Upsilon ^e_{mem}$ is the primed version of $\Upsilon ^e_{sync}$. %\footnote{We designed $\Upsilon ^e_{mem}$ relations to store all the emitted facts. 
%We set $\Upsilon ^e_{mem}$ to be the primed version of $\Upsilon ^e_{sync}$ just because schemas must be disjoined set.
%With a little abuse of notation, though, in following sections (and in particular in Chapter \ref{chap7}) we will sometime refer to instances over $\Upsilon ^e_{mem}$ to be composed by non-primed predicates.}. 
%For what concern the transducer program, for each $R \in \Upsilon_{com}$, $\mathcal{P}^e$ is composed by a set of queries 
of the form:
\begin{equation}
\begin{split}
&R^\prime_{ins}(\bar{u}) \leftarrow R(\bar{u}).\\
&R_{snd}(\bar{u}) \leftarrow R(\bar{u}).
\end{split}
\end{equation}
the first used to store into \emph{memory} a copy of each tuple sent by each node of the network; the second  to emit every received fact.  %and therefore is composed by the queries in the following form:
%We chose the environment transducer to not accept any inputs (nodes cannot directly communicate with the environment) and to not produce any output (the environment is not permitted to communicate with other nodes).
%
The use of $\mathcal{T}^e$ will be further clarified in Section~$3.3$ when we introduce the operational semantics of transducer networks.

Given a network configuration $(N, D)$ we denote by $\mathcal{N}_{N\!,D}$ a \emph{configured transducer network}, \ie a specification where all the transducers have been configured, and where each node $i$ holds a database assigned according to the distribution function $D$. %(i)(\textbf{I})$. 
When $|N| = 1$ $D$ returns the full instance;
 %\ie $P(i)(\textbf{I}) = \textbf{I}$ 
 we call this the \emph{trivial configuration}.\footnote{Note that this follows from the previously introduced assumption on $D$ that no facts from the database instance are ignored.}
Thus, a configured transducer network can be used to model a cluster with a parallel system set up, and that is ready to run the user-submitted program;
%We admit also the case in which a network is \emph{semi-configured}: a full configuration $(N, P)$ is not available, but, for example, only the set of nodes $N$ is known.
the trivial configuration represents a program that is run in local-mode.
%In general, $N$ is said to be trivial if $|N| = 1$, while $P$ is trivial if installs the entire initial instance on every node.
%In addition, we say that $P$ is \emph{fair} if every node in $N$ contains at least a fact over every relation $R \in \Upsilon_{db}$.
%We will use the notation $N_N$ to denote a semi-configured network.

A \emph{transducer network global state} is a tuple $(I^e, I^1, \ldots, I^n)$ where, for each $j \in N$ $\cup$ $\{e\}$, the $j$-th element is the related relational transducer state $I^j$. 
The definitions of oblivious, monotonic, and inflationary transducers provided at the end of Section \ref{sec:relational_transducer} are naturally generalized over transducer networks.
Now, given a specification $\mathcal{N}$, %similarly to what we did locally by introducing the concept of trace, 
many possible executions may exist, each of them representing one possible evolution of the global state. 
We describe how network global states may change over time through the notion of \textit{run} $\rho$, which binds \emph{logical time} values to global states.\footnote{In this paper we will consider logical time and \emph{physical time} to be two different entities: the former is used to reason about the computation progress of a distributed system; the second can be thought as the time in seconds returned by a local call to the operating system.} %If we assume time to be isomorphic to the set of natural numbers, a run $\rho$ is a mapping between $\mathbb{N}_0$ and the set of all possible global states. 
%: \mathbb{N}_0 \rightarrow \mathbb{G}$ where $\mathbb{G} = \{\mathbb{S}^e\times \mathbb S^1\times ... \times \mathbb S^n\}$, and $\mathbb S^i$ is the set of possible local states for node $i \in N \cup e$. %and $\mathbb S^e$ is the set of possible states for the environment.
%Following this formalism, for example, we can express the initial global state of the system by the run $r(0)$. As a convention, we use $r_i(t)$ to refer to the local state $s_i$ of the node $i \in L$ (the same can be stated for $s_m$).
%In the following, we will use the term point to denote both transducer point and run point.
%Which of the two definitions a point actually refer to will be clear from the context.
Then, if $\rho(t) = (I^e, I^1, ..., I^n)$ is the network global state at time $t$, a \emph{point} $(\rho^i, t)$ is the transducer state
%$\rho^i(t) = I^i$ 
of node $i \in N$. We assume that the initial global state $\rho(0)$ is such that: $(i)$  the $database$ local to each node contains the related partition of the initial instance; $(ii)$ the local $system$ relations are properly initialized with the node identifier and with the information about the other nodes; and $(iii)$ all the other relations are empty. %while $I^e$ is empty (except the $system$ relations).

Recall that a distributed system may have many possible runs, indicating all the possible ways the global state can evolve. In order to capture this, %assumed that an initial instance \textbf{I} is given, 
starting from a configured transducer network $\mathcal{N}_{N,D}$ and an initial instance $\textbf{I} \in inst(\Upsilon_{db})$, we define a \textit{system} $\mathcal{S}_{\mathcal{N}}(N, D, \textbf{I})$ as a set of runs, where $N, D$ and $\textbf{I}$ are the parameters shaping the system. 
Given a system $\mathcal{S}_{\mathcal{N}}(N, D, \textbf{I})$,  if its settings are irrelevant or clear from the domain we will often denote it simply by $\mathcal{S}$.

\eat{\stitle{Remark}: The above definitions permit us to give a declarative characterization of a distributed system. 
In other words, we don't need to specify in detail how nodes actually interact, but just which is the expected behavior of the network.
The notion of system serves  exactly this purpose: it encompasses all the possible executions of a distributed specification.  
\vspace{1mm}
}
%not just as a collection of interacting nodes, but, in a declarative fashion, we are directly modeling its behavior. %by a network specification, abstracting away many low level details
%We think that this approach is particularly important to the aim of maintaing a high level of declarativity  in our description.
%In particular, a system behavior is completely defined by a configured network $\mathcal{N}_{N, P}$, where the network $\mathcal{N}$ defines the actual system \emph{specification}. %while the configuration $(N, P)$ provide the \emph{network configuration}.

In the following we will also be interested in investigating \emph{classes} of systems, \ie sets of systems having identical specification but different configurations. %, or having distinct partition functions, or also different initial instances. 
Thus, if a system is defined starting from a configured transducer network and an instance, %$(\mathcal{N}_{N, P}, \textbf{I})$, %-- or \emph{fully instantiated specification} -- 
a class of systems is defined starting from a simple specification $\mathcal{N}$, by adding an instance and a \emph{partial configuration}: \ie a configuration having some unfixed parameter.
%Given a specification, for example, we can use a %the signature $\mathbfcal{S}_{\mathcal{N}, \textbf{I}}$ to define 
%class of systems to express when no constraint is imposed on the configuration, or when the partition function is bound.
Intui\-tively, if all the parameters are bound we obtain a specific system $\mathcal{S}_{\mathcal{N}}(N,D, \textbf{I})$.
Partial configurations and classes of systems are important because, in the next sections, we will study with particular attention which (instantiated) specifications are able to obtain a unique final outcome, \emph{independently of the provided configuration}.
%Note that hereafter we will often use the term \emph{specification} to denote a transducer network, \emph{fully instantiated specification} to identify a network instance and \emph{partially instantiated specification} to represent a semi-configured network instance.
%We will instead use the term \emph{full} (\emph{partial}) \emph{specification} if the network is (semi) configured, and a specific initial instance is not given.

\vspace{-1ex}
\subsection{Synchronous and Reliable Systems}
\label{sec:synchronous_systems}
We are mainly interested in synchronous systems with reliable communication. 
Informally, a distributed system is synchronous when all processing node's clocks run at the same rate and both the difference between two nodes clocks and the latency of data communication is bounded. One can construct a synchronous system by providing as input to each node an external reference \emph{global clock}, and assuming that difference between the global clock received as input by each node is bounded~\footnote{In fact, synchronization cannot be achieved if by the time the global clock reaches a node, the reference clock has changed significantly.}.
A similar bound on data communication also have to exists, otherwise nodes will be unable to reason about distributed property of the system. The next definition summarizes the above considerations.

\vspace{-0.5ex}
\begin{definition}\label{def:synchronous}
A \emph{synchronous system} $\mathcal{S}^{\emph{sync}}$ is a set of runs fulfilling the following conditions:

\begin{description}
\vspace{-0.5ex}
\item[\textbf{S1}] A global clock is defined and accessible by every node;
\item[\textbf{S2}] The relative difference between the time clock values of any two nodes is bounded; and
\item[\textbf{S3}] Emitted tuples arrive at destination at most after a certain bounded physical time $\Delta$.
\end{description}
\vspace{-1ex}
\end{definition}

\noindent In our framework, the first property can be expressed by linking the time value stored in the \texttt{Time} relation of  each node with the external logical time used to reason about system runs.
 This is accomplished by defining a \emph{timed local transition} $I$, $I_{rcv} \overset{t}{\Rightarrow} J$, $J_{snd}$ as a local transition where, at each time instant $t$, $I_{time} = \{\texttt{Time}(t)\}$.  
In this enriched setting, we have that each transducer accepts as input also the clock value, %\emph{driving the computation}. %This is exactly what we already did by defining a transducer to accept as input a global clock provided by an external tape.
%thus, by definition, %in a synchronous system we encode fact's time suffix directly in local states. However, for convenience, we don't express the time suffix directly on fact, but, instead, we made it explicit. We then specify the local state $s_i$ as the tuple $(\mathcal{I}_i, m)$ with $m$ the time suffix value for facts in $\mathcal{I}_i$. Then, 
%every local state $\rho^i(m) = I^i$, is such that $m = s$ and $\texttt{time}(s) \subset I^i$. %Then, given a couple of points $(r, t)$ and $(r', t')$, if $(r, t) \sim_i (r', t')$ for every node $i \in L$, then $t = t'$ \cite{Fagin:2003:RK:995831}. %In other words, we are assuming that the time is common knowledge among the nodes of the system. 
%Thus, by definition, %in a synchronous system we encode fact's time suffix directly in local states. However, for convenience, we don't express the time suffix directly on fact, but, instead, we made it explicit. We then specify the local state $s_i$ as the tuple $(\mathcal{I}_i, m)$ with $m$ the time suffix value for facts in $\mathcal{I}_i$. Then, 
%every local state $\rho^i(t) = (\mathcal{I}^i, n)$ will have now $t = n$. %Then, given a couple of points $(r, t)$ and $(r', t')$, if $(r, t) \sim_i (r', t')$ for every node $i \in L$, then $t = t'$ \cite{Fagin:2003:RK:995831}. %In other words, we are assuming that the time is common knowledge among the nodes of the system.
and the environment can be employed to directly provide the clock driving the computation for all the transducers~\footnote{We assume that the environment is the only transducer allowed to modify the time-related instances, which are then sent to the rest of the network to achieve synchronization.}.
%Under this perspective, each timed local transition is labeled with the time in which it is performed. 
\eat{To accomplish this, we have that $\Upsilon^e_{mem}$, $\Upsilon^e_{com}$ now respectively contain the relations \texttt{Time}, and $\texttt{Stime}$ (where \texttt{Stime} is for synchronous time), while $\mathcal{P}^e$ is augmented with the following clauses:
\begin{align}
\label{eq:time_f1}
&\texttt{Time}(0).\\
\label{eq:time_d1}
&\texttt{Time}_{ins}(s) \leftarrow \texttt{Time}(t), \texttt{Succ}(t, s).
\end{align}
\begin{align}
\label{eq:time_d2}
&\texttt{Time}_{del}(t) \leftarrow \texttt{Time}(t).\\
\label{eq:time_d3}
&\texttt{Stime}_{snd}(t) \leftarrow \texttt{Time}(t).
\end{align}
where eq. (\ref{eq:time_f1}) sets the initial clock value, (\ref{eq:time_d1})~\footnote{The built-in relation \texttt{Succ} has $\mathbb{N}_0$ as domain for both terms, and is used to compute the successor function: \ie the interpretation of $\texttt{Succ}(t, s)$ is $s = t + 1$.} and (\ref{eq:time_d2}) are used to move forward the clock tick, and eq. (\ref{eq:time_d3}) emits the new clock value~\footnote{We assume that the environment is the only transducer allowed to modify the time-related instances, which are then sent to the rest of the network to achieve synchronization.}.  % and (\ref{eq:time_f1}) inject into the program the starting value.
Finally $\mathcal{P}$ is augmented with the rule:
\begin{align}
\label{eq:trans_time_d1}
&\texttt{Time}(t) \leftarrow \texttt{Stime}(t).
\end{align}
}
Under this perspective, each timed local transition is basically labeled with the time value in which it is performed.

The second property of Definition~\ref{def:synchronous} can be added to our framework by assuming that programs proceed in \textit{rounds}, and that each round, operationally speaking, lasts enough to permit the  computation  at each node to reach the fixpoint~\footnote{Note that, because of the considered languages, we are assured that local transitions always terminate in at most PTIME.}.
In the following, \kwlog we will use the round number to refer to the time of the clock stored in the \texttt{Time} relation.
Finally, in order to express the third property, we assume that emitted facts are first buffered locally, and then, once the node has completed its local transition, delivered in batch to destination.
%Sent tuples will be received exactly at $\Delta$ physical time after emission, and will be available for being used starting from the successive round.
%We can therefore infer that $\Delta$ is composed of two quantities: the \emph{network latency} $\delta$, and %an environment computation latency $\varepsilon$.
%Remark, in fact, that every emitted tuple is first collected into the environment before actually been delivered.
%In addition, we also consider $\varepsilon$ to contain 
%a \emph{synchronization delay} that expresses the time required by the system to detect that every node has reached the fixpoint.
%By construction, the network latency is considered twice in the definition of $\Delta$, once from the sending node to the environment, and once from the environment to the receiving node.
%We use the tuple $(N, P)$ to specify a configuration for a synchronous specification, where $N$ is the set of nodes, %$t$ is the initial round value %-- \ie $I_{time}^e$ initially contains the tuple $\texttt{time}(t)$ --  
%and $P$ a partition function. 
%$\mathcal{N}_{N\!, P}$ then specifies a configured bulk synchronous transducer network. %is now indicated by the signature $\mathcal{N}_{N, t, P}$. %and, similarly, a synchronous system is denote by $\mathcal{S}^{rsync}_{\mathcal{N}_{N, t, P}, \textbf{I}}$. %where the tuple $(N, t, P)$ define its context, and \textbf{I} is a given initial instance.
%From here we always use the term transducer network to actually denote a bulk synchronous transducer network.

The above properties imply that a remote tuple is either received after a bounded amount of time or never received. Recall that initially, however, we have required each system to be reliable, \ie all emitted tuples arrive at destination:

\begin{definition}
To be reliable, a synchronous system must satisfy properties \textbf{S1} - \textbf{S3} along with the following additional conditions:
\begin{description}
\item[\textbf{R1}] In every run, for every received fact for node $i$ at round $t^\prime$, there exists a node $j$ and a round $t$ s.t. $t < t^\prime$ and a send query derived the same fact on node $j$ in round $t$; and

\item[\textbf{R2}] In every run, if a fact has been emitted by a node $i$ at round $t$, there exists a time $t^\prime$ s. t. $t < t^\prime$ and the same fact belongs to the input instance state of a node $j$ at round $t^\prime$.
\end{description}

\end{definition}

\noindent Informally, properties \textbf{R1} and \textbf{R2} specify that if an emitted fact exists in an instance at a given round, then it has been generated by a send query in a previous round, and vice-versa, if a send query derives a new fact, then that fact must appear in a successive round in the local state of a node. 
%Indeed we have that the communication rule syntax -- a distributed fact inferred from a communication rule is true in the successive time-step --  satisfies the union of the properties \textbf{S1} and \textbf{R1}-\textbf{R2}. 
We denote by $\mathcal{S}^{\emph{rsync}}$ the systems satisfying conditions \textbf{S1} - \textbf{S3} and \textbf{R1} - \textbf{R2}. 

To further simplify the model, we can add to \textbf{S3} an extra condition forcing emitted tuples not only to be eventually received after at most $\Delta$ physical time, but to be delivered \emph{exactly} after $\Delta$:
\begin{description}

\item[\textbf{S3}$^{\prime}$] Every tuple sent at round $t$ is delivered exactly after $\Delta$ physical time.

\end{description} 
We name this condition \emph{deterministic delivery}: without \textbf{S3}$^{\prime}$, tuples may be non-deterministically delivered inside the bounded range defined by $\Delta$ and, as we will see in Section~\ref{sec:async_del}, this situation creates different coordination patterns.
We can then assume that, between the end of one round and the start of the consecutive one, precisely $\Delta$ physical time elapses, so that all emitted tuples are received at the start of the new round. 
In other words, we can safely shift the beginning of each new round $t+1$ so that every tuple emitted at round $t$ is precisely delivered at the start of the new round.
%$\theta(t^\prime)$ when emitted at round $t$, where $t^\prime = \theta^{-1}(\theta(t+1) + var)$.
%Naturally, this is possible only because we are assuming a bounded variance.
%\kWlog, for simplicity of notation in the following we will write $t+1$ instead of $t^\prime$. 
%The fixed delivery property can be thus directly embedded into the system definition. 
We will use the signature $\mathcal{S}^{\emph{bsp}}$ to denote a synchronous and reliable system with deterministic delivery. %is denoted by the signature $.
%
%Except for Section \ref{sec:sync_calm}, we will then only consider \emph{rsync} systems. 
Figure~\ref{fig:bsp} gives a pictorial representation of \emph{bsp} systems. Note that this type of systems simulate how real-world BSP frameworks behave.

\begin{figure}[t]
\centering
\includegraphics[width=0.45\columnwidth]{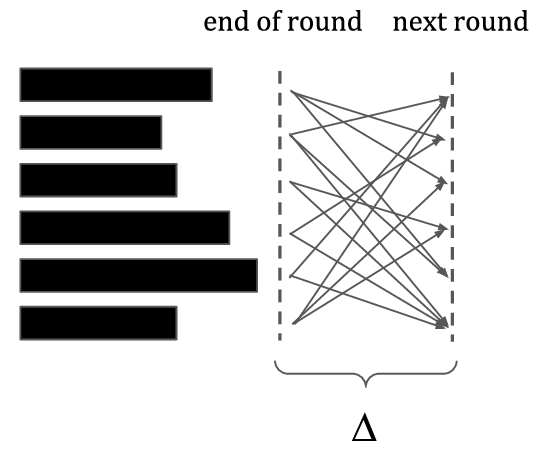}
\vspace{-5mm}
\caption{\emph{bsp} system computation model. Each node's computation is independent but bounded by the global end of the round (condition \textbf{S2}). Communication is reliable and takes exactly $\Delta$ time (conditions \textbf{S3} and \textbf{S3}$^\prime$). The next round starts when $\Delta$ time since the end of the previous round has elapsed  (conditions \textbf{S1} and \textbf{S3}$^\prime$).}
\label{fig:bsp}
\vspace{-1ex}
\end{figure}

Next we will show how the properties for \emph{bsp} systems are enforced during global transitions.

\eat{\stitle{Remark}: Transitions between global states must enforce all the above properties in order to make a proper \emph{rsync} system.
More precisely, during a global transition all the nodes composing the network simultaneously perform a local transition taking as input the associated tuples. %output of the $s-1$ round. 
A local transition for the environment is then executed, whose input is the set of tuples emitted by all the transducers.
%Recall that if the network is \emph{broadcasting}, then $\gamma^b(i)(\textbf{F}) = J^e_{snd}$ while, if it is \emph{communication free}, $\gamma^f(i)(\textbf{F}) = J^i_{snd}$, with $J^i_{snd} \subseteq J^e_{snd}$.
%Such input tuple are defined by the $\beta$ function. %The network then will output the set of external actions, denoting as $sdb$ the fact where the location specifier is the environment. 
%Basically $\beta$ track for each round and for each node $i$, the set of facts that must be delivered to $i$. Such function is used in order to achieve a sort of \emph{topology independence}, in the sense that is $\beta$ task to ensure that each node will receive the fact addressed to it at round $t-1$. 
%In addition we assume that one global transition, in order to satisfy property \textbf{S3}, can start only after that a certain amount $\Delta$ of physical time has elapsed after the end of the previous transition.
 %and the function $\beta$ is properly updated. 
Clearly, a global transition, in order to satisfy property \textbf{S3}, can start only when a certain amount $\Delta$ of physical time has elapsed after the end of the previous transition.
%Note that even though the environment is emitting a \texttt{stime} tuple containing the clock, in order to distinguish the emission of time values from the data tuples, we employed a different signature for the two.
As a final remark note that, since a global transition is composed by $|N|$ deterministic local transitions, and the communication is assumed to be reliable, global transitions are also deterministic.
From this follows that any \emph{rsync} system is composed by \emph{just one run}. %has a number of runs which is proportional to the number of achievable executions from the set of all possible partitioning of all the possible input instances.
%As can be clearly noticed, the number of runs composing a synchronous and reliable system is huge.
%However, as we will see, special classes of networks exists, for which all the runs in $\mathcal{S}^{rsync}$ or even $\mathbfcal{S}^{rsync}$ \emph{converge} to one single final state.
%Clearly, if no instance is provided, the number of possible runs is infinite (to see this, note that the number of instance of a given schema is infinite assumed an infinite domain).
%If instead the input instance is given, the number of possible runs 
%We will use the term \emph{round} to denote a certain transition number, so, for example, we could say that a computation is terminated after $n$ rounds to denote that after $n$ transitions the quiescence point is reached.
}

\vspace{-2ex}
\subsection{Global Transitions}
\label{sec:global_transitions}
%Similarly to what we did before, a \emph{global configuration} $\Theta$ is a function mapping a transducer network state \textbf{I} to the related global state $g = \Theta(\textbf{I}) = (\theta(I^e), \theta(I^1), \ldots, \theta(I^n))$.
%Let a transducer network initial state be a state where except \textbf{R}$_{db}$, \texttt{id}, and \texttt{all}, all the remaining relations are empty. 
%\mi{Attenzione, nel primo item tornava la notazione che abbiamo tolto precedentemente $P(i)(\textbf{I})$, con lo stesso problema formale. Ora provo a riscrivere l'item togliendola, lascio il vecchio  item commentato sotto. Guarda che ho anche spostato alcune cose sopra, mi sembra che non facciano parte della prima condizione}

Given a transducer network $(\mathcal{T}, \mathcal{T}^e, \gamma)$ and two global states $\textbf{F} = (I^e, I^1, ..., I^n)$, $\textbf{G} = (J^e, J^1, ..., J^n)$, let $t$ be the clock value and 
%$I^i_{db} = P(i)(\textbf{I})$, 
$\sigma$ be a \emph{state} function mapping each node $i$  and global state $\textbf{F}$ to the corresponding local state $I^i$; \ie $\forall i \in N$, $I^i = \sigma(\textbf{F})(i)$. Moreover, let $I_{rcv}^i = \gamma(i)$, and $I^i_{db} = D(\textbf{I})(i)$. %an initial network state $(\mathcal{N}_{N, t, P}, \textbf{I})$, %an initial instance, 
%and a partitioning function $P$.
A \emph{global transition} for a \emph{bsp} system, denoted by $\textbf{F}$ ${\Rightarrow}$ $\textbf{G}$, %, and $t \in \mathbb{N}_0$ is denoting the input from the time tape specifying the global clock. A global transition is defined 
is such that the following conditions hold:

\begin{itemize}
%\item $T_{in} = T^e_{net}$
\item  $(I^i$, $I_{rcv}^i$ $\overset{t}{\Rightarrow} J^i$, $J^i_{snd})$ is a timed local transition for transducer $\mathcal{T}^i_{N}$;
%\item Let $t$ be the clock value, $I^i_{db} = P(i)(\textbf{I})$, and $\sigma$ be a \emph{state} function mapping each node $i$ to a local state $I^i$. $\forall i \in N$, $I^i = \sigma(i)(\textbf{F})$, $I_{rcv}^i = \gamma(i)(\textbf{F})$, and $(I^i$, $I_{rcv}^i$ $\overset{t}{\Rightarrow} J^i$, $J^i_{snd})$ is a timed local transition for transducer $\mathcal{T}^i_{N}$; 
%\setminus J^+_{snd}$ where $J^+_{snd} = J^i_{snd} - J^e_{snd}$, 
  %\ie $\texttt{stime}(s) \in J^e_{snd}$, and %, where $S$ is the distribute instance inferred from $I$ %, where $T_{in}^i$ denotes the set of facts in $T_{in}$ for node $i \in L$, \ie $ T^I_{rcv} = \bigcup _{j \in L - i} \{R(k,\bar{u}) | R(k,\bar{u}) \in T_{in} \wedge j = k\}$
%\item $N^\prime_{out} = \bigcup_{i \in L} {\textbf{T}}^{\prime i, e}_{out}$ where ${\textbf{T}}^{\prime i, e} = \{R(\bar{u})| R(\bar{u}) \in {T}^{\prime i}_{out}$, $R \in \textbf{R}_{dsr}^i$ and $\bar{u}(1) = e \}$
%\item $\beta(i) = \bigcup_{j \in L - i} T_{out}^j$
%\vspace{-1mm}
\item ($I^e$, $I^e_{rcv}$ ${\Rightarrow} J^{e}$, $J^e_{snd})$ is a local transition for the environment, where $I^e_{rcv} =  \bigcup_{i \in N} {J}^{i}_{snd}$. %and $t$ is the initial time value.
\end{itemize}
%
%Initially we have a bootstrap super-step in which t is empty so no transducer transit execpt the environment.
%where we call the operation of annotating each tuple ${T}^{\prime i}_{out}$ with the location of the issuing node as \emph{location annotation}, and for the purpose we define a related $\$$ operators, which, given an instance $I$, creates its location annotated version $I\$i = \{R(i, \bar{u}) | R(\bar{u}) \in I\}$.
Informally, during a global transition all the nodes composing the network make simultaneously a local transition taking as input the associated tuples. %output of the $s-1$ round. 
A local transition for the environment is then executed, whose input is the set of tuples emitted by all the transducers. %(\ie $I^e_{rcv} =  \bigcup_{i \in N} {J}^{i}_{snd}$).
%Recall that if the network is \emph{broadcasting}, then $\gamma^b(i)(\textbf{F}) = J^e_{snd}$ while, if it is \emph{communication free}, $\gamma^f(i)(\textbf{F}) = J^i_{snd}$, with $J^i_{snd} \subseteq J^e_{snd}$.
%Such input tuple are defined by the $\beta$ function. %The network then will output the set of external actions, denoting as $sdb$ the fact where the location specifier is the environment. 
%Basically $\beta$ track for each round and for each node $i$, the set of facts that must be delivered to $i$. Such function is used in order to achieve a sort of \emph{topology independence}, in the sense that is $\beta$ task to ensure that each node will receive the fact addressed to it at round $t-1$. 
%In addition we assume that one global transition, in order to satisfy property \textbf{S3}, can start only after that a certain amount $\Delta$ of physical time has elapsed after the end of the previous transition.
 %and the function $\beta$ is properly updated. 
In addition, in order to satisfy property \textbf{S3}$^\prime$, we assume that a global transition can start only when a certain amount $\Delta$ of physical time has elapsed after the end of the previous transition.
%Note that even though the environment is emitting a \texttt{stime} tuple containing the clock, in order to distinguish the emission of time values from the data tuples, we employed a different signature for the two.
As a final remark note that, since a global transition is composed by $|N|$ deterministic local transitions, and the communication is assumed to be reliable, also global transitions are deterministic.
From this follows that any \emph{bsp} system $\mathcal{S}^{\emph{bsp}}$ is defined by \emph{just one run}. 
We refer to the specifications defining \emph{bsp} systems as \emph{synchronous}.
%has a number of runs which is proportional to the number of achievable executions from the set of all possible partitioning of all the possible input instances.
%As can be clearly noticed, the number of runs composing a synchronous and reliable system is huge.
%However, as we will see, special classes of networks exists, for which all the runs in $\mathcal{S}^{rsync}$ or even $\mathbfcal{S}^{rsync}$ \emph{converge} to one single final state.
%Clearly, if no instance is provided, the number of possible runs is infinite (to see this, note that the number of instance of a given schema is infinite assumed an infinite domain).
%If instead the input instance is given, the number of possible runs 
%We will use the term \emph{round} to denote a certain transition number, so, for example, we could say that a computation is terminated after $n$ rounds to denote that after $n$ transitions the quiescence point is reached.

%\subsection{Practical Implementations and Discussion}
%\label{sec:multicore}
%
%\input{multicore}
%
%%\subsection{Example: Distributed Processing}
%%\label{sec:distributed}
%\input{distributed}
%
%%\subsection{Discussion}
%%\label{sec:relaxing_simultaneously}
%\input{relaxing_simultaneity}

\vspace{-1ex}
\subsection{Query Computability}
\label{sec:sync_computability}
%The results illustrated in this section directly follow from \cite{Ameloot:2013:RTD:2450142.2450151}.
Given a run $\rho$ describing the execution of a synchronous transducer network, we use the notation $out(t)$ for the set of facts output by all nodes at time $t$, \ie $out(t) = \bigcup _{i \in N}I^i_{out}$ such that $I^i_{out} \in (\rho^i, t)$.
This definition models how parallel data-processing frameworks work in practice: the output remains distributed on each node composing a cluster and can be eventually collected by invoking a proper function, or written to a distributed file system (\eg HDFS).  %and no notion of which output belongs to which node exists. 
%Intuitively, in some point in time $t$, $out(t)$ can be the empty set. In such case we write $out(t) = \bot$.
%If $out(t) = \bot$, $\forall t \in \mathbb{N}_0$, we say that the synchronous network is \emph{not terminating}.
Now, let us assume that for a synchronous transducer network a time value $t^\prime$ exists such that $\forall t^{\prime \prime} > t^\prime, out(t^\prime) = out(t^{\prime\prime})$; that is, a quiescence state is reachable so that the output is stable and not changing any more.
%A synchronous transducer network is \emph{quiescent} if a time value $t^\prime$ exists such that $\forall t^{\prime \prime} > t^\prime, out(t^\prime) = out(t^{\prime\prime})$, that is, a quiescence state is reachable so that the output is stable and not changing any more. 
%We name $(\rho, s^\prime)$ the \emph{quiescence point}.
We define the output for a synchronous network to be the output up to network quiescence, denoted by $out(*)$, where we use $*$ to denote the time value in which the quiescence state is reached.
%If a value of $t$ exists such that $out(t) \neq \bot$, the transducer will certain \emph{terminate}, \ie a \emph{quiescence state} will be reached where no more facts will be added to the output.
%This because we consider only finite input instance, and therefore the set of output facts can only be finite, and once a quiescence state will be reached, it will be maintained forever since quiescence is a \emph{stable property}, \ie once true, it always remains true.
%The quiescence point of a synchronous network is the time value $t^\prime$ such that $\forall t^{\prime \prime} > t^\prime, out(t^\prime) = out(t^{\prime\prime})$.
%If the synchronous network is not quiescent, we say that the output instance is $\bot$.
%As a signature, we use $\mathcal{N}_{N\!,t\!,P}(\textbf{I})$ to denote the output of the transducer network $\mathcal{N}_{N\!, t\!, P}$ on input database instance $\textbf{I}$.
In practice, a transducer network initial state identifies also its output. 
This is because the system $\mathcal{S}^{\emph{bsp}}_{\mathcal{N}}(N, D, \textbf{I})$ is constituted by one and only one run, therefore, given an initial instance and a configuration, the output is uniquely determined.
As a signature, we use $\mathcal{N}_{N\!, D}(\textbf{I})$ to denote the output of the transducer network $\mathcal{N}_{N\!, D}$ on input database instance $\textbf{I}$ , and over a \emph{bsp} system. 

We are now able to state what we mean for a query to be computable by a transducer network $\mathcal{N}$: 
%a configuration $(N, D)$ exists such that the output of the configured transducer network, at network quiescence, is the output of the query.
\begin{definition}
\label{query_computable}
Given an input and an output schema, respectively $\mathcal{D}_{in}$ and $\mathcal{D}_{out}$, a total mapping $\mathcal{Q}: inst(\mathcal{D}_{in}) \rightarrow inst(\mathcal{D}_{out})$ is \emph{computable by a synchronous transducer network} if a configured transducer network $\mathcal{N}_{N\!, D}$ exists such that $\mathcal{D}_{in} = \Upsilon_{db}$, $\mathcal{D}_{out} = \Upsilon_{out}$ and $\mathcal{N}_{N\!, D}(\textbf{I}) = \mathcal{Q}(\textbf{I})$, for every initial database instance $\textbf{I}$ over $\mathcal{D}_{in}$.
%\begin{description}
%\vspace{-1mm}
%\item[C1] if $\mathcal{Q}_{out}(\textbf{I})$ is undefined, then $\mathcal{N}_{N\!, t\!, P}(\textbf{I}) = \bot$
%\vspace{-1mm}
%\item[C2] if $\mathcal{Q}_{out}(\textbf{I})$ is defined, then $\mathcal{N}_{N\!, t\!, P}(\textbf{I}) = \mathcal{Q}_{out}(\textbf{I})$
%\vspace{-1mm}
%\end{description}
\end{definition}

%Note that inflationary specifications always reach a quiescence point, while this is not the case for a non-inflationary one. 
\eat{
\noindent Because we assumed generic queries as building blocks of transducers, we have that the function $\mathcal{N}_{N\!, P}$ is generic for each oblivious synchronous transducer network.
%\end{pr}
}

\noindent Because we assumed generic queries as building blocks of transducers, we have that:
\begin{proposition}
The function $\mathcal{N}_{N, D}$ is generic for each oblivious synchronous transducer network.
\end{proposition}

Some specification might have the property that a unique final result is obtained independently of the chosen configuration; on the other hand, not all specifications have this property. Note that this property is
%Some transducer network may depend on the number of nodes available in the distributed system or on a particular partition function.
%Conversely, some specification may be able to obtain a unique final result independently of the chosen configuration.
%This is 
a common requirement in parallel systems: the same job must return consistent results, whichever the cluster size and partitioning scheme. %is used.
%In order to formally define this, we will use the notion of \emph{eventually consistenc}:

\begin{definition}
\label{df:convergence}
Given a specification $\mathcal{N}$, an input instance $\textbf{I} \in inst(\Upsilon_{db})$, and a partial configuration $\psi$, we say that the class $\mathbfcal{S}^{\emph{bsp}}_{\mathcal{N}}(\psi, \textbf{I})$ is \emph{convergent} if for all pairs $\rho$, $\rho^\prime \in \mathbfcal{S}^{\emph{bsp}}_{\mathcal{N}}(\psi, \textbf{I})$, the respective final outputs $out(*)$, $out^\prime(*)$ coincide, \ie $out(*) = out^\prime(*)$. %A specification is convergent if its class is convergent for all input instances.
\end{definition}

\noindent Informally, a class of systems is convergent if, at the quiescence state, all its runs have the same output.
Note that -- in order for the convergence property to hold -- each pair of runs are required to agree just on the initial instance \textbf{I} and on the final output state, and not necessarily on the entire execution.
Again, this means that, whichever configuration we select, we are assured that the same final outcome will be eventually returned.
%\footnote{ 
%
% per sempio Example 4.6 in ameloot
% ma questo Ž vero anche per space-oblivious transducers. Per esempio una network in cui ogni transducer butta fuori tutti i suoi fatti e butta come output ok() se riceve almeno un fatto. N(I) sar‡ ok se la network Ž composta da almeno 2 transducer, quindi non Ž network indipendent. 
We will then say that a specification $\mathcal{N}$ is \emph{network-independent} if, once fixed a distribution function $D$, the class $\mathbfcal{S}^{\emph{bsp}}_{\mathcal{N}}(D, \textbf{I})$ is convergent for all possible $\textbf{I} \in inst(\Upsilon_{db})$, or, in other words, $\mathcal{N}_{N\!, D}$ computes the same query $\mathcal{Q}$ for all networks $N$ and instances \textbf{I}.~\footnote{Recall that we only consider fully connected networks, i.e., networks are only distinguished by their cardinality. This too is a common assumption in modern data-parallel distributed systems.}
A similar definition applies for \emph{distribution-independent} specifications. %fixed a set of nodes $N$, and a partition function $P$, we say that a specification $\mathcal{N}$ is \emph{time independent} if $\mathbfcal{S}^{det}_{\mathcal{N}_{N, P}, \textbf{I}}$ is consistent, while, fixed $t$ and $P$, $\mathcal{N}$ is \emph{partition independent} if instead $\mathbfcal{S}^{det}_{\mathcal{N}_{N, t}, \textbf{I}}$ is consistent.
%This is all the runs in $\mathcal{S}^\mathcal{N}_N$ converge to the final output state $out(*)$.
Finally, if a specification $\mathcal{N}$ is network-distribution-independent, the class $\mathbfcal{S}^{\emph{bsp}}_{\mathcal{N}}$ is convergent, \ie for any instance \textbf{I}, all the possible runs in $\mathbfcal{S}^{\emph{bsp}}_{\mathcal{N}}(\textbf{I})$ compute the same query result $\mathcal{Q}(\textbf{I})$. This is because, whichever configuration is selected, $\mathbfcal{S}^{\emph{bsp}}_{\mathcal{N}}(\textbf{I})$ has a unique output. %whichever initial instance. %has been given from the set $inst(edb(\mathcal{Q}_{out}))$. 
In this case, $\mathcal{Q}$ is said to be \emph{distributively computable}, while $\mathcal{N}$ is said to be \emph{independent} (or \emph{convergent}).
Note that similar definitions have been used also in \cite{AmelootNB13}, where it was also shown that independence is an undecidable property of specifications~\footnote{Note that we could have added the initial round number to the network configuration parameters. Instead we chose to fix to 0 the initial round value not to make our model too complex. Because of this, we implicitly assume each query to be \emph{time-independent}: the value stored in the \texttt{Time} relation does not influence the result of the query. A (conservative) syntactical condition to achieve time-independence is to not have the \texttt{Time} relation in any query-rule.}.
%For completeness, we report in Appendix~\ref{sc:independency} a study on the decidability of independence under our model. 
%We will sometimes write $\mathcal{N}(\textbf{I})$ to denote the result of the specification distributively computing a query $\mathcal{Q}_{out}$ for the instance \textbf{I}. %where naturally $\mathcal{N}(\textbf{I}) = \mathcal{Q}_{out}(\textbf{I})$. %and all the runs in the class $\mathbfcal{S}$ converge to $out(*)$.

\begin{example}
\label{ex:net_independent}
Assume we want to compute a distributed version of the program of Example \ref{ex:transducer}.
This can be implemented using a %memory-less,
broadcasting and inflationary synchronous transducer network %\footnote{Recall that a bulk synchronous transducer network is inflationary if each relational transducer composing the network is inflationary.} 
in which every node emits one of the two relations, let us say $T$, and then joins $R$ with the facts over $T$ received from the other nodes.
Note that the sent facts will be used just starting from the successive round.
This program will then employ two rounds to compute the distributed join. 
{\sc ucq} is again expressive enough.
The transducer can be written as follows:

%\vspace{5mm}
\noindent\hrulefill
\vspace{-1mm}
\begin{equation*}
\begin{split}
\text{Schema: }&\Upsilon_{db} = \{R^{(2)}, T^{(2)}\}, \Upsilon_{com} = \{S^{(2)}\}, \Upsilon_{out} = \{Q^{(3)}\}\\
\text{Program: }&
S_{snd}(u, v) \leftarrow T(u, v).\\
%&T^\prime(\texttt{u}, \texttt{v}) \leftarrow T(\texttt{u}, \texttt{v}).\\
%&T^\prime(\texttt{u}, \texttt{v}) \leftarrow S(\texttt{u}, \texttt{v}).\\
&Q_{out}(u, v, w) \leftarrow R(u, v), S(u, w).
\end{split}
\end{equation*}

\vspace{-1mm}

\noindent\hrulefill
%\vspace{2mm}

\noindent The specification is clearly independent since the same output is obtained whichever configuration is selected. 
\end{example}
%\vspace{-2mm}
%
\eat{\begin{ex}
\label{ex:net_dependent}
As example for a \emph{network} and \emph{partition dependent} specification, let us consider a schema composed by a single database relation $R^{(1)}$ and a synchronous relation \texttt{ready}$^{(0)}$. A transducer will broadcast its id if its local view over $R$ is not empty. When at least 3 \texttt{ready} facts from different nodes are received, each transducer output its fragment of $R$. 
Clearly, if the network is composed by one or two nodes, the transducer will not output anything. 
However, if the network is composed by three or more nodes, depending on the partition, the specification may have different outputs.
For instance, if \textbf{I} is nonempty $|N| = 3$, and $P$ is such that the entire instance is installed on every node, $I_{out}$ will end up having all the tuples stored in $R$.
On the other side, if $P$ is defined in a way that only the first node contains the entire instance, while all the other transducers are empty, $T$ will be empty since not enough \texttt{ready} facts will be derived.
Also in this case we assume the transducer to be inflationary and broadcasting \footnote{Note that the transducer program in this example is not generic. We chose this form to ease the reading, but indeed a generic version of the program can be easily obtained.}.

\noindent\hrulefill
\vspace{-1mm}
\begin{equation*}
\begin{split}
\text{Schema: }&\mathcal{T}_{db} = \{R^{(1)}\}, \mathcal{T}_{sync} = \{\texttt{ready}^{(0)}\}, \mathcal{T}_{aux} = \{C^{(1)}\}, \\&\mathcal{T}_{out} = \{T^{(2)}\}\\
\text{Program: }&\\
\end{split}
\end{equation*}
\vspace{-3mm}
\begin{equation*}
\begin{split}
&\texttt{ready}_{snd}(i) \leftarrow R(u), \texttt{id}(i).\\
&C(count <i>) \leftarrow \texttt{ready}(i).\\
&T_{out}(u) \leftarrow R(u), C(i), i \geq \texttt{3}.
\end{split}
\end{equation*}
\noindent\hrulefill
\vspace{2mm}
\end{ex}
}

%\begin{ex}
%\label{ex:partition_intollerant}
%We conclude this section with an example of a not partition tolerant network. 
%Let retrieve the distributed join program of Example \ref{ex:net_independent}. 
%The following implementation for the join clearly depend on how the input instance is partitioned (and on the network context).

%\noindent\hrulefill
%\vspace{-1mm}
%\begin{equation*}
%\begin{split}
%&\text{Schema: }\mathcal{T}_{db} = \{R^{(2)}, T^{(2)}\}, \mathcal{T}_{sync} = \{S^{(2)}\},\mathcal{T}_{out} = \{J^{(3)}\}\\
%&\text{Program: }\\
%\end{split}
%\end{equation*}
%\vspace{-3mm}
%\begin{equation*}
%\begin{split}
%&S_{snd}(\texttt{u}, \texttt{v}) \leftarrow T(\texttt{u}, \texttt{v}).\\
%&U(\texttt{u}, \texttt{v}) \leftarrow S(\texttt{u}, \texttt{v}).\\
%&J_{out}(\texttt{u}, \texttt{v}, \texttt{w}) \leftarrow R(\texttt{u}, \texttt{v}), U(\texttt{v}, \texttt{w}).
%\end{split}
%\end{equation*}
%\noindent\hrulefill
%\vspace{2mm}

%\noindent One can show that if tuples $T(a, b)$ and $R(b, c)$ exist only in partition $P(i)$, the tuple $J(a, b, c)$ will be never derived, while if $T(a, b)$ and $R(b, c)$ belongs to different partitions, $J(a, b, c)$ is derived.
%\end{ex}
%
Synchronous specifications have the required expressive power:
\begin{lemma}
\label{lm:computable}
Let $\mathcal{L}$ be a language containing {\sc ucq} and such that $\mathcal{L} \subseteq${\sc datalog}$^{\neg}$. 
Every query expressible in $\mathcal{L}$ can be distributively computed in 2 rounds by a broadcasting, inflationary and oblivious $\mathcal{L}$-transducer network.
\end{lemma}
\vspace{-1ex}
\begin{proof}
This is an adaptation to our context of Theorems $4.9$ - $4.10$ of \cite{AmelootNB13}. 
Let $\mathcal{Q}$ be a query expressed in $\mathcal{L}$ having input schema $\mathcal{D}_{in}$ and output schema $\mathcal{D}_{out}$.
We have two cases, $(i)$ $\mathcal{Q}$ is monotonic or $(ii)$ $\mathcal{Q}$ is non-monotonic.
In case $(i)$, we can program a transducer $\mathcal{T}$ so that initially all nodes send out their local database facts. In the next round all nodes will have received all database facts because communication is synchronous and reliable. Relations \texttt{Id} and \texttt{All} are not needed. 
%The transducer can now insert in memory all its database facts plus the received ones.
$\mathcal{Q}$ is evaluated over the union of the local input with the received facts. %In the first round nothing will be output (no facts is still received), while in the second super-step all the facts are generated.
%Indeed no wrong results are generate because the query is monotone.

Formally, %if with $\mathcal{D}_{out}$ we denote the output schema of $Q$, 
the transducer schema will have $\Upsilon_{db} = \mathcal{D}_{in}$, $\Upsilon_{com} = \{R^{^\prime(a)} \mid R^{(a)} \in \Upsilon_{db}\}$, %$\Upsilon_{sync} = \{(R^{sync}, a) \mid R^{(a)} \in \Upsilon_{db}\}$, 
$\Upsilon_{out} = \mathcal{D}_{out}$, while the $system$ and $time$ schemas are as usual.
Denote with $\mathcal{Q}^{\prime}$ the version of the query $\mathcal{Q}$ where all the $edb$ relations are primed.
The transducer program $\mathcal{P}$ is composed by the queries $Q_{snd}$ %, $Q_{ins}$ 
and $Q_{out}$, where $Q_{snd}$ is composed by one rule in the form: $R^{\prime}_{snd}(\bar{u}) \leftarrow R(\bar{u})$
\eat{\begin{align}
\label{eq:computable_1}
&R^{\prime}_{snd}(\bar{u}) \leftarrow R(\bar{u})
%\label{eq:computable_2}
%R^{\prime\prime}(\bar{u}) \leftarrow R(\bar{u}).\\
%\label{eq:computable_3}
%R^{\prime\prime}(\bar{u}) \leftarrow R^\prime(\bar{u}).
\end{align}}
for each $R \in \Upsilon_{db}$, %where (\ref{eq:computable_2}) and (\ref{eq:computable_3}) are copy-rules defining the auxiliary predicate $R^{\prime\prime}$.
while $Q_{out}$ simply contains $\mathcal{Q}^{\prime}$. %having as head a predicate in $\mathcal{D}_{out}$. %and $Q_{inst}$ contains two insertion rules
%\begin{align}
%\label{eq:computable_2}
%&R^{\prime}_{ins}(\bar{u}) \leftarrow R^{sync}(\bar{u}).\\
%\label{eq:computable_3}
%&R^{\prime}_{ins}(\bar{u}) \leftarrow R(\bar{u}).
%\end{align}
%for each $R^{sync}$ ($R$) in $\Upsilon_{sync}$ ($\Upsilon_{db}$).
The specification is monotonic and oblivious.

In case $(ii)$ we follow the same approach as before, but this time the query $\mathcal{Q}$, being non-monotonic, cannot be applied immediately because wrong results could be derived.
To avoid this, we use the transducer $\mathcal{T}$ of case $(i)$, in which we
%We create a new transducer $\mathcal{T}^\prime$ starting from $\mathcal{T}$ by allowing $\mathcal{Q}$ to be evaluated only on the emitted tuples.
%To obtain this, we simply have to drop the rules in the form of eq.s (\ref{eq:computable_2}).
%Intuitively, in this case the rules copying $R^\prime$'s tuples to $R^{\prime\prime}$ are unimportant and $\mathcal{Q}$ can be evaluated just over emitted tuples, by substituting $\mathcal{Q}^{\prime\prime}$ with $\mathcal{Q}^{\prime}$ (the version of $\mathcal{Q}$ where every $edb$ predicate is primed) and setting $Q_{out} = \mathcal{Q}^{\prime}$.
%Note that with this implementation, still negation can be applied at the first round.
add to $\Upsilon_{mem}$ the nullary relation $\texttt{Ready}$, and to $Q_{ins}$ the query: $\texttt{Ready}_{ins}() \leftarrow \neg \texttt{Ready}().$
\eat{\vspace{-2ex}
\begin{align}
\label{eq:computable_4}
&\texttt{Ready}_{ins}() \leftarrow \neg \texttt{Ready}().
%&\texttt{init}_{ins}() \leftarrow \neg \texttt{init}().\\
%\label{eq:computable_5}
%&\texttt{ready}_{ins}() \leftarrow \texttt{init}().
\end{align}
}
In addition we modify the rules in $\mathcal{Q}^\prime$ by adding the literal \texttt{Ready} to their body.
In this way, the query $\mathcal{Q}$ will be evaluated just starting from the second round since at the first one \texttt{Ready} is false. 
We therefore reach our goal.
%Note that the network can be made memoryless if we give up the time-obliviousness. 
\end{proof}

%Interested readers can find in \ref{sc:expressive-power} a thorough study on the expressive power of the \emph{bsp} model.
Examples \ref{ex:computable_monotone} and \ref{ex:computable_non_monotonic} below show two transducers, each computing  a query of one of  the two categories used in the proof of Lemma~\ref{lm:computable}.

\begin{example}
\label{ex:computable_monotone}
Let $\mathcal{Q}$ be the following {\sc ucq}-query:\vspace{-1ex}
\begin{equation*}
%\begin{split}
%&P(\texttt{u}, \texttt{v}) \leftarrow T(\texttt{u}, \texttt{v}).\\
T(u, v) \leftarrow P(u, v), R(u).
%\end{split}
\end{equation*}
%\noindent %with $\mathcal{D}_{in} = \{R^{(1)}, P^{(2)}\}$ and $\mathcal{D}_{out} = T^{(2)}$.
$\mathcal{Q}$ can be computed by the following {\sc ucq}-transducer:

\noindent\hrulefill
\vspace{-1mm}
\begin{equation*}
\begin{split}
\text{Schema: }&\Upsilon_{db} = \{R^{(1)}, P^{(2)}\}, \Upsilon_{snd} = \{R^{\prime(1)}, (P^{\prime(2)}\}, %\mathcal{T}_{sync} = \{(R^{sync}, 1)\},
\\&\Upsilon_{out} = \{T^{(2)}\}\\
\text{Program:}\text{ } &
%&R^{sync}_{snd}(\texttt{u}) \leftarrow R(\texttt{u}).\\
%&R_{ins}^{\prime}(\texttt{u}) \leftarrow R^{sync}(\texttt{u}).\\
R_{snd}^{\prime}(u) \leftarrow R(u).\\
&P_{snd}^{\prime}(u, v) \leftarrow P(u, v).\\
%&R^{\prime\prime}(u) \leftarrow R(u).\\
%&P^{\prime\prime}(u, v) \leftarrow P(u, v).\\
%&R^{\prime\prime}(u) \leftarrow R^\prime(u).\\
%&R^{\prime\prime}(u, v) \leftarrow R^\prime(u, v).\\
%&P(\texttt{u}, \texttt{v}) \leftarrow T(\texttt{u}, \texttt{v}).\\
&T_{out}(u, v) \leftarrow P^{\prime}(u, v), R^{\prime}(u).
\end{split}
\end{equation*}

\noindent\hrulefill
\end{example}
\vspace{-1ex}
\begin{example}
\label{ex:computable_non_monotonic}
Let $\mathcal{Q}$ be the following (non-monotonic) {\sc fo}-query:
\begin{equation*}
\begin{split}
&T(u, z) \leftarrow R(u, v), P(v, z).\\
&Q(u, z) \leftarrow S(u, v), \neg T(v, z), P(w, z).
\end{split}
\end{equation*}

\noindent with $\mathcal{D}_{in} = \{R^{(2)}, P^{(2)}, S^{(2)}\}$ and $\mathcal{D}_{out} = Q^{(2)}$.
A {\sc fo}-transducer computing the same query is: %\vspace{10ex}

\vspace{-1ex}
\noindent\hrulefill
\vspace{-1mm}
\begin{equation*}
\begin{split}
\text{Schema: }&\Upsilon_{db} = \{R^{(2)}, P^{(2)}, S^{(2)}\}, \Upsilon_{mem} = \{\texttt{Ready}^{(0)}\}, \\& \Upsilon_{snd} = \{R^{\prime(2)}, P^{\prime(2)}, S^{\prime(2)}\}, \Upsilon_{out} = \{Q^{(2)}\}\\
\text{Program: }&
%&R^{snd}_{snd}(\texttt{u}, \texttt{v}) \leftarrow R(\texttt{u}, \texttt{v}).\\
%&P^{snd}_{snd}(\texttt{u}, \texttt{v}) \leftarrow P(\texttt{u}, \texttt{v}).\\
%&S^{snd}_{snd}(\texttt{u}, \texttt{v}) \leftarrow S(\texttt{u}, \texttt{v}).\\
%&R_{ins}^{\prime}(\texttt{u}, \texttt{v}) \leftarrow R^{snd}(\texttt{u}, \texttt{v}).\\
%&P_{ins}^{\prime}(\texttt{u}, \texttt{v}) \leftarrow P^{snd}(\texttt{u}, \texttt{v}).\\
%&S_{ins}^{\prime}(\texttt{u}, \texttt{v}) \leftarrow S^{snd}(\texttt{u}, \texttt{v}).\\
 R_{snd}^{\prime}(u, v) \leftarrow R(u, v).\\
&P_{snd}^{\prime}(u, v) \leftarrow P(u, v).\\
&S_{snd}^{\prime}(u, v) \leftarrow S(u, v).\\
&\texttt{Ready}_{ins}() \leftarrow \neg \texttt{Ready}().\\
%&\texttt{ready}_{ins}() \leftarrow \texttt{init}().\\
&T(u, z) \leftarrow R^\prime(u, v), P^\prime(v, z), \texttt{Ready}().\\
&Q(u, z)_{out} \leftarrow S^\prime(u, v), \neg T(v, z), P^\prime(w, z), \texttt{Ready}().
\end{split}
\end{equation*}

\vspace{-1.5ex}
\noindent\hrulefill
\end{example}

\eat{\begin{proof}[(sketch)]
This is an application to our context of Theorems $4.9$ - $4.10$ of \cite{AmelootNB13}. 
%Here we just sketch the idea while the full proof is available in Appendix \ref{a:proof_computable}.
Let $\mathcal{Q}$ be an $\mathcal{L}$-query. % having input schema $\mathcal{D}_{in} = edb(\mathcal{Q})$, and output schema $\mathcal{D}_{out}$.
We have two cases: $(i)$ $\mathcal{Q}$ is monotonic or $(ii)$ it is not.
In case $(i)$, we can program a transducer so that initially all nodes send out their local database facts. In the next round all the nodes will have received all database facts since we are considering \emph{rsync} systems. 
%The transducer can now insert in memory all its database facts plus the received ones.
\emph{System} and \emph{time} relations are not needed.
The query $\mathcal{Q}$ is then evaluated on the union of the local initial partition with the received facts. %In the first round nothing will be output (no facts is still received), while in the second super-step all the facts are generated.
%Indeed no wrong results are generate because the query is monotone.
%The network is monotone, inflationary and oblivious. 

In case $(ii)$ we follow the same approach as before, but this time the query $\mathcal{Q}$, being non-monotonic, cannot be applied immediately because wrong results could be derived.
To avoid this, let us assume $\mathcal{T}$ is the transducer of case $(i)$. We create a new transducer $\mathcal{T}^\prime$ starting from $\mathcal{T}$ by adding to $\Upsilon_{mem}$ the nullary relation $\texttt{Ready}$, and by modifying $Q_{ins}$ so that $\texttt{Ready}$ will be true only after the first round.
In addition, \texttt{Ready} is AND-attached to all the queries composing $\mathcal{Q}$, so that they will be evaluated once  the emitted instances are received.
%Also in this case the network is inflationary and oblivious.
\end{proof}
\vspace{-2mm}
}

Lemma~\ref{lm:computable} permits us to draw the following conclusion: under the \emph{bsp} semantics, monotonic and non-monotonic queries behave in the same way; two rounds are needed in both cases. This is due to the fact that, contrary to what happens in the asynchronous case \cite{AmelootNB13}, we are guaranteed by the reliability of the communication and the synchronous assumption that, starting from the second round on, every node will compute the query over every emitted instance. 
Conversely, in the asynchronous case, as a result of the non-determinism of the communication, we are never guaranteed, without coordination, when every sent fact will be actually received.
As a consequence, under this latter model we don't know -- without coordination -- when negation can be safely applied, because it could be applied ``too early", i.e., before all facts over the negated literal are received (or deduced). 

\vspace{-1ex}
\subsection{The CALM Conjecture}
\label{sec:co_co}
The \emph{CALM conjecture} \cite{Hellerstein10} specifies that the class of \emph{monotonic programs} can be distributively computed in an eventually consistent, \emph{coordination-free} way. 
%Although consistency is not decidable in general, there is a well defined class of distributed systems which is always \emph{eventually consistent}: \ie the one that fall under the \emph{CALM conjecture} \cite{Hellerstein:2010:DIE:1860702.1860704}. 
%
%\emph{network independent, partition-tolerant, time-oblivious transducer networks}.
%\begin{pr}
%All the network independent, partition-tolerant, time-oblivious transducer networks are (eventually) consistent.
%\end{pr}
%\begin{proof}
%All the network independent, partition-tolerant, timo-oblivious transducer networks are consistent so whatever pair of configurations we chose in the class, the related configured networks  are always (eventual) consistent.
%\end{proof}
%\begin{co}
%All the network-independent time-oblivious disseminating transducer networks are (eventually) consistent.
%\end{co}
%If we are interested in more general results about eventual consistency, coordination and monotonicity, the \emph{CALM conjecture} states \cite{Hellerstein:2010:DIE:1860702.1860704}:
%
%\begin{cn}[CALM]
%\label{cn:calm}
%A program has an eventually consistent, coordination-free evaluation strategy if and only if it is expressible in (monotonic) datalog.
%\end{cn}
\noindent CALM has been proven in this (revisited) form for asynchronous systems \cite{AmelootNB13}:
\begin{conjecture}
\label{cn:calm_new}
A query can be distributively computed by a coordination-free transducer network if and only if it is monotonic. %it is expressible in datalog.
\end{conjecture}
%
%where the concept of \emph{coordination-freedom} subsume the one of eventually consistent. 
Surprisingly enough, the only-if direction on the conjecture does not hold in \emph{bsp} settings under the broadcasting communication model.
Before showing this, we have to adapt the definition of \emph{coordination-free} as expressed in \cite{AmelootNB13} to our synchronous model.

The concept of coordination suggests that all the nodes in a network need to exchange information and wait until an agreement is reached about a common property of interest. 
Following this intuition, Ameloot et al. established that a specification is coordination-free if communication is not strictly necessary to obtain a consistent final result. 
Put in a more formal context: a specification is coordination-free if $(i)$ it is independent and $(ii)$ a ``perfect'' distribution function exists such that communication is not required to achieve the final outcome. 
That is, the class generated from that specification admits a unique output, independently of the  configuration.
Hence, among all the possible distribution functions we can select one such that, if we turn communication off,  the correct result is still returned.
%a communication-free specification can be build, eventually consistent with the given one.
%Let us first start with eventually consistent.
\eat{
\vspace{-0.5mm}
\begin{df}
Let $\mathcal{N}$, $\mathcal{N}^\prime$ two specifications such that $\Upsilon_{db} =
\Upsilon^\prime_{db}$ and $\Upsilon_{out} = \Upsilon^\prime_{out}$, an initial instance \textbf{I} and  two configurations $(N, t, P)$, $(N^\prime, t^\prime, P^\prime)$, we say that two systems $\mathcal{S}^{\emph{rsync}}_{\mathcal{N}}(N, t, P, \textbf{I})$ and $\mathcal{S}^{\emph{rsync}}_{\mathcal{N}^\prime}(N^\prime,t^\prime,P^\prime, \textbf{I})$ are \emph{eventually consistent}  if $out(*) = out^\prime(*)$.
\end{df}
\vspace{-0.5mm}
Informally, two systems are eventually consistent if they have the same \emph{database} and \emph{output} schema, and the output instances at (possibly different) quiescence are equal. 
}
%Note that we have employed the term ``eventually'' to highlight the fact that the two systems are required to agree just on the initial instance \textbf{I} and on the final output state, and not on the entire execution.
%
%We can now introduce the formal definition of coordination-freedom as stated in \cite{Ameloot:2011:RTD:1989284.1989321} and adjusted to fit  our context.
%
%Formally:

\vspace{2ex}
\begin{definition}\label{def:coordinationfree}
\eat{
Let $\mathcal{N}$ be an independent specification, and let $\mathbfcal{S}^{\emph{bsp}}_{\mathcal{N}}$ be the related convergent class of systems.
%Assume an instance \textbf{I}, and 
In addition,  if $\mathcal{F}$ is the com\-munication-free version of $\mathcal{N}$ obtained by properly setting the communication function, denote with $\mathbfcal{S}^{\emph{bsp}}_{\mathcal{F}}$ the class constructed starting from $\mathbfcal{S}^{\emph{bsp}}_{\mathcal{N}}$ by adjoining the run composing $\mathcal{S}^{\emph{bsp}}_{\mathcal{F}}(N, D, \textbf{I})$.
An independent specification $\mathcal{N}$ is said \emph{coordination-free} if, for every initial instance \textbf{I}, a non-trivial configuration $(N, D)$ exists s.t., $\mathbfcal{S}^{\emph{bsp}}_{\mathcal{F}}$ is convergent.
}
\eat{
An independent specification $\mathcal{N}$ is \emph{coordination-free} if for every initial instance \textbf{I} a non-trivial configuration $(N, D)$ exists s.t., if $\mathcal{N}^f$ is the com\-munication-free version of $\mathcal{N}$ obtained by setting $\gamma^f$ as the communication function, %derived by making each emission query in $\mathcal{N}$ an insertion query, 
and we denote with $out^f(*)$, and $out(*)$ respectively the outputs of systems $\mathcal{S}^{\emph{bsp}}_{\mathcal{N}^{f}}(N, D, \textbf{I})$, and $\mathcal{S}^{\emph{bsp}}_{\mathcal{N}}(N, D, \textbf{I})$, we have that $out^f(*) = out(*)$.
}
Let $\mathcal{N}$  be an independent specification, and $\mathcal{F}$ its communication-free version. 
We say that $\mathcal{N}$ is \emph{coordination-free} if $\forall \textbf{I} \in inst(\Upsilon_{db})$ a non-trivial configuration $(N, D)$ exists s.t., $\mathbfcal{S}^{\emph{bsp}}_{\mathcal{F}}$ is convergent; where, $\mathbfcal{S}^{\emph{bsp}}_{\mathcal{F}}$ is constructed by adjoining to the convergent class $\mathbfcal{S}^{\emph{bsp}}_{\mathcal{N}}(\textbf{I})$ the run defining $\mathcal{S}^{\emph{bsp}}_{\mathcal{F}}(N, D, \textbf{I})$.
\end{definition}

\noindent In the following, we will say that runs such as $\mathcal{S}^{\emph{bsp}}_{\mathcal{N}}(N, D, \textbf{I})$ and $\mathcal{S}^{\emph{bsp}}_{\mathcal{F}}(N, D, \textbf{I})$ are \emph{eventually consistent} with each other. That is, their complete execution may differ, but they eventually converge to the same unique final output.
To simplify the notation, we will sometimes directly say that $\mathcal{N}$ and $\mathcal{F}$ are eventually consistent.  

It is now easy to see that with this definition there are non-monotonic queries that can be distributively computed by coordination-free specifications, as the next example shows.
%Note that in the above definition we are considering for coordination freedom only consistent specification, \ie specifications distributively computing a query.
%Under our setting, it turns out that even the specifications used to distributively compute non-monotonic queries can be coordination-free, as the next example shows.

\begin{example}
\label{ex:false_calm}
Let $\mathcal{Q}_{emp}$ be the ``emptiness'' query of \cite{AmelootNB13}: given a nullary database relation schema $R^{(0)}$ and a nullary output relation $T^{(0)}$, $\mathcal{Q}_{emp}$ outputs $T$ iff $I_R$ is empty. 
The query is non-monotonic: if $I_R$ is initially empty, then $T$ is produced, but
if just one fact is added to $R$, $T$ is not derived, \ie $I_T$ is now empty.
A broadcasting {\sc fo}-transducer network $\mathcal{N}$ can be easily generated to distributively compute $\mathcal{Q}_{emp}$: first every node emits $R$ if its local partition is not empty, and then each node locally evaluates the emptiness of $R$. 
Since the whole initial instance is installed on every node when $R$ is checked for emptiness, $T$ is true only if $R$ is actually empty on the initial instance.
The complete specification follows.

%\vspace{-0.5mm}
\noindent\hrulefill
\vspace{-1mm}

\begin{equation*}
\begin{split}
\text{Schema: }&\Upsilon_{db} = \{R^{(0)}\}, \Upsilon_{mem} = \{\texttt{Ready}^{(0)}\}, \Upsilon_{com} = \{S^{(0)}\}, \\\vspace{-5mm}&
\Upsilon_{out} = \{T^{(0)}\}.\\
\text{Program:} 
&\text{ }S_{snd}() \leftarrow R().\\
&\text{ }\texttt{Ready}_{ins}() \leftarrow \neg \texttt{Ready}().\\
&\text{ }T_{out}() \leftarrow \neg S(), \texttt{Ready}().
\end{split}
\end{equation*}

\noindent\hrulefill
\vspace{1mm}

\noindent Assume that $(N, D)$ is a non-trivial configuration and let $\mathcal{F}$ be the communication-free version of $\mathcal{N}$ above. 
Clearly, whichever initial instance \textbf{I} we select, $\mathcal{S}^{\emph{bsp}}_{\mathcal{N}}(N, D, \textbf{I})$ and $\mathcal{S}^{\emph{bsp}}_{\mathcal{F}}(N, D, \textbf{I})$ are eventually consistent when $D$ installs \textbf{I} on every node.

Note that,  in asynchronous settings, the same query requires coordination: since emitted facts are non-deterministically received, the only way to compute the correct result is that every node coordinates with each other in order to understand if the input instance is \emph{globally} empty.
%\cite{Ameloot:2013:RTD:2450142.2450151} depicts a protocol to achieve this.
\end{example}

%We are now able to falsify the CALM conjecture in synchronous and reliable systems, after a little re-formalization following \cite{Ameloot:2011:RTD:1989284.1989321} and our notation:
%\begin{cn}[CALM]
%\label{cn:calm_new}
%A query can be distributively computed by a coordination-free transducer network if and only if it is expressible in datalog.
%\end{cn}
%\begin{proof}
%The if direction of the conjecture is easily demonstrable: from Preposition \ref{lm:computable} we know that a transducer network exists such that a monotonic query is distributively computable. 
%The transducer network is indeed coordination-free: the specification is consistent by definition of distributively computable, and the communication-free version is consistent whichever non-trivial configuration we chose having $P$ installing the full instance on every node. 
%For what concern the only-if direction, this is clearly not true, as one can realize from Example \ref{ex:false_calm}.
%\end{proof}
%
%\noindent Note that eventually consistency expressed in Conjecture \ref{cn:calm} has been dropped from Conjecture \ref{cn:calm_new} because intrinsic in the definition of coordination-freedom (and of convergence).
The above example shows that the only-if direction of CALM, stating that only monotonic queries can be computed in a coordination-free way, \emph{ does not hold in general in synchronous settings}~(Contribution 1).
This result is indeed interesting although expected: when we move from the general asynchronous model to the  \emph{more restrictive, \emph{bsp} setting}, we no longer have a complete understanding of which queries can be computed without coordination, and which ones, instead, do require coordination.
\eat{It turns out that 
%\emph{rsync} systems are far more interesting than the asynchronous ones, and 
both the communication model and the definition of coordination proposed in \cite{Ameloot:2011:RTD:1989284.1989321} are not broad enough to work in general for synchronous systems.
As the reader may have realized, this is due to the fact that in broadcasting synchronous systems, coordination -- as defined by Ameloot et al. and recalled in this section -- is already ``baked'' into the model.
In the next sections we will see that our definition of coordination-freedom (Definition \ref{df:specification_coordination_free}) 
guarantees eventually consistent computation for those queries that do not rely on broadcasting in order to progress. %all nodes must necessarily have received a certain fact. 
That is, the discriminating condition for eventual consistency is not monotonicity, but the  fact that it \emph{is not necessary } to send a fact to all the nodes composing a network.  }
%An intuition of this is given already by Example \ref{ex:false_calm}. 
%We will then see that the CALM conjecture can indeed be proven also in our context, under a more general definition of coordination and a more flexible communication model.
%We will then see that different types of coordination actually exist. %and that they are not visible in the asynchronous case since they coincide under such type of systems. 
%As a concluding remark, we want to point out that the CALM conjecture does not hold in our settings because the definition of coordination-freedom proposed in \cite{Ameloot:2011:RTD:1989284.1989321} is not generic, in the sense that it does not take in consideration different types of systems.
%In synchronous and reliable systems, in fact, we have that broadcasting communication always bring to common knowledge \cite{}. 
%Coordination is therefore always obtained 
%For instance, under the broadcasting communication model, in a synchronous and reliable system every emitted fact becomes \emph{common knowledge} in the successive super-step.
%We can therefore able to state that a relation between coordination and (reliable and synchronous) broadcasting exists.
%Despite this unintuitive result, In the next chapter we will see that the CALM conjecture is indeed true also in our context, under a more general definition of coordination-freedom and a more flexible communication model. % which take in consideration the above remarks. 
It turns out that the communication model and the definition of coordination proposed in \cite{AmelootNB13} cannot capture a notion of coordination freedom appropriate also for synchronous systems. As the reader may have realized, this is due to the presence of $(i)$ broadcasting communication, and $(ii)$ \emph{bsp} system semantics: in broadcasting synchronous systems, the form of coordination defined by Ameloot et al. is already ``baked" into the model because synchronization barriers per se provide nodes with the possibility of ``indirectly deducing" the global status of the network. As a result, some of the queries that, in the asynchronous communication model, were not computable in a coordination-free way turn out to be so in synchronous systems.
Arguably, one could conjecture that all computable queries are actually coordination-free computable (using again the notion introduced in \cite{AmelootNB13}) in synchronous systems.
%Conversely, in asynchronous systems, the knowledge of the global status cannot be acquired ``indirectly" but only directly through communication among nodes. 

In Section~\ref{sec:coordination_sync} we provide a new, less permissive definition of \emph{coordination-freedom}, 
%(Definition~\ref{df:specification_coordination_free})
more appropriate for synchronous settings. 
Under this definition, \emph{the discriminating condition for coordination freedom is not the absence of communication among nodes, but the more restrictive one that nodes do not need to acquire knowledge of a global property of the network to correctly compute a query}.
We will then see that by weakening first $(i)$ (Section~\ref{sec:network_independence}) and then $(ii)$ (Section~\ref{sec:sync_calm}), the coordination baked into the synchronization barrier results inappropriate to make specific classes of queries (described later) consistent, whereby additional forms of coordination are required.
In Section~\ref{sec:sync_calm} we will additionally show that indeed our definition of coordination subsumes the one previously introduced by Ameloot et al.: when we  remove all the constraints imposed on the synchronous system, the two definitions fall together.
The intuition is that, if a node does not know when all the other nodes are ``done",  conclusions (such as deducing that negation can be applied over a relation) might be applied ``too early". That is, in order to apply certain specific types of deductions (which specific type will be made clear in the next sections), \emph{common knowledge} \cite{FaginHMV03} is required. Acquiring common knowledge of a global property in asynchronous systems requires exchange of messages among all nodes, and this is why communication-freedom can be connected to coordination freedom in this type of systems (i.e., no coordination can be reached without communication). Conversely, using synchronous broadcasting specifications, common knowledge can be acquired by each node indirectly, even without having any fact being communicated. %In the next Section we devise a new type of transducer network allowing us to to go beyond this limitation.\vspace{-1ex}

%The condition determining coordination is therefore not communication freedom, but the common knowledge of a property like local termination. 
%Interestingly, while such condition can be hard to formally define, \cite{AmelootNB13} already provided us with an hint on how it can be indirectly formulated: coordination freedom is directly related with obliviousness. This is, in order to acquire common knowledge, nodes of a transducer network need to access \emph{system} relations; with no knowledge of the nodes composing the network, no common knowledge can be reached. We can that state the following coordination property:
%\begin{property}
%A sufficient condition for a specification to use coordination is to access \emph{system} relations.
%\end{property}
%Indeed broadcasting synchronous specifications can be seen as always using a form of coordination since (i) they are sending facts to all nodes composing the network (i.e., all nodes composing the network are implicitly known); and (ii) all received facts are simultaneously received (condition \textbf{S3}$^\prime$)

\section{Hashing Transducer Networks and Parallel Computation}
\label{sec:network_independence}
%!TEX root = TPLP.tex
%
%In the previous section we have considered only broadcasting %(and communication-free) 
%networks.
In the previous section we have seen that synchronization barriers, together with a broadcasting communication model, allow non-monotonic queries to be computed in a coordination-free way (\cf Example~\ref{ex:false_calm}); more than that,  broadcasting specifications  do not appear to be really useful from a practical perspective.
As a consequence, following other parallel programming models such as MapReduce and Spark, in this section we are going to introduce \emph{hashing transducers} (Contribution 2) \ie relational transducers equipped with a \emph{content-based communication model}. 
Under this new model,  the node to which an emitted fact must be addressed is derived using a hash function applied to a subset of its terms called \emph{keys}.
%defining two new types of transducer networks: hashing and addressing.

\subsection{Hashing Transducer Networks}

Let $\Upsilon$ be a transducer schema.
For each relation $R^{(a)} \in \Upsilon_{com}$, %\setminus (\Upsilon_{db} \cup \Upsilon_{sys})$, 
we fix a subset of its terms as the key-terms for that relation. 
\kWlog we will then use the notation $R^{(k,a)}$ %, or equivalently $(R, k, a)$ 
 to refer to a relation $R$ of arity $a$ having the first $k$ terms specifying its key.
 %\footnote{Note that if a relation $R$ has its keys scattered in the tuple (for example the first and third term) w.l.o.g. we can substitute all the occurrence of $R$ with a new relation having all the keys in the first $k$ terms.}.
As a notation, we associate to every transducer schema $\Upsilon$ a \emph{key-set} $\mathcal{K}$ mapping every relation $R$ for which a key is defined, to the related set of keys.

It is now appropriate to define how a hashing transducer $\mathcal{T}$ is represented.
$\mathcal{T}$ is hashing if defined by a tuple $(\mathcal{P}, \Upsilon, \mathcal{K})$, where $\mathcal{P}$ is the transducer program, $\Upsilon$ the schema, and $\mathcal{K}$ a key-set. 
%Similarly to \cite{Zinn:2012:WC:2274576.2274588}, 
With each transducer network we can now associate a \emph{distributed hash} \emph{mapping} $\mathcal{H}$ %: fact(\Upsilon_{com}) \rightarrow 2^N {\setminus}\emptyset$ 
binding each $communication$ fact with the nonempty set of nodes to which the fact belongs.
 Given a family $H$ of unary hash functions $h: \textbf{dom} \rightarrow N$ and a fact $R(\texttt{u}_1, \ldots \texttt{u}_a)$ over a relation $R^{(k, a)} \in \Upsilon_{com}$, $\mathcal{H}$  distributes $R(\texttt{u}_1, \ldots \texttt{u}_a)$ to the nodes in the following way:
\begin{equation}
\label{eq:hashing}
\mathcal{H}(R(\texttt{u}_1, \ldots \texttt{u}_a)) = 
\begin{cases}
N & \text{if } k = 0\\
\bigcup_{i \in 1..k} \{ h_i(\texttt{u}_i)\}  & \text{otherwise}
\end{cases}
\end{equation}
Informally, we employ hash functions to deterministically obtain the location(s) to which a fact belongs from its key-term values specified in $\mathcal{K}$, so that $\mathcal{H}$ maps each fact to the set of nodes to which it must be delivered.
%Note that we allow the same hash function to be used for different terms.
Two characteristics are noteworthy in our model.
Firstly, we allow a fact to be hashed to multiple target nodes\footnote{Note that this one-to-many distributed mapping is used here for coherence with the set oriented {\sc datalog} semantics. One-to-one communication behavior can also be employed, for example by using surrogate keys.}. This multicasting model enables us to gracefully move from broadcasting to unicast situations, while studying how different classes of queries behave. Additionally, this allows to model both regular MapReduce-style shuffling, as well as Hypercube shuffling~\cite{AfratiU10} requiring records to be replicated over multiple nodes. %In practice, this can be seen as having \emph{speculative tasks} running, either for fault-tolerance or for stragglers mitigation.
Secondly, we consider a family of hash functions instead of a single function. In this way we are able to express specific behaviors when needed; 
for instance, we can establish that facts containing a certain constant %, for example $\texttt{c}$, 
in certain cases are addressed to a predefined node, while, in general, they can be addressed to all nodes. %in $\bigcup_{h \in H} \in h(\texttt{c})$.
%For technical reasons we always require that whichever sequence of constant $\texttt{c}_1, \ldots, \texttt{c}_k \in \textbf{dom}$, the number of nodes returned by $h_k(\texttt{c}_1, \ldots, \texttt{c}_k)$ is equal to 1.
To maintain our model simple, we assume the codomain of each hash family to coincide with $N$, i.e., $N = \bigcup_{h \in H} \bigcup_{c \in \textbf{dom}}\{h(c)\}$. %We suggest readers to consult Appendix C if interested in knowing what happen when this assumption does not hold.

Given a relation $R^{(k, a)}$, two key-settings are of special interest: $(i)$ no key is set, and we write $k = 0$; and $(ii)$ the set of keys is \emph{maximal}, \ie $k = a$.
In the former case, we have that
%In the special case in which for a relation $R^{(a)}$ no key is set, we write $k = \infty$, and we have that for each fact $R(\texttt{u}_1, \ldots, \texttt{u}_a) \in I_R$,  $\mathcal{H}(R(\texttt{u}_1, \ldots, \texttt{u}_a)) = \bigcup_{h \in H} \bigcup_{\texttt{c} \in \textbf{dom}} h(\texttt{c})$, \ie 
every tuple is addressed to all the nodes $N$.
We then say that a send query is \emph{broadcasting} if the head relation $R$ is such that $k = 0$.
%The intuition is that every fact derived by such query is broadcasted to every node. 
Furthermore, in the case in which all the relations in the domain of $\mathcal{K}$ have $k = 0$, we say that $\mathcal{K}$ is \emph{unrestricted};
it is \emph{restricted} if, instead, for no relation $k = 0$. 
%If, for every $R^{(k, a)} \in \mathcal{K}$, $k = a$, we refer to $\mathcal{K}$ as \emph{maximal}.
\eat{Given the above, we can identify two possible kinds of hash families: \emph{partitioned}, when $N_H \subsetneq N$, and \emph{non-partitioned} otherwise. 
%
%and one instead where $k = \emptyset$.
%For the former we have that $H(I_R) = N$ -- every tuple belongs to all the possible nodes reachable by the hashing family $\mathcal{H}$ -- while, for the latter, $H(I_R) = i$, where $i$ is the node emitting the instance $I_R$.
%For similar reasons, in the special case is which for a relation $R$, $k = 0$, we have that $H(R(\bar{u})) = N$.
%In practice we then have that a broadcasting transducer is an hashing transducer in which $k = 0$ for every relation $R \in \Upsilon$. 
%
%We can thus identify two possible kinds of hash families: \emph{partitioned} when $\bigcup_{h \in \mathcal{H}} h(\textbf{dom}) \subset N$, and \emph{non-partitioned} otherwise.
In the former case we have that the set of nodes composing the network can be divided into two sets: the nodes to which data can be hashed, and that  therefore are actively involved in the computation, and the nodes in $N \setminus N_H$. The latter do not receive any emitted tuple.
To make each node aware that potentially only a subset of $N$ is actively receiving data, we assume that the system relation \texttt{All} contains a value for every node in $N_H$ (instead of $N$).

\stitle{Remark}: The model just discussed is as generic as possible, and therefore it allows us to represent different practical situations. 
For instance, broadcasting is a feature both of MapReduce and Spark, while partitioned hash families permit to elegantly model elastic situations where node resources can be returned back to a manager (such as Mesos \cite{HindmanKZGJKSS11} or YARN \cite{VavilapalliMDA13}) when not required anymore. 
However, as we will see soon, for non-monotonic queries the knowledge that no data will be received as input by a node does not automatically mean that node resources can be freed.
\vspace{1mm}
}
%\eat{For partitioned hashing families we hence name the nodes in $\bigcup_{h \in H} \bigcup_{\texttt{c} \in \textbf{dom}} h(\texttt{c})$ as \emph{active}, and the nodes in $N \setminus \bigcup_{h \in H} \bigcup_{\texttt{c} \in \textbf{dom}} h(\texttt{c})$ as \emph{passive}.
%}
%To make each node aware that %under partitioned hash families 
%potentially only a subset of $N$ is actively receiving data, we assume that the $system$ relation $\texttt{All}$ contains a value for every node in $\bigcup_{h \in H} \bigcup_{\texttt{c} \in \textbf{dom}} h(\texttt{c})$.
%. 
%In practice, we have that a broadcasting transducer is an hashing transducer in which $k = 0$ for every relation $R \in \Upsilon$, and $\mathcal{H}$ is allowed to address all the nodes in $N$.
%we don't allow $database$ and $synchronous$ relations to be nullary.
%Note that because of this, the emptiness query of Example \ref{ex:false_calm} cannot be computed by an hashing transducer. 
%and for the special case is which $h$ is applied over a nullary relation, we have that $h(\emptyset) = f(N)$, where $f$ is a determinist function.
%To simplify the notation, if the key of the relation is not specified, we assume the key to be by default the first term of the tuple.
%As a convention, we always specify just the key of $synchronous$ relations into transducer programs in order to maintain the notation not too heavy.
We can now base the definition of the communication function $\gamma$ on $\mathcal{H}$ and %in such a way that $\gamma(i)(\textbf{G})$
%We can now refine the communication function $\gamma(i)$ as an application of $\mathcal{H}$ over an instance over $\Upsilon^e_{snd}$, and 
%returns only the facts in the global state \textbf{G} node $i$ is responsible for.
%More concretely, we use $\gamma^h$ to denote this new communication function: $\forall i \in N,\gamma^h(i)(\textbf{G}) = \{R(\bar{\texttt{u}}) \mid R \in \Upsilon_{com}, i = \mathcal{H}(R(\bar{\texttt{u}})), \text{ and } R(\bar{\texttt{u}}) \in I^e_{snd}\}$.
we then call \emph{hashing} this new type of synchronous specification. %and say that tuples are \emph{shuffled} %\footnote{This name is inspired by the MapReduce/Spark shuffling phase.} 
%when they are sent using the hashing model. %(instead of broadcasting). 
Unless stated otherwise, from now on the term ``specification'' will actually denote a \emph{hashing} specification (transducer network). %, and hence we will drop the $h$ suffix.
%W.l.o.g. we assume that every relation in $\Upsilon \setminus (\Upsilon_{db} \cup \Upsilon_{sys} \cup \Upsilon_{sync})$ is always such that $k = \infty$.
%For conciseness we then consider $\mathcal{K}$ to contain just the keys for $synchronous$ relations.
%and we use the signature $\gamma^h$ to specify that a network is actually an hashing network.
The definition of communica\-tion-free can now be slightly revisited: a hashing specification $\mathcal{N}$ is communica\-tion-free also when, for all relations $R \in \Upsilon_{com}$, if $R(\bar{\texttt{u}})$ is a fact derived at node $i$ by a sending query, then $\mathcal{H}(R(\bar{\texttt{u}})) = \{i\}$.
%Clearly, every $synchronous$ relation with $k = 0$ is communication-free.
%Remark that even with hashing networks, the program of Example \ref{ex:detect} can still be used to find if a network is composed by one or more nodes.
%With a little abuse of notation, in the following when we refer to an hashing transducer schema we will actually refer to both the schema $\Upsilon$ and the related associated set of keys $\mathcal{K}$.
%To remain concise, we will always express into the schema just the specification of the keys for $synchronous$ relations.

%The special case in which the key is $\infty$ can instead be modelled by assigning the broadcasting address to the tuple.

%In this section we give some additional (really naive and unoptimized) examples of hashing transducer networks.
\begin{example}
\label{ex:hashing} 
This program is the hashed version of Example \ref{ex:net_independent}:

\noindent\hrulefill
\vspace{-1mm}
\begin{equation*}
\begin{split}
\text{Schema: }&\Upsilon_{db} = \{R^{(2)}, T^{(2)}\}, \Upsilon_{com} = \{S^{(1, 2)}, U^{(1, 2)}\}, \\&\Upsilon_{out} = \{J^{(3)}\}\\
\end{split}
\end{equation*}
\begin{equation*}
\begin{split}
\text{Program:} \text{ }&
S_{snd}(u, v) \leftarrow R(u, v).\\
&U_{snd}(u, v) \leftarrow T(u, v).\\
%&\texttt{ready}() \leftarrow \neg \texttt{ready}().\\
&J_{out}(u, v, w) \leftarrow S(u, v), U(u, w).
\end{split}
\end{equation*}

\vspace{-1mm}
\noindent\hrulefill
\vspace{1mm}

\noindent In this new guise, every tuple emitted over $S$ and $U$ is hashed on the first term, so that we are assured that at least a node exists to which each pair of joining tuples is issued. 
Note that such program models exactly how (reduce-side) joins are implemented in the MapReduce framework.
\end{example}

\paragraph{Shuffling Transducers.}
Given a hashing transducer network, %similarly to disseminating transducers under the broadcasting communication model, 
we can directly apply hashing functions to the $database$ relations and, once all the initial tuples are hashed, the actual queries can be applied.
\noindent We refer to such type of transducers as \emph{shuffling}. %~\footnote{This name is inspired to the MapReduce shuffling phase \cite{DeanG20}.}. %although their functionality is not quite the same.}.
%We similarly say that a transducer is shuffling if its program is shuffling.
Intuitively, when $\mathcal{K}$ is unrestricted the same behavior described in the proof of Lemma \ref{lm:computable} under the broadcasting communication model is obtained.

More formally, given a query $\mathcal{Q}: inst(\mathcal{D}_{in}) \rightarrow inst(\mathcal{D}_{out})$ over input and output schema $\mathcal{D}_{in}$ and $\mathcal{D}_{out}$ respectively, a shuffle transducer network can be devised so that the transducer schema $\Upsilon$ has $\Upsilon_{db} = \mathcal{D}_{in}, \Upsilon_{mem} = \{\texttt{Ready}^{(0)}\}, \Upsilon_{snd} = \mathcal{D}^\prime_{in}$ and $\Upsilon_{out} = \mathcal{D}_{out}$, where $\mathcal{D}^\prime_{in}$ is composed by the primed version of the relations in $\mathcal{D}_{in}$; the transducer program $\mathcal{T}$ is instead formed by the queries $Q_{ins} = \{\texttt{Ready}_{ins}() \leftarrow \neg \texttt{Ready}().\}$, $Q_{snd} = \{R^\prime_{snd}(\bar{u}) \leftarrow R(\bar{u}) | R \in \Upsilon_{db}\}$, and $Q_{out} = \mathcal{Q}^\prime$,  where $\mathcal{Q}^\prime$ is generated from $\mathcal{Q}$ by priming all relations over the input schema $\mathcal{D}_{in}$ and where the predicate $\texttt{Ready}()$ is adjoined to the query.

We will often use shuffling transducers in the proofs and examples we are going to introduce in the remainder of this section and in the following ones. But we first introduce some properties of shuffling specifications.

\subsection{Properties of Shuffling Transducer Networks}
\emph{Safety} and \emph{liveness} are common properties used to describe specifications over distributed systems \cite{kindler1994safety}.
Informally, the safety property is used to state that ``\emph{nothing bad will ever happen}'', while the liveness property certifies that ``\emph{something good will eventually happen}''.
In our model the two properties can be respectively restated as ``\emph{no wrong fact will ever be derived}'', %while the liveness property can be reformulated as 
and ``\emph{some fact will eventually be derived}''.
%In this new guise, as Example \ref{ex:emptiness_partition} showed, the safety property must be enforced must be clearly enforced when non-monotonic queries exist in a program, which consequently must be evaluated only if it is safe to do so.
%Conversely, the liveness property must always be satisfied if we want that every fact will actually be returned.

%\begin{definition}
%\label{df:live}
%Let $\mathcal{T} = (\mathcal{P}, \Upsilon, \mathcal{K})$ be a hashing transducer.
%We say that a query $q_R$ in $\mathcal{P}$ is \emph{evaluated on a proper instance} if for every relation $P^{(a)} \in \Upsilon_{com}$ %\setminus (\Upsilon_{db} \cup \Upsilon_{sys})$
% appearing in the body of $q_R$, every fact $\textbf{f}$ derived by a send query $P_{snd}$ and satisfying $\mathcal{H}(\textbf{f}) \cap \{ i \} \neq \emptyset$ is in $I^i_P$.
%%an instance $I^i_R$ exists on node $i \in N$ such that every fact over $R$ satisfying $H(R(\bar{u})) \cap i \neq \emptyset$ is in $I^i_R$.
%%\begin{equation}
%%\label{eq:proper_instance}
%%\forall i \in N, H(R(\bar{u})), I^i_{R}
%%\end{equation}
%\end{definition}
%
First of all, note that, by the properties \textbf{S1} - \textbf{S3}$^\prime$ and \textbf{R1} - \textbf{R2} of \emph{bsp} systems and by definition of hashing communication, every fact $P(\bar{\texttt{u}})$ over a relation $P \in \Upsilon_{com}$ derived by a send query in round $t$, will be in the instance $I^i_P$ of node $i$ at round $t+1$, for all $i \in \mathcal{H}(P(\bar{\texttt{u}}))$.
%\noindent Informally, a query is evaluated on a proper instance if every fact derived by a send query and hashed to a certain node, is actually also in its state.
%%Intuitively, when $\mathcal{K}$ is unbounded, every query is evaluated on a proper instance.
%Nevertheless, we would like not only to know when we are properly evaluating queries, but also when a query is \emph{live}, \ie it will eventually derive something, and also when it is \emph{safe} to derive a new fact.
%
%By employing hashing as addressing model, we are able to specify when a query will eventually derive something, %rule is evaluated on a \emph{proper instance}: a query is evaluated on a proper instance if every fact hashed to a certain node, is actually also in its state.
%Intuitively, when $\mathcal{K}$ is unbounded, every query is evaluated on a proper instance.
%Nevertheless, we would like not only to known when we are evaluating proper queries, but also when a 
%\ie it is \emph{live}: %and when instead is \emph{safe} to derive a new fact.
Starting from the above consideration, we can define a query to be live if all the derived $communication$ facts (globally) satisfying the body of a query are all co-located on at least one node. A specification is then live if all the queries satisfy the liveness property.
%We consider a query to be live if each instance over $communication$ relations satisfying the body of one of the query rules, is available at least in one node.
More formally:
\begin{definition}
\label{df:live}
Let $\mathcal{T} = (\mathcal{P}, \Upsilon, \mathcal{K})$ be a shuffling transducer, and \textbf{I} an arbitrary instance over $\Upsilon_{db}$.  
Assume that a query $q_R$ in $\mathcal{P}$ exists such that $q_R$ is satisfied in the trivial configuration of $\mathcal{T}$ with input \textbf{I}. Let \textbf{q} be one of such instantiation of $q_R$, and denote with $\textbf{b}$ the set of facts over the $communication$ predicates in the body of $\textbf{q}$. Then $q_R$ is \emph{live} if for every instantiation $\textbf{q}$ and (non-trivial) configuration we have that:\vspace{-1ex}
%for every relation $P_i \in \Upsilon_{com}$ %\setminus (\Upsilon_{db} \cup \Upsilon_{sys})$ 
%in the body of $q_R$, if with $P(\bar{\texttt{u}})$ we denote a fact over $P$ available in the current transducer network global state \textbf{G}, we have that: %an instance $I^i_R$ exists on node $i \in N$ such that every fact over $R$ satisfying $H(R(\bar{u})) \cap i \neq \emptyset$ is in $I^i_R$.
\begin{equation}
\label{eq:proper_instance}
(\!\!\! \bigcap_{P(\bar{\texttt{u}}) \in \textbf{b}} \!\!\!\!\mathcal{H}(P(\bar{\texttt{u}}))\text{ }) \neq \emptyset 
%\text{ if } \!\!\!\!\!\!\!\!  \bigvee_{P \in body(q_R)}\!\!\!\!\!\!\!\! \{ P(\bar{\texttt{u}}) \mid P(\bar{\texttt{u}}) \in \textbf{G} \} \text{ satisfies }q_R \text{ then } (\!\!\!\!\!\!\!\!\bigcap_{P \in body(q_R)} \!\!\!\!\!\!\!\{\mathcal{H}(P(\bar{\texttt{u}})) \mid P(\bar{\texttt{u}}) \in \textbf{G} \}) \neq \emptyset 
\end{equation}
$\mathcal{T}$ is said to be live if the above property holds for all input instance \textbf{I} and queries in $\mathcal{P}$.
\end{definition}
By showing that a query is live we can state that a (shuffling) specification has at least the same opportunity to distributively derive a fact as the original query computed locally on a single node.

On the other hand, a query is considered safe if every node evaluates all the negated literals on an instance containing all the facts, over that literal, available in the network:
\begin{definition}
\label{df:safe}
Let $\mathcal{T} = (\mathcal{P}, \Upsilon, \mathcal{K})$ be a shuffling transducer, and $N$ a set of nodes. %and $\mathcal{N}$ the related network.
We say that evaluating a query $q_R$ in $\mathcal{P}$ is \emph{safe} if, for every fact $P(\bar{\texttt{u}})$ over a relation $P \in \Upsilon_{com}$ % \setminus (\Upsilon_{db} \cup \Upsilon_{sys})$ 
appearing negated in the body of $q_R$, we have that: %an instance $I^i_R$ exists on node $i \in N$ such that every fact over $R$ satisfying $H(R(\bar{u})) \cap i \neq \emptyset$ is in $I^i_R$.
\vspace{-1ex}
\begin{equation}
\label{eq:proper_instance}
\forall i \in N, i \in \mathcal{H}(P(\bar{\texttt{u}}))
\end{equation}
$\mathcal{T}$ is safe if every query in $\mathcal{P}$ can be safely evaluated on every input instance.
\end{definition}

\noindent If a query is not safe, the correctness of a specification can be jeopardized, as shown next.  

\begin{example}
Consider the following {\sc fo}-query:
\begin{equation*}
\begin{split}
&Q(v, w) \leftarrow R(u, v, w), \neg P(v, w).
\end{split}
\end{equation*}

\noindent with $\mathcal{D}_{in} = \{R^{(3)}, P^{(2)}\}$ and $\mathcal{D}_{out} = Q^{(2)}$.
The following shuffling {\sc fo}-transducer $\mathcal{N}$ implementing the same query is not safe:

\noindent\hrulefill
\vspace{-1mm}
\begin{equation*}
\begin{split}
\text{Schema: }&\Upsilon_{db} = \{R^{(3)}, P^{(2)}\}, \Upsilon_{mem} = \{\texttt{Ready}^{(0)}\}, \\& \Upsilon_{snd} = \{S^{(1,3)}, T^{(1,2)}\}, \Upsilon_{out} = \{Q^{(2)}\}\\
\text{Program: }&
%&R^{snd}_{snd}(\texttt{u}, \texttt{v}) \leftarrow R(\texttt{u}, \texttt{v}).\\
%&P^{snd}_{snd}(\texttt{u}, \texttt{v}) \leftarrow P(\texttt{u}, \texttt{v}).\\
%&S^{snd}_{snd}(\texttt{u}, \texttt{v}) \leftarrow S(\texttt{u}, \texttt{v}).\\
%&R_{ins}^{\prime}(\texttt{u}, \texttt{v}) \leftarrow R^{snd}(\texttt{u}, \texttt{v}).\\
%&P_{ins}^{\prime}(\texttt{u}, \texttt{v}) \leftarrow P^{snd}(\texttt{u}, \texttt{v}).\\
%&S_{ins}^{\prime}(\texttt{u}, \texttt{v}) \leftarrow S^{snd}(\texttt{u}, \texttt{v}).\\
 S_{snd}^{}(u, v, w) \leftarrow R(u, v, w).\\
&T_{snd}^{}(u, v) \leftarrow P(u, v).\\
&\texttt{Ready}_{ins}() \leftarrow \neg \texttt{Ready}().\\
%&\texttt{ready}_{ins}() \leftarrow \texttt{init}().\\
&Q(v, w)_{out} \leftarrow S(u, v, w), \neg T(v, w), \texttt{Ready}().
\end{split}
\end{equation*}
\noindent\hrulefill

\noindent Assume the following input instance \textbf{I} = \{\texttt{R(1,2,3)}, \texttt{P(2,3)}\}. Clearly the original query will have empty output over this input instance. However, Let $(N, D)$ be a non-trivial configuration, and $H$ a hash family such that constants \texttt{1} and \texttt{2} are hashed to two different nodes. We then have that $\mathcal{N}$ outputs the fact \texttt{Q(2,3)} when using this configuration, i.e., $\mathcal{N}$ is not safe.\vspace{-1ex}

\eat{
\label{ex:emptiness_partition}
Consider the emptiness query of Example \ref{ex:false_calm} computed by a hashing transducer network over a nonempty instance \textbf{I}. % \footnote{The following program is not syntactically safe since a negative term exists which is unbounded. In this specific situation, however, this is not an issue since we are computing an emptiness query.}.
%
%\noindent\hrulefill
%\vspace{-1mm}
%\begin{equation*}
%\begin{split}
%\text{Schema: }&\mathcal{T}_{db} = \{R^{(0)}\}, \mathcal{T}_{mem} = \{\texttt{ready}^{(0)}\}, \mathcal{T}_{snd} = \{S^{(0)}\}, \mathcal{T}_{out} = \{T^{(0)}\}.\\
%\text{Program:}&\\
%\end{split}
%\end{equation*}
%\vspace{-3mm}
%\begin{equation*}
%\begin{split}
%&S_{snd}(u) \leftarrow R(u).\\
%&\texttt{ready}_{ins}() \leftarrow \neg \texttt{ready}().\\
%&T_{out}() \leftarrow \neg S(u), \texttt{ready}().
%\end{split}
%\end{equation*}
%\noindent\hrulefill
%\vspace{2mm}
%
\noindent Note that $\mathcal{K}$ is unrestricted. %having the pair $(S^{(1)}, 1)$.
If we select a partitioned family $H$, every non-trivial configuration will output $T$ true, even if the initial instance is nonempty. %where $\bigcup_{\textbf{f} \in I_S}\mathcal{H}(\textbf{f}) \subset N$.
In fact, when $H$ is partitioned, at least a node exists having relation $S$ empty when \texttt{Ready} is set to true.
%With unsafe evaluations we are therefore forced to chose a proper family of hash functions if we don't want to return wrong results.
}
\end{example}

\eat{
\begin{ex}
\label{ex:filter_monotone_count}
The previous example, although correct, in some case can be inefficient, e.g., when the filter predicate is composed by one single value and therefore all the values are addressed to one single node to compute the count.  
Indeed the same problem exists also in the MapReduce implementation. 
In such case, the \emph{combiner} function can be exploited in order to obtain partial counts at the mapper side. The \emph{reduce} function then can just sum up all the partial counts.
Using our model we can easily achieve the same behavior. 
While the transducer schema remains as in Example \ref{ex:filter_count}, the program is now:

\noindent\hrulefill
\vspace{-1mm}
\small
\begin{equation*}
\begin{split}
&S_{snd}(u, count<v>) \leftarrow R(u, v), F(u).\\
%&S_{snd}(\texttt{u}, count<\texttt{v}>) \leftarrow P(\texttt{u}, \texttt{v}).\\
&T_{out}(u, sum<v>) \leftarrow S(u, v).
\end{split}
\end{equation*}
\normalsize
\noindent\hrulefill
\vspace{2mm}
\end{ex}

\begin{ex}
\label{ex:network_dependent}
If we still focus our attention on the transducer network of Example \ref{ex:filter_count}, we can see that just changing a little bit the transducer definition we can obtain a network dependent transducer.
In fact, if instead of the first term, we set the key of the synchronous relation $S^{(1, 2)}$ to be the second term, clearly the correct sum cannot be obtained for all the possible sets of nodes since the tuples are grouped over the wrong term ($v$ instead of $u$).
\end{ex}

%\begin{example}
%\label{monotonic_tc}
%In this example we show how the counting of the number of paths between nodes in a directed acyclic graph can be computed in a distributed way. 
%The transducer schema is: $\mathcal{T}_{db} = \{R^{(2)}\}$, $\mathcal{T}_{sync} = \{S^{(1,2)}\}$, $\mathcal{T}_{out} = \{T^{(3)}\}$.
%The program is:
%\begin{equation*}
%\begin{split}
%&U(\texttt{u}, \texttt{v}, mcount<\texttt{u}>) \leftarrow R(\texttt{u}, \texttt{v}).\\
%&U(\texttt{u}, \texttt{w}, mcount<(\texttt{v}, \texttt{c})>) \leftarrow U(\texttt{u}, \texttt{v}, \texttt{c}), S(\texttt{v}, \texttt{w}).\\
%&T_{out}(\texttt{u}, \texttt{v}, max<\texttt{c}> \leftarrow U(\texttt{u}, \texttt{v}, \texttt{c}).\\
%\end{split}
%\end{equation*}
%\end{example}

\begin{ex}
\label{ex:non_recursive_tc}
In this example we show how the counting of the number of paths between nodes in a directed acyclic graph can be computed in a distributed way. 
%Starting from the previous example, we may be interested in checkpointing the computation every iteration of the recursive rule.
%To achieve this we could use a memoryless Nrdatalog$^{\neg s}$ transducer instead of the previous inflationary datalog$^{\neg s}$. 

\noindent\hrulefill
\vspace{-1mm}
\small
\begin{equation*}
\begin{split}
\text{Schema: }&\mathcal{T}_{db} = \{R^{(2)}\}, \mathcal{T}_{sync} = \{S^{(1,2)}\}, \mathcal{T}_{aux} = \{P^{(3)}\}, \\&\mathcal{T}_{out} = \{T^{(3)}\}\\
\text{Program:}&\\
\end{split}
\end{equation*}
\vspace{-3mm}
\begin{equation*}
\begin{split}
&P(u, v, count<(u, v)>) \leftarrow R(u, v).\\
&T_{out}(v, u, c) \leftarrow P(u, v, c).\\
&T_{out}(v, z, sum<(c)>) \leftarrow S(u, v, c1), T(u, z, c2), c = c1 * c2.\\
&S_{snd}(v, u, c) \leftarrow T(u, v, c).
\end{split}
\end{equation*}
\normalsize
\noindent\hrulefill
\vspace{2mm}

%Note that in this case we used non-monotonic aggregates and therefore we don't need to compute the $max$ anymore.
%The number of rounds in this implementation is clearly proportional to the number of recursive iterations.
\end{ex}
}

\subsection{Parallel Computation of Queries}
\label{sec:parallel}
\label{sec:strategy}

In our effort towards the creation of a general model for parallel computation,  having modeled communication by means of hashing we are now able to %link the two properties of network independence and partition tolerance.
%In this situation, we can 
create different computational strategies by simply customizing the adopted hash functions and relation keys.
%Note that his will not intact the network independence of the transducer program since is given by assumed.
%We will then write $\mathcal{H}(i)$ to denote the partition of the input instance node $i$ is responsible for.
%Note that we are allowed to use this notation since whichever $P$ is selected, the actual partition on which a query is evaluated is determined by the hash family, not $P$.
%To remember that hashing network employ hashing partition function, we define an hashing network by a tuple in the form $(\mathcal{T}, \mathcal{T}^e, \sigma, \gamma^h, P^h)$.
%% 
%It is then reasonable to denote the pair $(\mathcal{K}, H)$ as \emph{parallelization strategy}%\footnote{The term parallelization strategy has been taken from \cite{Wolfson:1990:NPP:93597.98723}.}
%: by selecting a set of keys and a family of hash functions, the programmer chooses how computation should be allocated over the nodes of a transducer network.
While we have already seen that $\mathcal{K}$ is embedded directly into the definition of a transducer, $H$ can be added to the configuration parameters of transducers and specifications.
%added to the configuration parameters,
%since, as just described, computation allocation is totally determined by the parallelization strategy.
Henceforth, for configured hashing transducers (resp. transducer networks) we will refer to transducers (networks) determined by configurations of the forms $(N, i, H)$ (respectively  $(N, D, H)$). %where, as usual, $N$ is the set of nodes, $i$ a node identifier and $P$ a partition function.
The definition of independent specification can now be extended accordingly.
Thus, once a set of keys $\mathcal{K}$ has been fixed, a specification $\mathcal{N}$ is called \emph{strategy-independent} if, whatever $H$ is chosen, %for every selected hash function in $H$, 
$\mathcal{N}$ computes the same final result.
%Note that every strategy-independent specification is tolerant to the fact that $H$ is partitioned or not. 
%On the other hand, a specification can be strategy-dependent just because of the fact that $H$ is partitioned. %while for every choice of 
%We will explore in detail in the next Sections the class of queries which are strategy-dependent.
%non-partitioned $H$, the specification is strategy independent. 
%We name this type of specification as \emph{partition intolerant}. 
%Finally, a specification is called \emph{independent} if both network, time, partition and strategy independent.
%Every unbounded specification is clearly strategy independent. %specification is clearly tolerant to the fact that $\mathcal{H}$ is partitioned or not.
%On the other hand, an example of a specification which is strategy dependent is one in which it is not tolerated that $\mathcal{H}$ is partitioned.
%because of the fact that $\mathcal{H}$ is partitioned, while for every chose of non-partitioned $\mathcal{H}$, the specification is strategy independent.
%We name this type of specifications as \emph{partition biased}.
Finally, a hashing specification is called \emph{independent} if it is altogether network-distribution-strategy-independent. 
The definition of convergent class of systems (Definition \ref{df:convergence}) can naturally be generalized over (shuffling) hashing specifications.

In the next Sections we will explore in detail both the connections existing between class of queries and type of specifications, and which class is more expensive to compute with respect to others (e.g., which query can only be computed by an unrestricted specification).

In Section \ref{sec:synchronous_transducer_network} we have seen that, if the transducer network is broadcasting, a large class of queries can be distributively computed. %by an oblivious specification.
%Indeed, broadcasting and disseminating networks are useless in practice since every transducer ends up to have all the input instance in its state before the query can be actually computed. 
%We start by defining what it means for a query expressed in a language $\mathcal{L}$ to be \emph{distributively computable}.
%In this section we will %describe how employing addressing transducers we are able to compute queries in \emph{parallel}.
%One question then arise naturally: is it possible to make a query distributively computable without employing a broadcasting transducer?
%From Theorem \ref{te:hashing_tolerance} we already have a partial answer: if the query is monotonic, an oblivious monotonic hashing network can be generated which is able to distributively compute it.
%But what about a generic (non-monotonic) query?
%In the following we will therefore %introduce a new type of transducer network, and we will 
%enhance the techniques developed in Section \ref{sec:sync_computability} to generate, starting from a query expressed in $\mathcal{L}$, a shuffling transducer network implemented in $\mathcal{L}^\prime$ and able to \emph{parallelly} compute it.  
We say that a query $\mathcal{Q}$ is \emph{parallelly computable} if a hashing specification exists such that all the possible runs in $\mathbfcal{S}^{\emph{bsp}}_{\mathcal{N}}(\textbf{I})$ compute the same query $\mathcal{Q}(\textbf{I})$, whichever initial instances $\textbf{I}$ is given.
%The reader should not be surprised by the following result:

\begin{proposition}
\label{pr:generic_parallel_computable}
Let $\mathcal{L}$ be a query language.
Every query that is expressible in $\mathcal{L}$, and that can be distributively computed by a broadcasting $\mathcal{L}$-transducer network, can also be parallelly computed by a hashing $\mathcal{L}$-transducer network.
\end{proposition}\vspace{-2ex}
\begin{proof}
By definition every hashing transducer with unrestricted $\mathcal{K}$ emulates a broadcasting one.
\end{proof}

\noindent It is now a straightforward exercise to show that the expressiveness result of Lemma \ref{lm:computable} applies also to hashing specifications.
%It is now a straightforward exercise to show that expressiveness results of the previous sections naturally apply also for hashing networks. 
\kWlog we will hereafter call  \emph{broadcasting} every hashing transducer network where $\mathcal{K}$ is unrestricted. % and $H$ is non-partitioned.
%As stated before, we are interested in computing queries without employing broadcasting specifications: \ie queries computable by an independent shuffling specification. %\ie specifications which are independent and therefore are also able to tolerates partitioned strategies.
%This means queries computable by an independent specification. 
%\ie a bounded $\mathcal{K}^\prime$ exists which is correct \footnote{Note that this doesn't mean that under specific input instances a fact cannot be emitted to every node, as for instance Example \ref{ex:hashing_partition_depend} shows.}.
%bounded $\mathcal{K}^\prime$ exists which is correct \footnote{Note that this doesn't mean that under specific input instances a fact cannot be emitted to every node, as for instance Example \ref{ex:hashing_partition_depend} shows.}, or the specification is independent -- it tolerated partitioned strategies.
\vspace{-1ex}

\paragraph{Correct Specifications.}
In Section~\ref{sec:strategy} we have seen that the set of keys $\mathcal{K}$ can be used to parametrize a transducer.
In this way,  multiple specifications can be produced by selecting different sets of keys. Let $\mathbfcal{N}$ be the class of specifications that can be generated by changing $\mathcal{K}$ in a specification $\cal{N}$. Intuitively, a wrong selection of the keys can result in a wrong specification: consider Example \ref{ex:hashing} and assume that we choose the second term as key for both $S$ and $U$; we can then incur in the situation in which a tuple is not derived because the joining facts are issued to two different nodes. 
Yet, we can define a subclass of $\mathbfcal{N}$ wherein each specification is consistent with the broadcasting version.
By Proposition~\ref{pr:generic_parallel_computable}, in fact, a ``correct" broadcasting specification always exists and it does not depend on the chosen set of keys: it simply broadcasts every emitted fact!
Specifically, let $\mathbfcal{S}^{\emph{bsp}}_{\mathcal{N}}$ be the convergent class derived from the broadcasting specification $\mathcal{N}$. %where every configuration admits a non-partitioned hash family.
For every input instance \textbf{I}, we can start to add to $\mathbfcal{S}^{\emph{bsp}}_{\mathcal{N}}$ all the runs $\mathcal{S}^{\emph{bsp}}_{\mathcal{N}^\prime}(N, D, H, \textbf{I})$ such that $\mathcal{N}^\prime \in \mathbfcal{N}$ is a convergent specification, %$(N, D, H)$ is a non-partitioned configuration, 
and $\mathcal{S}^{\emph{bsp}}_{\mathcal{N}}(N, D, H, \textbf{I})$ and $\mathcal{S}^{\emph{bsp}}_{\mathcal{N}^\prime}(N, D, H, \textbf{I})$ are eventually consistent.
Denote with $\mathbfcal{S}^{\emph{bsp}}_{\mathbfcal{CN}}$ the class which is maximal with the above property, and where with $\mathbfcal{CN} \subseteq \mathbfcal{N}$ we identify the \emph{correct} class of specifications.
We have the following:

%
%\begin{df}[Consistent Specifications]
%Let $\mathcal{N}$ be a convergent specification. %where $\mathcal{K}$ is unbounded.
%Assume that the specification $\mathcal{N} = (\mathcal{T}, \mathcal{T}^e, \sigma, \gamma^h)$ in $\overline{\cal{N}}$ is convergent.
%Let denote with $\overline{\cal{N}}$ the class of specifications arising from $\mathcal{N}$.
%A subclass of $\overline{\cal{N}}$ is said to be \emph{consistent}, and we write $\mathbfcal{CN}$, if for all $\mathcal{N}^\prime \in \mathbfcal{CN}$, a configuration $(N, t, P, H)$ and instance \textbf{I} exist such that $\mathcal{S}^{det}_{\mathcal{N}^\prime_{N, t, P, Q}, \textbf{I}}$ and $\mathcal{S}^{det}_{\mathcal{N}_{N, t, P, Q}, \textbf{I}}$ are (eventually) consistent.
%\end{df}
%
%\noindent Among the consistent specifications in $\overline{\cal{N}}$ we have that a further and more restricted subclass exists in which 
%, \ieeach specification has a \emph{correct key-set} $\mathcal{K}$.

\begin{definition}
\label{df:correct}
Let $\mathcal{N} = (\mathcal{T}, \mathcal{T}^e, \gamma)$ be a hashing transducer network with $\mathcal{T} = (\mathcal{P}, \Upsilon, \mathcal{K})$, and assume that a nonempty correct class of specifications $\mathbfcal{CN}$ exists.
A set of keys $\mathcal{K}$ (specification $\mathcal{N}$) is said to be \emph{correct} if $\mathcal{N} \in \mathbfcal{CN}$.
\end{definition}

\noindent %Every specification having $\mathcal{K}$ correct is intuitively also independent. 
%We refer to the sub-class of $\mathbfcal{N}$ composed by specifications with correct key-sets as \emph{correct}, and write $\mathbfcal{CN}$.
Note that the specifications in $\mathbfcal{CN}$ are independent by construction. %but might be strategy-dependent. This is because certain queries (i.e., non-monotonic unconnected queries) are not partition-tolerant, as we will see next.

\eat{
\subsubsection{Queries Computable by an Independent Specification}
\label{sec:strategy-independent}

Overall we are interested in computing queries without employing broadcasting specifications.
Additionally, we are interested in the queries computable by an independent specification.
We will say that a query is \emph{parallelly computable} if it is distributively computable by an independent specification. %$\mathcal{N} \in \overline{\cal{N}}$.
In other words, the parallel computable queries are the ones that tolerate partitioned hash families. 
But which class of queries can be parallelly computed?
Let us start with monotonic queries. %because of the CALM conjecture.
%Unfortunately not every monotonic query is embarrassingly parallelly computable as the next example shows.

\begin{lemma}
\label{lm:parallel_monotonic}
%Let $\mathcal{L} \subseteq$ datalog.
Every monotonic query expressible in $\mathcal{L} $ can be parallelly computed by an inflationary and oblivious $\mathcal{L}$-transducer network.
\end{lemma}
\begin{proof}
Let $\mathcal{Q}$ be a monotonic query, and \textbf{I} an instance over $\mathcal{D}_{in}$. 
A specification $\mathcal{N}$ can be created by shuffling the input instance and then applying the query, %from the second round over the entire instance, 
as described in the first case of Lemma \ref{lm:computable}.
$\mathcal{K}$ is unrestricted, $H$ is non-partitioned, and $\mathcal{N}_{N, D, H}$ distributively computes $\mathcal{Q}$.
We have to show that $\mathcal{N}$ is actually independent and hence $\mathcal{Q}$ is parallelly computed by $\mathcal{N}$.
Assume $N$ to be not trivial.
Consider the same specification but with a new configuration $(N, D, H^\prime)$ where $ H^\prime$ is partitioned.
%$(\mathcal{N}_{N, t, P, \mathcal{H}^\prime}, \textbf{I})$ is now non-broadcasting.
The liveness property still holds also with the new configuration: %$\bigcap_{R_i \in \mathbf{B}} \mathcal{H}(R_i) \neq \emptyset$ 
$\mathcal{K}$ is unrestricted and every query is evaluated over the full initial instance starting from the second round by at least a node, whichever $H^\prime$ is chosen.
Every fact in $\mathcal{Q}(\textbf{I})$ is then also in $\mathcal{N}_{N, D, H^\prime}(\textbf{I})$ since the query is monotonic.
It remains to show that no wrong fact will be derived by $\mathcal{N}_{N, D, H^\prime}$.
We have that the safety property also holds for $\mathcal{N}_{N, D, H^\prime}$ -- although we are using a partitioned hashing family -- again because the query is monotonic. 
In particular, note that every passive node will always have all the $communication$ predicates empty. 
Hence, when $\mathcal{Q}$ is applied, no fact will be derived since the query is monotonic.
As a consequence we have that passive nodes do not contribute in any way to the computation of the final query, although they participate in the initial shuffling. 
\eat{
\begin{proof}[(sketch)]
Let $\mathcal{Q}_{out}$ be a monotonic query, and \textbf{I} an instance. %over $\mathcal{D}_{in}$. 
A specification $\mathcal{N}$ can be created by shuffling the input $database$ instance and then applying the query. %from the second round over the entire instance, 
%as described in the first case of Lemma \ref{lm:computable}.
$\mathcal{K}$ is unrestricted, $H$ is non-partitioned, and $\mathcal{N}_{N, P, H}$ distributively computes $\mathcal{Q}_{out}$.
We have to show that $\mathcal{N}$ is actually independent. %and hence $\mathcal{Q}_{out}$ is parallelly computed by $\mathcal{N}$.
Assume $N$ to be non trivial.
Consider now the same specification and configuration but with a partitioned hash family  $H^\prime$.
%$(\mathcal{N}_{N, t, P, \mathcal{H}^\prime}, \textbf{I})$ is now non-broadcasting.
The liveness property still holds also with the new configuration. %$\bigcap_{R_i \in \mathbf{B}} \mathcal{H}(R_i) \neq \emptyset$ 
%$\mathcal{K}$ is unbounded and every query is evaluated over the full initial instance starting from the second round by at least a node, whichever $H^\prime$ is chosen.
%Every fact $\textbf{f}$ in $\mathcal{Q}_{out}(\textbf{I})$ is then also in $\mathcal{N}_{N, P, H^\prime}(\textbf{I})$ since the query is monotonic.
It remains to show that no wrong fact will be derived by $\mathcal{N}_{N, P, H^\prime}$.
We have that the safety property doesn't hold for $\mathcal{N}_{N, P, H^\prime}$ -- although note that it does hold for $\mathcal{N}_{N, P, H}$ -- however we have that every passive node will always have all the predicates in $\Upsilon_{com}$ empty. Hence, when $\mathcal{Q}_{out}$ is applied, no fact will be derived since the query is monotonic.
%As a consequence we have that passive nodes do not contribute in any way to the computation of the final query, although they participate in the initial shuffling. 
}
\end{proof}

\noindent It turns out that monotonic queries are not the only parallelly computable queries.
In fact, also a class of non-monotonic queries is parallelly computable: \emph{connected queries} \cite{Guessarian90}.
Informally, a query is connected if every relation in a rule-body is connected through a join-path with every other relation composing the same rule-body.
%the issue with Examples \ref{ex:emptiness_partition}, \ref{ex:hashing_partition_depend} is that an atom exists -- $T$ in both the examples -- whose terms are unrelated with any other term of the query.
%We make this more precise:
%\mi{Le due definizioni che seguono mi lasciano ancora un pochino perplessa. Innanzitutto, siamo sicuri che nella prima def dobbiamo dire \textbf{relation} invece che \textbf{literal }? Sono sicura che ne abbiamo parlato ma non ricordo cosa avevamo deciso e perche'. Per entrambe le def,  guarda se non sono magari piu' chiare le mie, le precedenti le lascio commentate. In ogni caso, anche con le modifiche che ho fatto alla seconda, mi sembra che non copra il caso in cui la connessione avviene attraverso  letterali nella head... o no? Ad esempio questa qui secondo te  dovrebbe essere connessa? non ne sono sicura, cmq la def non mi pare la copra:
%\begin{equation*}
%\begin{split}
%&T(u, v) \leftarrow R(u, v), T(v,w).\\
%&T(u, v) \leftarrow S(u, v), W(u).
%\end{split}
%\end{equation*}
%}
\begin{definition}
Let $\textit{body}(q_R)$ be the conjunction of literals defining the body of a rule $q_R$. %and $R(\bar{u})$ the head atom.
We say that two different literals $R_i(\bar{u}_i)$, $R_j(\bar{u}_j) \in \textit{body}(q_R)$ are \emph{connected} in $q_R$ if either: 
\begin{itemize}
\item $\bar{u}_i \cap \bar{u}_j \neq \emptyset$% where with $\bar{v}^\prime$, $\bar{v}^{\prime\prime}$ we denote the variables in $\bar{u}^\prime$, $\bar{u}^{\prime\prime}$ not appearing in any aggregate function
; or
\item a third literal $R_k(\bar{u}_k) \in q_R$ different from $R_i(\bar{u}_i)$ and $R_j(\bar{u}_j)$ exists such that $R_i(\bar{u}_i)$ is connected with $R_k(\bar{u}_k)$, and $R_k(\bar{u}_k)$ is connected with $R_j(\bar{u}_j)$.
\end{itemize}
Two relations $R_i$ and $R_j$ are said to be connected in $q_R$  if there are two literals $R_i(\bar{u}_i)$ and $R_j(\bar{u}_j)$ that are connected in $q_R$.
\end{definition}

%\begin{definition}
%Let $\textit{body}(q_R)$ be the conjunction of literals defining the body of a rule $q_R$. %and $R(\bar{u})$ the head atom.
%We say that two different relations $R_i(\bar{u}_i)$, $R_j(\bar{u}_j) \in \textit{body}(q_R)$ are \emph{connected} in $q_R$ if either: 
%\begin{itemize}
%\item $\bar{u}_i \cap \bar{u}_j \neq \emptyset$% where with $\bar{v}^\prime$, $\bar{v}^{\prime\prime}$ we denote the variables in $\bar{u}^\prime$, $\bar{u}^{\prime\prime}$ not appearing in any aggregate function
%; or
%\item a third relation $R_k \in q_R$ different from $R_i$, $R_j$ exists such that $R_i$ is connected with $R_k$, and $R_k$ is connected with $R_j$.
%\end{itemize}
%\end{definition}

\begin{definition}
\label{df:connected}
We say that a $\mathcal{L}$-query $\mathcal{Q}$ is \emph{connected} if there is no rule $q_R \in \mathcal{Q}$ whose body contains either a nullary relation or a relation $R_i $  which is not connected with any other relation $R_j \in body(q_R)$. %where $i \neq j$.
%\end{itemize}
\end{definition}

%\begin{definition}
%\label{df:connected}
%We say that a $\mathcal{L}$-query $\mathcal{Q}$ is \emph{connected} if for every rule $q_R \in \mathcal{Q}$, no body relation $R_i \in q_R$  exists which is either 
%%\vspace{-1mm}
%%\begin{itemize} 
%nullary or
%not connected with any other relation $R_j \in body(q_R)$. %where $i \neq j$.
%%\end{itemize}
%\end{definition}

\stitle{Remarks}: $(i)$ Literals can be either positive or negative predicates~\footnote{Note that non-nullary negative literals are connected by definition since we are only considering safe queries.}; $(ii)$ every %non-aggregate 
positive query composed by a single rule containing just one non-nullary body atom is connected by definition; and $(iii)$ %\cite{AmelootKNZ13} shows that connected queries capture the fragment of queries that \emph{distributes over connected component}s; and $(iv)$ 
every non-nullary unconnected query is an \emph{existential query} \cite{RamakrishnanBK88}.

\begin{example}
\label{ex:connected}
The following {\sc datalog}$^\neg$ query returns all the nodes being source of a path of length three, if no path of length four exists.
\begin{equation*}
\begin{split}
&T(u) \leftarrow E(u, v), E(v, w), E(w,x).\\
&F() \leftarrow E(u, v), E(v, w), E(w,x), E(x,y)\\
&Q(u) \leftarrow T(u), \neg F().
\end{split}
\end{equation*}
While the first two rules are connected, the third rule is not.  The query is therefore not connected.
\eat{The following is an example of unconnected {\sc sql}-query.
%\vspace{15mm}

\noindent\hrulefill
\vspace{-1mm}
\begin{equation*}
\setlength{\jot}{0pt} 
\begin{split}
&\texttt{SELECT attribute-1}\\
&\texttt{FROM table-1}\\ 
&\texttt{WHERE EXISTS (SELECT * FROM table-2).}
\end{split}
\end{equation*}

\noindent\hrulefill
\vspace{2mm}
 
\noindent Note that the inner and the outer sub-queries are unrelated.
}
\end{example} 

%while every aggregation query having an empty group-by set of terms is always not connected \footnote{To see why this is the case, note that every aggregation query is normally rewritten by pushing the aggregation function into the body. We then have that if the group-by term is empty, the aggregation function is not connected with the remaining part of the boconnecteddy.}.
%By considering connected queries we are basically ruling out all the queries in the form of Examples \ref{ex:false_calm}. %\ref{ex:hashing_partition_depend}.
%Similarly to semi-positive queries in which the only allowed negative literals are extensional, we call \emph{semi-connected queries} the one in which the only un-connected predicates are extensional. 
%For what concern the latter case, note that every aggregation rule is rewritten as in eq. \ref{} and hence the predicate readtable is never connected to any predicate if the set of group-by terms is empty.
%connected queries are the only types of 
Non-monotonic connected queries can be evaluated independently of the choice of $H$. 
Conversely, differently than Lemma~\ref{lm:computable}, non-monotonic not connected queries are not parallelly computed because not tolerant to  partitioned hash families.
%even if they are not safe.
%In addition, for connected queries we have that the safety property collapse.
%It is in fact enough to demonstrate that every query is evaluated on a proper instance to show that a connected query is safe.
%To see why this is the case, every negative literal in a connected query shares a variable term with at least another literal.
%A proper set of keys must then exists such that every tuple over the negative literal potentially joining with the other literal are addressed to the same node.
%If every query is evaluated on a proper instance, and a new tuple is derived, this means that globally the instance of the negative literal does not contain a tuple falsifying the query.
%Such specifications are he independent without resorting broadcasting.
%We have the following lemma:

\begin{lemma}
\label{lm:parallel_non_monotonic}
Let $\mathcal{L}$ be a language such that $\mathcal{L} \subseteq$ {\sc datalog}$^{\neg}$. % and Nrdatalog$^{\neg sp} \subseteq \mathcal{L}^\prime$ .
Every connected query expressible in $\mathcal{L}$ can be parallelly computed by an inflationary and oblivious $\mathcal{L}$-transducer network.
\end{lemma}
\begin{proof}
We proceed as in the proof of Lemma \ref{lm:parallel_monotonic}.
%Let $\Pi_1, \ldots, \Pi_n$ be a stratification of the rules of the input query 
Let $\mathcal{Q}$ be the input query, and \textbf{I} an instance over $\mathcal{D}_{in}$. 
A transducer network $\mathcal{N}_{N, D, H}$ can be created by shuffling the input instance and then applying the query from the second round over the entire instance.
$\mathcal{K}$ is unrestricted, $H$ is non-partitioned, and $\mathcal{N}_{N, D, H}$ distributively computes $\mathcal{Q}$.
%We have to show that $\mathcal{N}$ is actually independent and hence $\mathcal{Q}_{out}$ is an embarrassingly parnoteallel query.
Assume $N$ to be not trivial.
We are going to show that $\mathcal{N}$ is actually independent.
Consider the same specification but with a new configuration $(N, D, H^\prime)$ where $H^\prime$ is partitioned.
%$(\mathcal{N}_{N, t, P, \mathcal{H}^\prime}, \textbf{I})$ is now non-broadcasting.
Since all rules are connected, the liveness property  holds also with the new configuration: $\mathcal{K}$ is unrestricted and every query is evaluated over the full initial instance starting from the second round by at least a node, whichever $H^\prime$.
%Every fact $\textbf{f}$ in $\mathcal{Q}_{out}(\textbf{I})$ will hence also be in $\mathcal{N}_{N, t, P, \mathcal{H}^\prime}(\textbf{I})$.
It remains to show that no wrong fact will be derived by $\mathcal{N}_{N, D, H^\prime}(\textbf{I})$.
In this case we have that the safety property doesn't hold for $\mathcal{N}_{N, D, H^\prime}$.
Note, however, that every passive node at the second round will have all the $communication$ predicates empty.
Hence, when the query is applied, no fact will be derived because a new fact, to be derived, must be at least unary (the query is connected) and a constant must exist in the instance of a positive literal (only syntactically safe queries are considered).
But this is impossible since every $communication$ relation is initially empty.
%Again, if we restrict our attention to active nodes, all rules are safe.
%We can then conclude that 
We can then conclude that $\mathcal{N}$ parallelly computes $\mathcal{Q}$.
\end{proof}

\stitle{Remark}:
From the proof of Lemma~\ref{lm:parallel_non_monotonic} we can conclude that non-monotonic connected queries are parallelly computable because, although not safe in general, they are safe on active nodes while passive nodes do not derive any fact.
%In summary, the following, more restricted version of the safety condition of eq. (\ref{eq:proper_instance}) holds for connected queries:
%\begin{equation}
%\label{eq:safe-connected}
%\forall i \in N_H, i \in \mathcal{H}(\textbf{f})
%\end{equation}
\vspace{1ex}

Connected queries have an interesting semantic property: they \emph{distribute over components} \cite{AmelootKNZ14}.
An instance $\textbf{J}$ is said to be \emph{connected} if whichever pair of constants \texttt{a}, \texttt{b} $\in adom(\textbf{J})$, a chain of facts \textbf{f}$_1$, \ldots, \textbf{f}$_n$ exists such that $\texttt{a} \in adom(\textbf{f}_1)$, $\texttt{b} \in adom(\textbf{f}_n)$, and for any pair of consecutive fact $\textbf{f}_i$, $\textbf{f}_{i+1}$ in the chain with $0 \leq i < n$, $adom(\textbf{f}_i) \cap adom(\textbf{f}_{i+1}) \neq \emptyset$.
Now, if $\textbf{I}$ is an instance, $\textbf{J}$ is a \emph{component} of $\textbf{I}$ if $(i)$ $\textbf{J} \subseteq \textbf{I}$, $(ii)$ $\textbf{J}$ is nonempty, and $(iii)$ $\textbf{J}$ is connected and maximal with this property in $\textbf{I}$.
Finally, if with $co(\textbf{I})$ we denote the components of $\textbf{I}$, a query $\mathcal{Q}$ is said to \emph{distribute over components} if $\mathcal{Q}(\textbf{I}) = \bigcup_{\textbf{J} \in co(\textbf{I})} \mathcal{Q}(\textbf{J})$ for all $\textbf{I}$.
\vspace{-0.5ex}

\begin{proposition}[Ameloot at al. 2014]%[\cite{AmelootKNZ14}]
\label{pr:connected1}
Let $\mathcal{L} \subseteq$ {\sc datalog}$^{\neg}$.
Every connected $\mathcal{L}$-query distributes over components.
\end{proposition}

\begin{proposition}[Ameloot at al. 2015]%[\cite{AmelootKNZ15}]
\label{pr:connected2}
Every query computable by a {\sc datalog}$^\neg$-query that distributes
over components can be computed by a connected {\sc datalog}$^\neg$-query.
\end{proposition}

\begin{example}
Consider again Example~\ref{ex:connected}. 
The query is not connected and in fact it does not distribute over components.
To see why this is the case, assume that the input instance is composed by two components~\footnote{To be more formal, the Gaifman graph of the input instance has two components.}. %: one containing a path of length three but no path of length four, and one instead containing a path of length four.
The query does not distribute because the result on each component might depend on the presence of a path of length four on the other component.

Conversely, assume the query:
\begin{equation*}
\begin{split}
&T(u) \leftarrow E(u, v), E(v, w), E(w,x).\\
&F(u) \leftarrow E(u, v), E(v, w), E(w,x), E(x,y)\\
&Q() \leftarrow T(u), \neg F(u).
\end{split}
\end{equation*}
returning true if a component exist having a path of length three, but no path of length four starting from the same source node.
The query is connected and distributes over components.
\end{example} 
\vspace{-0.5ex}

\stitle{Remark}: Differently from \cite{AmelootKNZ15}, but according to \cite{AmelootKNZ14}, we are considering as not connected \emph{all} the queries having nullary relations in the body. 
In \cite{AmelootKNZ15}, a specific type of nullary relations is identified which can be safely used in rule bodies while maintaining connectedness.
Such nullary relations are however copied to each connected component, which in our parallel settings means  broadcasting all nullary relations, which in turn means that, for partitioned hash families, if negation is applied over a nullary relation the query is no longer parallelly computable (because unsafe).
We therefore restrict our attention to nonempty instances (having at least one component) and connected queries without nullary atoms in the body.
}
\vspace{-1ex}

\subsection{Queries Computable by a Restricted Specification}
Among the parallelly computable queries, the most interesting ones from a parallel processing perspective are the ones for which a restricted specification (\ie having a restricted set of keys), parallelly computing them, exists.
We denote this kind of queries as \emph{restricted}. %\footnote{The term embarrassingly parallel comes from the field of parallel computing, wherein it refers to the class of problems which can be parallelly solved by a set of tasks without resorting communication \cite{Foster:1995:DBP:527029}.
%In the following we will see that also in our case embarrassingly parallelly computable queries are communication-free.}.
Intuitively, restricted queries are the ones which can be parallelly computed without having to resort to broadcasting. %any fact to all the nodes in $N_H$.
%In other words, these are the queries parallelly computable by non-broadcasting networks. 
%Intuitively, embarrassing parallel queries are the ones which can be computed without resorting broadcasting.
 %which is independent. 
Note that the reader should not be deluded into believing that every monotonic query is trivially a restricted one. %as the next example shows. 

\begin{example}
\label{ex:not_connected}
Assume two relations $R^{(2)}$ and $T^{(1)}$, and the following query $\mathcal{Q}$ returning the full $R$-instance if $T$ is nonempty. 
\vspace{-1ex}
\begin{equation*}
%\begin{split}
Q(u, v) \leftarrow R(u, v), T(\_).
%\end{split}
\end{equation*}
\normalsize

\noindent The query is monotonic.
Let $\mathcal{T}$ be the following broadcasting {\sc ucq}-transducer computing $\mathcal{Q}$.

\noindent\hrulefill
\vspace{-1mm}
\begin{equation*}
\begin{split}
\text{Schema: }&\Upsilon_{db} = \{R^{(2)}, T^{(1)}\}, \Upsilon_{com} = \{S^{(0, 2)}, U^{(0, 1)}\}, \Upsilon_{out} = \{Q^{(2)}\}\\
\text{Program: }
&S_{snd}(u, v) \leftarrow R(u, v).\\
&U_{snd}(u) \leftarrow T(u).\\
&Q_{out}(u, v) \leftarrow S(u, v), U(\_).
\end{split}
\end{equation*}

\vspace{-1ex}
\noindent\hrulefill
\vspace{1mm}

\noindent Assume now a restricted set of keys $\mathcal{K}$. 
We have that, whichever (restricted) $\mathcal{K}$ we chose, the related specification might no longer be convergent.
To see why this is the case, consider an initial instance \textbf{I} and $\mathcal{K}$ maximal, \ie every term of every relation is key\footnote{
We chose $\mathcal{K}$ to be maximal because if we fail to generate a convergent specification for the maximal case, even more so specifications where $\mathcal{K}$ is less than maximal will fail.}.
Assume \textbf{I} such that $adom(I_{R}) \supset adom(I_T)$, and a configuration in which $N$ is large.
%Moreover, assume that a connected component $J_R$ exists in $co(I_{R})$, such that $adom(J_R) \cap adom(I_T) = \emptyset$.
In this situation, it may well happen that a nonempty set of facts in $I_R$ is hashed to a certain node $i$, while no fact over $T$ is hashed to $i$.
Hence no tuple emitted to $i$ will ever appear in the output, although they do appear in $\mathcal{Q}(\textbf{I})$.
Thus this transducer is not convergent.
\end{example}

\eat{\noindent We then have the following Theorem:
\vspace{-1mm}
\begin{te}
\label{te:embarassing_monotonic_unconnected}
Let $\mathcal{L} \subseteq$ datalog.
A class of monotonic $\mathcal{L}$-queries exists which is parallelly computable but not hashing.
%Not every monotonic query expressible in $\mathcal{L}$ is not broadcasting. %can be parallelly computed by an inflationary and oblivious $\mathcal{L}^\prime$-transducer network.
\end{te}
\vspace{-3mm}}

\subsubsection{Connected Queries}
A class of monotonic queries exists which is restricted: \emph{connected queries} \cite{Guessarian90}.
Informally, a query is connected if every relation in a rule-body is connected through a join-path with every other relation composing the same rule-body.
%the issue with Examples \ref{ex:emptiness_partition}, \ref{ex:hashing_partition_depend} is that an atom exists -- $T$ in both the examples -- whose terms are unrelated with any other term of the query.
%We make this more precise:
%\mi{Le due definizioni che seguono mi lasciano ancora un pochino perplessa. Innanzitutto, siamo sicuri che nella prima def dobbiamo dire \textbf{relation} invece che \textbf{literal }? Sono sicura che ne abbiamo parlato ma non ricordo cosa avevamo deciso e perche'. Per entrambe le def,  guarda se non sono magari piu' chiare le mie, le precedenti le lascio commentate. In ogni caso, anche con le modifiche che ho fatto alla seconda, mi sembra che non copra il caso in cui la connessione avviene attraverso  letterali nella head... o no? Ad esempio questa qui secondo te  dovrebbe essere connessa? non ne sono sicura, cmq la def non mi pare la copra:
%\begin{equation*}
%\begin{split}
%&T(u, v) \leftarrow R(u, v), T(v,w).\\
%&T(u, v) \leftarrow S(u, v), W(u).
%\end{split}
%\end{equation*}
%}
\begin{definition}
Let $\textit{body}(q_R)$ be the conjunction of literals defining the body of a rule $q_R$. %and $R(\bar{u})$ the head atom.
We say that two different literals $R_i(\bar{u}_i)$, $R_j(\bar{u}_j) \in \textit{body}(q_R)$ are \emph{connected} in $q_R$ if either: 
\begin{itemize}
\item $\bar{u}_i \cap \bar{u}_j \neq \emptyset$% where with $\bar{v}^\prime$, $\bar{v}^{\prime\prime}$ we denote the variables in $\bar{u}^\prime$, $\bar{u}^{\prime\prime}$ not appearing in any aggregate function
; or
\item a third literal $R_k(\bar{u}_k) \in q_R$ different from $R_i(\bar{u}_i)$ and $R_j(\bar{u}_j)$ exists such that $R_i(\bar{u}_i)$ is connected with $R_k(\bar{u}_k)$, and $R_k(\bar{u}_k)$ is connected with $R_j(\bar{u}_j)$.
\end{itemize}
Two relations $R_i$ and $R_j$ are said to be connected in $q_R$  if there are two literals $R_i(\bar{u}_i)$ and $R_j(\bar{u}_j)$ that are connected in $q_R$.
\end{definition}

%\begin{definition}
%Let $\textit{body}(q_R)$ be the conjunction of literals defining the body of a rule $q_R$. %and $R(\bar{u})$ the head atom.
%We say that two different relations $R_i(\bar{u}_i)$, $R_j(\bar{u}_j) \in \textit{body}(q_R)$ are \emph{connected} in $q_R$ if either: 
%\begin{itemize}
%\item $\bar{u}_i \cap \bar{u}_j \neq \emptyset$% where with $\bar{v}^\prime$, $\bar{v}^{\prime\prime}$ we denote the variables in $\bar{u}^\prime$, $\bar{u}^{\prime\prime}$ not appearing in any aggregate function
%; or
%\item a third relation $R_k \in q_R$ different from $R_i$, $R_j$ exists such that $R_i$ is connected with $R_k$, and $R_k$ is connected with $R_j$.
%\end{itemize}
%\end{definition}

\begin{definition}
\label{df:connected}
We say that a $\mathcal{L}$-query $\mathcal{Q}$ is \emph{connected} if there is no rule $q_R \in \mathcal{Q}$ whose body contains either a nullary relation or a relation $R_i $  which is not connected with any other relation $R_j \in body(q_R)$. %where $i \neq j$.
%\end{itemize}
\end{definition}

%\begin{definition}
%\label{df:connected}
%We say that a $\mathcal{L}$-query $\mathcal{Q}$ is \emph{connected} if for every rule $q_R \in \mathcal{Q}$, no body relation $R_i \in q_R$  exists which is either 
%%\vspace{-1mm}
%%\begin{itemize} 
%nullary or
%not connected with any other relation $R_j \in body(q_R)$. %where $i \neq j$.
%%\end{itemize}
%\end{definition}

\stitle{Remarks}: $(i)$ Literals can be either positive or negative predicates~\footnote{Note that non-nullary negative literals are connected by definition since we are only considering safe queries.}; $(ii)$ every %non-aggregate 
positive query composed by a single rule containing just one non-nullary body atom is connected by definition; and $(iii)$ %\cite{AmelootKNZ13} shows that connected queries capture the fragment of queries that \emph{distributes over connected component}s; and $(iv)$ 
every non-nullary unconnected query is an \emph{existential query} \cite{RamakrishnanBK88}.

\begin{example}
\label{ex:connected}
The following {\sc datalog} query returns all the nodes being source of a triangle, if a non-cyclic path of length four exists. In this query, the build-it predicate $\neq$ is used (in infix notation) to express inequality between variable instantiations.
\begin{equation*}
\begin{split}
&T(u) \leftarrow E(u, v), E(v, w), E(w,u).\\
&F() \leftarrow E(u, v), E(v, w), E(w,x), E(x,y), x \neq u, y \neq u.\\
&Q(u) \leftarrow T(u), F().
\end{split}
\end{equation*}
While the first two rules are connected, the third rule is not.  The query is therefore not connected.
\eat{The following is an example of unconnected {\sc sql}-query.
%\vspace{15mm}

\noindent\hrulefill
\vspace{-1mm}
\begin{equation*}
\setlength{\jot}{0pt} 
\begin{split}
&\texttt{SELECT attribute-1}\\
&\texttt{FROM table-1}\\ 
&\texttt{WHERE EXISTS (SELECT * FROM table-2).}
\end{split}
\end{equation*}

\noindent\hrulefill
\vspace{2mm}
 
\noindent Note that the inner and the outer sub-queries are unrelated.
}
\end{example} 
\begin{proposition}
\label{lm:embarassing_monotonic}
Let $\mathcal{L} \subseteq$ {\sc datalog} be a language.
For every connected (monotonic) query expressible in $\mathcal{L}$, an equivalent one exists that is restricted. %can be parallelly computed by an inflationary and oblivious $\mathcal{L}^\prime$-transducer network.
\end{proposition}\vspace{-1ex}
\begin{proof}
From Proposition~\ref{pr:generic_parallel_computable} we know that a specification $\mathcal{N}$ exists parallelly computing every monotonic query $\mathcal{Q}$ for all input instances \textbf{I}.
Starting from $\mathcal{N}$ we can construct a new specification $\mathcal{N}^\prime$ where every rule in $Q_{out}$ is primed and moved to $Q_{snd}$. %and every auxiliary relation (if exists) is made $communication$.
We then add to $Q_{out}$ a rule to output every fact over the output schema $\mathcal{D}_{out}$.
The behavior of $\mathcal{N}^\prime$ is very simple: every time a new fact is derived by a rule, it is shuffled.
We have that the liveness property is naturally enforced also in $\mathcal{N}^\prime$ because $\mathcal{K}$ is unrestricted and, for this reason, every query is live.
Every fact in $\mathcal{Q}(\textbf{I})$, and no more, must hence also be in $\mathcal{N}^\prime(\textbf{I})$ whichever configuration we chose since the query is monotonic.

Consider $\mathbfcal{CN}^\prime$, the class of correct specifications defined by $\mathcal{N}^\prime$.
Assume now a new specification $\mathcal{N}^{\prime\prime}$ derived from $\mathcal{N}^\prime$ by considering a restricted keys-set $\mathcal{K}^{\prime\prime}$.
For simplicity we fix $\mathcal{K}^{\prime\prime}$ to be maximal.
We have to show that $\mathcal{N}^{\prime\prime}$ is eventually consistent with $\mathcal{N}^\prime$. %parallelly computes the same query of $\mathcal{N}^\prime$.
This is quite straightforward: every query is connected thus the liveness property still holds because $\mathcal{K}^{\prime\prime}$ is maximal and hence a non-null intersection exists among the destinations of all the atoms composing every rule-body.
Reasoning in the same way, every rule is also evaluated on a instance.
For what concerns safety, this is trivially satisfied because the query is monotonic.

We are now going to show how an inflationary transducer $\mathcal{T} = (\mathcal{P}, \Upsilon, \mathcal{K}^{\prime\prime})$ composing the specification $\mathcal{N}^{\prime\prime}$ can be built.
%An inflationary transducer $\mathcal{T} = (\mathcal{P}, \Upsilon, \mathcal{K})$ can be created where $\mathcal{K}$ is maximal and every rule in $\mathcal{Q}$ is made a send query. 
%Every derived fact over $\mathcal{D}_{out}$ is then output.
%Since every rule is connected we have the same opportunity of deriving facts as in the case in which $\mathcal{K}$ is unbounded.
%In the unchained case, instead, some rule must be added in order to enforce the liveness property.
%Because of the liveness property, and adding the fact that no wrong result can be derived since the query is monotone, the specification $(\mathcal{T}, \mathcal{T}^e, \gamma)$ parallelly computes $\mathcal{Q}$.
%$\mathcal{Q}$ is therefore bonunded.
%More precisely, 
Let $\Upsilon_{db} = \mathcal{D}_{in}$, $\Upsilon_{snd} = \{R^\prime \mid R \in sch(\mathcal{Q})\}$, $\Upsilon_{out} = \mathcal{D}_{out}$ ($systems$ and $time$ relations are as usual), and every term in every $communication$ relation is key.
%In addition, let the transducer program $\mathcal{P}$ be such that $Q_{snd}$ is the union of all the rules in $\mathcal{Q}_{out}$ having the head predicate in $\mathcal{D}_{out}$, while $Q_{ins}$ contains all the remaining rules.
Since $\mathcal{Q}$ is monotonic, no result can be derived that will be retracted in the future.
%We first consider the chained case.
The idea was to apply every rule as it is, and every new derived fact is sent to the other nodes composing the network.
Concretely, we first add to $Q_{snd}$, for each $R \in \Upsilon_{db}$, the following rule  implementing the shuffling: %$R^{\prime}_{snd}(\bar{u}) \leftarrow R(\bar{u}).$
\vspace{-2mm}
\begin{align}
\label{eq:partitioning}
R^{\prime}_{snd}(\bar{u}) \leftarrow R(\bar{u}).
\end{align}
%In addition we add to $Q_{snd}$ the rule:
%\begin{align}
%\label{eq:partitioning_1}
%R^\prime_{snd}(\bar{u}) \leftarrow R^\prime(\bar{u}).
%\end{align}
%to persist tuples during rounds.
%This is necessary since a tuple derived at round $s$ could be required at round $s^\prime$ with $s^\prime > s + 1$.
Now, let $\mathcal{Q}^{\prime}$ the version of $\mathcal{Q}$ where every relation is primed. 
We add to $Q_{snd}$ all the rules in $\mathcal{Q}^{\prime}$, and to $Q_{out}$ a rule: %$R_{out}(\bar{u}) \leftarrow R^\prime(\bar{u}).$
\vspace{-1ex}
\begin{align} 
\label{eq:output}
&R_{out}(\bar{u}) \leftarrow R^\prime(\bar{u}).
\end{align}
for each relation $R \in \mathcal{D}_{out}$ to output the results.
Note that the transducer is oblivious and monotonic.
Is easy to see that a transducer program generated in this way computes the initial query $\mathcal{Q}$.
\eat{We now prove iteratively that a transducer program generated in this way computes the initial query $\mathcal{Q}$.
Let \textbf{I} be an instance over $\mathcal{D}_{in}$, and $(N, D, H)$ a configuration.
Assume $\textbf{f}$ to be a fact over the relation $R \in \mathcal{D}_{out}$ and appearing in $\mathcal{Q}(\textbf{I})$.
We are going to show that $\textbf{f} \in \mathcal{N}^{\prime\prime}_{N, D, H}(\textbf{I})$. We proceed by contradiction: we assume that $\textbf{f} \notin \mathcal{N}^{\prime\prime}_{N, D, H}(\textbf{I})$ and show that this is impossible.
Consider the quiescent point $(\rho, t)$ of $\mathcal{S}^{\emph{bsp}}_{\mathcal{N}^{\prime\prime}}(N, D, H, \textbf{I})$.
Being $\textbf{f} \notin \mathcal{N}^{\prime\prime}_{N, D, H}(\textbf{I})$, by rule (\ref{eq:output}) we have that the related primed fact $\textbf{f}^\prime$ is not in the instance of the relation $R^\prime$ at round $t-1$.
This means that at round $t-2$ an emission rule $r^\prime$ exists in $Q_{snd}$ which has not derived the fact $\textbf{f}^\prime$, \ie the body of the rule was not satisfied.
Iteratively repeat this reasoning until a send rule is reached whose body contains more that one relation.
If such rule doesn't exist, $\textbf{f}^\prime$ was not derived because an extensional tuple $\textbf{e}$ was not in the initial instance.
Clearly this is a contradiction since $\textbf{f}$ is in $\mathcal{Q}(\textbf{I})$ and hence also $\textbf{e}$ must be in \textbf{I}.
Assume then that an emission rule exists having more that one relation in the body.
\kWlog we can assume that this rule is not satisfied because a fact $\textbf{f}^{\prime\prime}$ over the $communication$ relation $R^{\prime\prime}$ was not in the local instance of $R^{\prime\prime}$ and, in the case it was, the rule would be satisfied. 
This may have occurred for two reasons: $(i)$ $\textbf{f}^{\prime\prime}$ was addressed to a wrong node; or $(ii)$ $\textbf{f}^{\prime\prime}$ was never derived by the respective rules in the previous round. % or $(iii)$ $\textbf{f}^{\prime\prime}$ was derived previously but was not persisted in memory.
%Case $(iii)$ is trivially impossible since the transducer is inflationary. %every derived fact is persisted by rule (\ref{eq:partitioning_1}).
Case $(i)$ is impossible since every sent tuple is hashed over all its terms and, being the query connected and composed by more than one relation, at least a literal $R^{\prime\prime\prime}(\bar{v})$ exists which is connected with $R^{\prime\prime}(\bar{u})$, and contains an instance that will join with $\textbf{f}^{\prime\prime}$.
Since such instance must exist by construction, and it has been correctly hashed because $\bar{u} \cap \bar{v} \neq \emptyset$, if instances over $R^{\prime\prime\prime}$ are correctly addressed, also every instance over $R^{\prime\prime}$ must be addressed to the correct node. %since $H(R^{\prime\prime}(\bar{u})) \cap H(R^{\prime\prime\prime}(\bar{v})) \neq \emptyset$.
We can then conclude that $\textbf{f}^{\prime\prime}$ was never derived.
Repeating iteratively the same procedure one can show that $\textbf{f}$ was not derived because a base fact $\textbf{e}$ was not in the initial instance.
We have already see that this is impossible, therefore \textbf{f} must be in $\mathcal{N}^{\prime\prime}_{N, D, H}(\textbf{I})$.
On the other hand it can be easily seen that every fact in $\mathcal{N}^{\prime\prime}_{N, D, H}(\textbf{I})$ must also be in $\mathcal{Q}(\textbf{I})$ since $\mathcal{T}$ is composed by the same rules forming $\mathcal{Q}$.}
%If now we consider the unchained case, let $\mathcal{N}$ be the transducer of the previous case.
%Let $U^{(k, a)}$ be an unchained relation, and $r$ the rule of the transducer program $\mathcal{T}$ defining $\mathcal{N}$ where $U$ appears in the body.
%Let $r^\prime$ the emission rule defining $U$.
%We substitute $U$ in $r$ and $r^\prime$ with a new $emit$ relation $(U^{adom}, 1, 1)$.
\end{proof} 

%%%%%%%%%%%%%%%%%%% TC not required for TPLP %%%%%%%%%%%%%%%%%%%%
\eat{
\vspace{-1.5ex}
\begin{example}
\label{ex:computable_monotonic}
Let $\mathcal{Q}$ be the usual transitive closure query:
\begin{equation*}
\begin{split}
&T(u, v) \leftarrow R(u, v).\\
&T(u, w) \leftarrow R(u, v), T(v, w).
\end{split}
\end{equation*}

\noindent For the reader not familiar with recursive queries, note how the transitive closure relation $T$ is recursively filled by joining $R$ with $T$ itself.
A datalog transducer parallelly computing this query is:

\noindent\hrulefill
\vspace{-1mm}
\begin{equation*}
\begin{split}
\text{Schema: }&\Upsilon_{db} = \{R^{(2)}\}, \Upsilon_{snd} = \{S^{(1, 2)}, U^{(1, 2)}\}, %\mathcal{T}_{snd} = \{(R^{snd}, 1)\},
\Upsilon_{out} = \{T^{(2)}\}\\
\text{Program: }&S_{snd}(u, v) \leftarrow R(u, v).\\
%&S_{snd}(u, v) \leftarrow S(u, v).\\
&U_{snd}(u, v) \leftarrow S(u, v).\\
&U_{snd}(v, w) \leftarrow S(u, v), U(u, w).\\
%&U_{snd}(u, v) \leftarrow U(u, v).\\
&T_{out}(u, v) \leftarrow U_{snd}(u, v).
\end{split}
\end{equation*}
\vspace{-0.5ex}
\noindent\hrulefill
%\noindent Note that even in this case, if we employ the multi-core optimization and the syntactic sugaring for shuffling transducer programs, the transitive closure becomes similar to Example \ref{ex:tc_micro}. 
\end{example}
}

Connected queries have an interesting semantic property: they \emph{distribute over components} \cite{AmelootKNZ14}.
An instance $\textbf{J}$ is said to be \emph{connected} if whichever pair of constants \texttt{a}, \texttt{b} $\in adom(\textbf{J})$, a chain of facts \textbf{f}$_1$, \ldots, \textbf{f}$_n$ exists such that $\texttt{a} \in adom(\textbf{f}_1)$, $\texttt{b} \in adom(\textbf{f}_n)$, and for any pair of consecutive fact $\textbf{f}_i$, $\textbf{f}_{i+1}$ in the chain with $0 \leq i < n$, $adom(\textbf{f}_i) \cap adom(\textbf{f}_{i+1}) \neq \emptyset$.
Now, if $\textbf{I}$ is an instance, $\textbf{J}$ is a \emph{component} of $\textbf{I}$ if $(i)$ $\textbf{J} \subseteq \textbf{I}$, $(ii)$ $\textbf{J}$ is nonempty, and $(iii)$ $\textbf{J}$ is connected and maximal with this property in $\textbf{I}$.
Finally, if with $co(\textbf{I})$ we denote the components of $\textbf{I}$, a query $\mathcal{Q}$ is said to \emph{distribute over components} if $\mathcal{Q}(\textbf{I}) = \bigcup_{\textbf{J} \in co(\textbf{I})} \mathcal{Q}(\textbf{J})$ for all $\textbf{I}$.

\begin{proposition}[Ameloot at al. 2014]%[\cite{AmelootKNZ14}]
\label{pr:connected1}
Let $\mathcal{L} \subseteq$ {\sc datalog}$^{\neg}$.
Every connected $\mathcal{L}$-query distributes over components.
\end{proposition}\vspace{-1ex}

\begin{proposition}[Ameloot at al. 2015]%[\cite{AmelootKNZ15}]
\label{pr:connected2}
Every query computable by a {\sc datalog}$^\neg$-query that distributes
over components can be computed by a connected {\sc datalog}$^\neg$-query.
\end{proposition}
\vspace{-1ex}
\begin{example}
Consider again Example~\ref{ex:connected}. 
The query is not connected and in fact it does not distribute over components.
To see why this is the case, assume that the input instance is composed by two components. %~\footnote{To be more formal, the Gaifman graph of the input instance has two components.}. %: one containing a path of length three but no path of length four, and one instead containing a path of length four.
The query does not distribute because the result on each component might depend on the presence of a path of length four on the other component.

Conversely, assume the query:
\begin{equation*}
\begin{split}
&T(u) \leftarrow E(u, v), E(v, w), E(w,u).\\
&F(u) \leftarrow E(u, v), E(v, w), E(w,x), E(x,y), x \neq u, y \neq u.\\
&Q() \leftarrow T(u), F(u).
\end{split}
\end{equation*}
returning true if a component exist having a triangle, with a path of length four starting from the same source node.
The query is connected and distributes over components. 
\end{example} 
\vspace{-1ex}

\stitle{Remark}: Differently from \cite{AmelootKNZ15}, but according to \cite{AmelootKNZ14}, we are considering as not connected \emph{all} the queries having nullary relations in the body. 
In \cite{AmelootKNZ15}, a specific type of nullary relations is identified which can be safely used in rule bodies while maintaining connectedness.
Such nullary relations are however copied to each connected component, which in our parallel settings means  broadcasting all nullary relations, which in turn means that, for partitioned hash families, if negation is applied over a nullary relation the query is no longer parallelly computable (because unsafe).
We therefore restrict our attention to nonempty instances (having at least one component) and connected queries without nullary atoms in the body.

\subsubsection{Non-monotonic Queries} Clearly unconnected non-monotonic queries are not restricted. Interestingly, not every connected non-monotonic query is however restricted.
%This is due to the fact that shuffling transducer networks do not allow explicit broadcasting.
%The following example explain why a broadcasting specification is, for instance, necessary for connected queries which are not recursion bounded -- we have already seen in Example \ref{ex:false_calm} the case in which the non-monotonic query is recursion bounded but non-connected. 

\begin{example}
\label{ex:ct_parallel}
Consider the following {\sc datalog}$^\neg$-query computing the facts in $T$ not in the transitive closure of $R$.\vspace{-2ex}
\begin{equation*}
\begin{split}
&CS(u, v) \leftarrow R(u, v).\\
&CS(u, w) \leftarrow CS(u, v), R(v, w).\\
&Q(u, v) \leftarrow T(u, v), \neg CS(u, v).
%&CU(u, v) \leftarrow U(u, v), \neg CS(u, v), \texttt{ready}().\\
%&CU_{snd}(v, w) \leftarrow U(u, v), CU(u, w), \neg CS(v, w), \texttt{ready}().\\
%&Q_{out}(u, v) \leftarrow CU(u, v).\\
\end{split}
\end{equation*}
This query can be parallelly computed by the following oblivious and inflationary {\sc fo}-transducer:

\noindent\hrulefill
\vspace{-1mm}
\begin{equation*}
\begin{split}
\text{Schema: }&\Upsilon_{db} = \{R^{(2)}, T^{(2)}\}, \Upsilon_{com} = \{S^{(1, 2)}, U^{(1, 2)}, CS^{(1, 2)}\},\\&\Upsilon_{mem} = \{\texttt{Ready}^{(0)}\}, \Upsilon_{out} = \{Q^{(2)}\}\\
\text{Program: }\text{ }&
S_{snd}(v, u) \leftarrow R(u, v).\\
&U_{snd}(u, v) \leftarrow T(u, v).\\
&CS_{snd}(u, v) \leftarrow S(u, v).\\
&CS_{snd}(u, w) \leftarrow S(u, v), CS(v, w).\\
&\texttt{Ready}_{ins}() \leftarrow \neg \texttt{Ready}().\\
&Q_{out}(u, v) \leftarrow U(u, v), \neg CS(u, v), \texttt{Ready}().\\
%&CU(u, v) \leftarrow U(u, v), \neg CS(u, v), \texttt{ready}().\\
%&CU_{snd}(v, w) \leftarrow U(u, v), CU(u, w), \neg CS(v, w), \texttt{ready}().\\
%&Q_{out}(u, v) \leftarrow CU(u, v).\\
\end{split}
\end{equation*}

\vspace{-1ex}
\noindent\hrulefill
\vspace{1mm}

\noindent This specification correctly computes the query only when some specific hashing functions are used. 
Indeed, it might happen that facts are distributed unevenly among the nodes, and
that a node ends up deriving a new fact over $CS$ after all the other nodes have already finished their computation.
This may result in the possibility that a fact over $Q$ be retracted: intuitively, the problem is that negation is applied too early. 
In order to avoid this situation, common knowledge of local termination for each node is required, i.e., nodes must \emph{synchronize}: every node should notify every other node that it has locally terminated the computation of $CS$ and then, 
when every node has locally terminated the computation of the closure of $R$, and all nodes know this, $Q$ can be safely evaluated.
Clearly, this pattern requires a non-restricted specification.
On the other hand the same query can be correctly computed if every node has the entire instance locally installed.
In this case we are in fact guaranteed that every node will apply negation over the complete transitive closure.
This again is obtainable only with an unrestricted specification.
\end{example}

Nevertheless, a class of non-monotonic queries  exists that does not require broadcasting rules: \emph{recursion-delimited connected} queries.
%The issue with the previous example is that an atom exists -- $T$ in the above example -- whose terms are unrelated with any other term of the query.
%We now make this precise:

\begin{definition}
Given a {\sc datalog}$^{\neg}$ query $\mathcal{Q}$, we say that $\mathcal{Q}$ is \emph{recursion-delimited} if, whichever stratification $\mathcal{Q}_1, \ldots, \mathcal{Q}_n$ we choose, $\mathcal{Q}_1, \ldots, \mathcal{Q}_{n-1}$ are non-recursive programs, while $\mathcal{Q}_{n}$ is expressed in $\mathcal{L}$, with $\mathcal{L} \subseteq$ {\sc datalog}$^{\neg}$, and negation is applied over extensional database schemas only~\footnote{This language is commonly referred to as \emph{semi-positive} {\sc datalog}$^{\neg}$ (see also Appendix B.1).}. 
\end{definition}

\noindent Informally, recursion-delimited queries are non-monotonic queries in which recursion is only allowed in the last stratum. 
\begin{proposition}
\label{pr:embarassing_non_monotonic}
Let $\mathcal{L} \subseteq$ {\sc datalog}$^{\neg}$.
Every recursion-delimited, connected query expressible in $\mathcal{L}$ is restricted. %can be parallelly computed by an inflationary and oblivious $\mathcal{L}^\prime$-transducer network.
\end{proposition}\vspace{-2ex}
\begin{proof}
We follow the same procedure presented in the proof of Proposition \ref{lm:embarassing_monotonic}.
Let $\mathcal{Q}_1, \ldots, \mathcal{Q}_n$ be a stratification of the rules of the input query $\mathcal{Q}$.
An inflationary transducer $\mathcal{T} = (\mathcal{P}, \Upsilon, \mathcal{K})$ can be created where $\mathcal{K}$ is maximal and every stratum is evaluated sequentially.
Every derived fact over $\mathcal{D}_{out}$ is then output.
Since all the rules are connected, and the set of keys is maximal, the liveness property is always satisfied so we have the same opportunity of deriving fact as in the case in which $\mathcal{K}$ is unrestricted.
Although the safety property is not satisfied, no wrong result can be inferred because each stratum is evaluated sequentially and every rule is connected end evaluated on a instance.
%As for Theorem \ref{te:parallel_monotonic}, since all the rules are chained we have that $\bigcap_{R_i \in \textbf{B}} H(R_i) \neq \emptyset$ so we have the same opportunity in deriving fact as in the case in which $\mathcal{K}$ is maximal.
We can then conclude that the specification $(\mathcal{T}, \mathcal{T}^e, \gamma)$ parallelly computes $\mathcal{Q}$, and $\mathcal{Q}$ is a restricted query. 
%\mi{Non e' che in questo caso, poiche' la prova riecheggia la 4.17, vogliamo lasciare solo quella intuitiva? potremmo mettere quella formale in appendice... o forse eliminarla del tutto, e fare qualche riferimento nell'esempio 4.22}
More precisely, initially set $\Upsilon_{db} = \mathcal{D}_{in}$, $\Upsilon_{snd} = \{R^\prime \mid R \in sch(\mathcal{Q})$, $\Upsilon_{out} = \mathcal{D}_{out}$ ($systems$ and $time$ relations are as usual), and every term in every schema relation is key.
We start to generate $\mathcal{P}$ by adding to $Q_{snd}$ a  rule in the form of eq. (\ref{eq:partitioning}) for each $R \in \Upsilon_{db}$ to implement the shuffling (similarly to the proof of Proposition \ref{lm:embarassing_monotonic}).
Since $\mathcal{Q}$ is non-monotonic, an appropriate order of evaluation of rules must be enforced if we don't want to derive wrong results.
Thus, consider the stratification $\mathcal{Q}_1, \ldots, \mathcal{Q}_n$ of $\mathcal{Q}$.
First consider the first $n-1$ strata. 
By definition of recursion-delimited query, such strata are non-recursive, therefore, looking at the predicate dependency graph, we can assign a set of predicates, inside the same stratum, to a \emph{stage}.
This assignment follows the dependency graph, so that each predicate that depends on another predicate belongs to a higher stage.
Intuitively, the stratification is maintained since all predicates belonging to a higher stratum also belong to a higher stage.
We then have that stages, basically, are used to implement distributed stratification.
Let $m$ the highest stage thus obtained.
We create a new stage $m+1$ containing all the predicates in the last stratum $\mathcal{Q}_n$.
Consider now the query $\mathcal{Q}^\prime$ obtained from $\mathcal{Q}$  by $(i)$ priming all relations and $(ii)$ appending to the body of each rule in $\mathcal{Q}_j$ a nullary atom \texttt{Stage}$^j()$, with $j \in 1, \ldots, m+1$, in order to bind the evaluation of rules to the respective stage. 
We now add to $Q_{snd}$ all the queries in $\mathcal{Q}^\prime$ %plus a rule in the form of eq. \ref{eq:partitioning_1} for each predicate
, and to $Q_{out}$ a rule in the form of eq. (\ref{eq:output}) for each $output$ relation.
Finally, in order to advance stage by stage we add one rule in the form:\vspace{-2ex}
\begin{equation}
\label{eq:stratifed_2}
\texttt{Stage}^j_{ins}() \leftarrow \texttt{Stage}^i().
\end{equation}
\vspace{-1ex}
for each $0 < i < j < n$ and a rule: 
\begin{equation}
\label{eq:stratifed_1}
\texttt{Stage}^1_{ins}() \leftarrow \neg \texttt{Stage}^1(), \neg \texttt{Stage}^2(), \ldots, \neg \texttt{Stage}^n.
\end{equation}
to define when the first stage can start.

We are now going to prove that the transducer network derived from $\mathcal{P}$ actually computes the initial query $\mathcal{Q}$.
The main difference with respect to the proof of Proposition \ref{lm:embarassing_monotonic} is that now the query is non-monotonic and therefore each negative literal cannot be evaluated before all the related tuples are generate by the lower strata.
Let us proceed inductively: initially all the stages are false and only the rules implementing the shuffling can be evaluated.
In the successive super-step, the first stage is active.
Now, all the rules having just extensional relations of $\mathcal{Q}$ in the body are evaluated.
Let us denote with $q_R$ one of such rules. 
Since every previously sent fact is hashed over all the terms composing the tuple, and since every rule is connected, %we have that $\bigcap_{R_i \in \mathbf{B}} H(R_i) \neq \emptyset$, where $\mathbf{B}$ is the body of $r$.
%Because of this, 
we have that at least a node which is able to satisfy $body(q_R)$ exists, and a set of facts will be sent.
No wrong tuple can be derived because every rule is connected.
All the newly derived intensional facts, plus the previous extensional tuples, are now sent to a set of nodes based on the parallelization strategy.
The queries of the successive stage will then be evaluated in the successive round, and again, by construction, they are all evaluated on a instance.
Let now assume we are at stage $m$ and that ever query has been evaluated sequentially on a instance until that point.
Again this means that a new set of facts will be sent, together with the previously sent ones.
At stage $m+1$ every rules is clearly still correctly evaluated.
Note that the $m+1$-th stage can take more than one round to produce all tuples since it is allowed to be recursive.
%Again this is because every fact derived at stage $m$ are emitted based on the hashing on the full set of terms.
We finally have that every rule in $Q$ is evaluated on an instance by construction.
%This stage can be recursive and monotonic fall under Theorem \ref{te:parallel_monotonic}.
This concludes the proof.
\end{proof} 

\begin{example}
\label{ex:computable_monotonic}
Let $\mathcal{Q}$ be the following recursion-delimited, connected query:
\begin{equation*}
\begin{split}
&T(u, v) \leftarrow E(u, v), \neg F(u).\\
&T(u, w) \leftarrow E(u, v), T(v, w).
\vspace{-2mm}
\end{split}
\end{equation*}

\noindent which computes a transitive closure applied over a filtered set of edges.
A {\sc fo}-transducer parallelly computing the query is the following~\footnote{Note that nullary relations \texttt{Stage}$^i$ used to postpone the evaluation of rules can be actually omitted for this specific example: for clarity we however follow the technique used in the proof of Proposition~\ref{pr:embarassing_non_monotonic}.}:

\noindent\hrulefill
\vspace{-1mm}
\begin{equation*}
\begin{split}
\text{Schema: }&\Upsilon_{db} = \{E^{(2)}, F^{(1)}\}, \Upsilon_{snd} = \{S^{(2, 2)}, U^{(1,1)}, T^{(2, 2)}\}, \\&\Upsilon_{mem} = \{\texttt{Stage}^{1(0)}, \texttt{Stage}^{2(0)}\}, \Upsilon_{out} = \{Q^{(2)}\}\\
\text{Program: }&
\texttt{Stage}^1_{ins}() \leftarrow \neg \texttt{Stage}^1(), \neg \texttt{Stage}^2().\\
&\texttt{Stage}^2_{ins}() \leftarrow \texttt{Stage}^1().\\
&S_{snd}(u, v) \leftarrow E(u, v).\\
&U_{snd}(u) \leftarrow F(u).\\
%&S_{snd}(u, v) \leftarrow S(u, v).\\
&T_{snd}(u, v) \leftarrow S(u, v), \neg U(u), \texttt{Stage}^1().\\
&T_{snd}(v, w) \leftarrow S(u, v), T(u, w), \texttt{Stage}^2().\\
&Q_{out}(u, v) \leftarrow T(u, v).
\end{split}
\end{equation*}
\noindent\hrulefill
%\vspace{2mm}
\end{example}

\stitle{Remarks}: $(i)$ Monotonic queries do not have the same problem as non-recursion-delimited queries: even if tuples are hashed unevenly, no retraction can occur because of monotonicity. 
$(ii)$ If a mechanism existed for which a recursive {\sc datalog} query could be rewritten in a non-recursive form, all connected non-monotonic queries would be non-broadcasting. %in a communication-free way.
The problem of determining if a given recursive query is equivalent to some non-recursive program is called the \emph{boundedness} problem, and, unfortunately, is undecidable \cite{Gaifman:1993:UOP:174130.174142}.
\vspace{1mm}
%From Propositions \ref{lm:embarassing_monotonic} - \ref{pr:embarassing_non_monotonic} and Example~\ref{ex:not_connected} follows that:

\begin{corollary}
\label{co:bounded-is-connected}
Every restricted query is connected.
\end{corollary}
\begin{proof}
The fact that every restricted query computes a connected one directly follows from the definition of connectedness. 
Conversely a restricted query cannot be unconnected as shown for instance in Examples~\ref{ex:not_connected} and \ref{ex:ct_parallel}.
\end{proof}

\noindent Figure \ref{fig:hierarchy_queries} depicts the query taxonomy as discussed so far.

\begin{figure}[t]
\centering
\includegraphics[width=0.6\columnwidth]{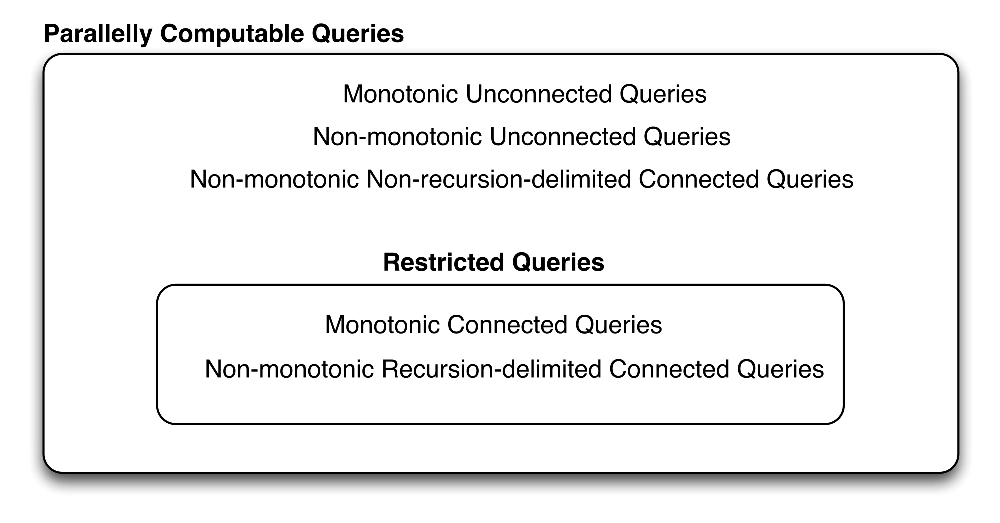}
\vspace{-5mm}
\caption{Query Taxonomy}
\label{fig:hierarchy_queries}
\end{figure}
%\eat{Figure \ref{fig:hierarchy_queries} depicts the taxonomy of the queries discussed so far.
%
%\vspace{-2mm}
%\begin{figure}[h]
%\centering
%\includegraphics[width=8cm,height=5cm]{img/chap_5_queries}
%\vspace{-5mm}
%\caption{Queries' Taxonomy}
%\label{fig:hierarchy_queries}
%\end{figure}
%\vspace{-2mm}
%
%}
%
%\noindent $(ii)$ If a mechanism existed for which a recursive datalog query could be rewritten in a non-recursive form, all connected non-monotonic queries would be non-broadcasting. %in a communication-free way.
%The problem of determining if a given recursive query is equivalent to some non-recursive program is called the \emph{boundedness} problem, and, unfortunately,  is undecidable \cite{Gaifman:1993:UOP:174130.174142}.
%
%Finally, from the above theorems we have the following corollary stating that a special kind of (non-connected) monotonic queries are embarrassingly parallel if we assume the partition function to be fair.
%
%\begin{co}
%Let $\mathcal{L}\subseteq$ datalog$^\neg$ be a language.
%Every semi-connected query expressible in $\mathcal{L}$ is embarrassingly parallel if a fair partition function is assumed.
%\end{co}

\section{Coordination-freedom Refined}
\label{sec:refine_coordination}
%!TEX root = TPLP.tex

We have seen in Section \ref{sec:co_co} that, for synchronous and reliable systems, a particular  notion of coordination-freedom is needed.
In fact we have shown that certain non-monotonic queries -- \eg Example \ref{ex:false_calm} -- that under the asynchronous model require coordination can be computed in a coordination-free way.
The key-point is that, as observed in \cite{AmelootNB13}, in asynchronous systems coordination-freedom is directly related to communica\-tion-freedom under ideal distribution. 
That is, if the distribution is right, no communication is required to correctly compute a coordination-free query because $(i)$ no data must be sent (the distribution is correct), and $(ii)$  no ``control message'' is required to obtain the correct result (the query is coordination-free).
However, due to its synchronous nature, in \emph{bsp} settings non-monotonic queries can be computed in general without resorting to coordination.
As already anticipated, this is due to the fact that the above concept of coordination is already ``baked'' into the \emph{bsp} model: each node is synchronized with every other one, hence ``control messages'' are somehow implicitly assumed. %therefore agreement on distributed property of the system -- such as when negation can be safely applied -- can be reached without communication. 
In this section we will introduce a novel knowledge-oriented perspective linking coordination with the way in which explicit and implicit information %(which we will call \emph{direct} and \emph{indirect} dependencies) 
flow in the network (Contribution 3).
Under this new perspective, we will see that a specification needs coordination if, in order to maintain convergence, a node must have \emph{some form of information exchange} with all the other nodes.\vspace{-1ex}

\subsection{Syncausality}
\label{sec:syncausal}

Achieving coordination in asynchronous systems -- \ie systems where each process proceeds at an arbitrary rate and no bound exists on message delivery time -- is a very difficult and costly task.
A necessary condition for coordination in such systems is the existence of  primitives enforcing some control over the ordering of events  \cite{Hunt:2010:ZWC:1855840.1855851}.
In a seminal paper \cite{Lamport:1978:TCO:359545.359563}, Lamport proposed a synchronization algorithm exploiting the relation of \emph{potential causality} ($\rightarrow$) over asynchronous events.
According to Lamport, given two events $e, e^\prime$, we have that $e \rightarrow e^\prime$ if  $e$ happens before $e^\prime$~\footnote{The potential causality relation is often quoted as \emph{happened-before}.}, and thus $e$ might have caused $e^\prime$. % and necessarily $e$ happen before $e^\prime$\footnote{The potential causality relation is often referred as \emph{happen-before}.}. 
%On the other side, $e \not \rightarrow e^\prime$ means that $e$ surely does not cause $e^\prime$, and if we have both $e \not \rightarrow e^\prime$ and $e^\prime \not \rightarrow e$, $e$ and $e^\prime$ are said to be \emph{concurrent}.
From a high-level perspective, the potential causality relation models how information flows among processes, and therefore can be employed as a tool to reason on the patterns which cause coordination in asynchronous systems.
A question now arises: what is the counterpart of the potential causality relation for synchronous systems?
\emph{Synchronous potential causality} (\emph{syncausality} in short) has been recently proposed \cite{Ben-Zvi:2014:BLH:2605175.2542181} to generalize Lamport's potential causality  to synchronous systems.
%In the following we will employ a slightly different definition for syncausality in order to accommodate the requirements of our framework.
%
Using syncausality we are able to model how information flows among nodes with the passing of time.
Given a run %specification class $\mathbfcal{N}$, and an input instance \textbf{I}, if 
$\rho$ 
%\in \mathbfcal{S}_{\mathbfcal{N}}(\textbf{I})$ is a run in the class of systems defined by $\mathbfcal{N}$, 
and  two points (\cf Section \ref{sec:transducer_network}) $(\rho^i, t)$, $(\rho^{j}, t^\prime)$ for not necessarily distinct nodes $i, j \in N$, we say that $(\rho^{j}, t^\prime)$ causally depends on $(\rho^{i}, t)$ if either $i = j$ and $t \leq t^\prime$ -- \ie a local state depends on the previous one -- or a tuple has been emitted by node $i$ at time $t$, and received by node $j$, %(and possibly $i = j$)
with $t < t^\prime$~\footnote{Note that a point in a synchronous system is what Lamport defines as an event in an asynchronous system.}.
We refer to these two types of dependencies as \emph{direct}.

\begin{definition}
Given a generic system $\mathcal{S}$, and a run $\rho \in \mathcal{S}$,
we say that two points $(\rho^{i}, t)$, $(\rho^{j}, t^\prime)$ are related by a \emph{direct potential causality} relation $\rightarrow$%, and write $(\rho^{i}, s) \rightarrow (\rho^{j}, s^\prime)$
, if at least one of the following is true:
\begin{enumerate}
\item $t^\prime = t + 1$ and $i = j$; 
\item $t^\prime \geq t + 1$ and node $j$ receive a tuple at time $t+1$ which was sent by $i$ at time $t$; 
\item there is a point  $(\rho^{k}, t^{\prime\prime})$ s.t. $(\rho^{i}, t) \rightarrow  (\rho^{k}, t^{\prime\prime})$ and  $(\rho^{k}, t^{\prime\prime}) \rightarrow (\rho^{j}, t^\prime)$.
\end{enumerate}
\end{definition}

\noindent Note that direct dependencies define precisely Lamport's happen-before relation -- and hence here we maintain the same symbol $\rightarrow$.

Differently from asynchronous systems, we however have that a point in node $j$ can occasionally \emph{indirectly} depend on another point in node $i$ even if no fact addressed to $j$ is actually sent by $i$.
This is because $j$ can still draw some conclusion simply as a consequence of the \emph{bounded delay guarantee}  and \emph{deterministic delivery} (\textbf{S3}) of synchronous systems. %-- \ie property \textbf{S3}$^{\prime\prime}$ of Section \ref{sec:synchronous_systems}.
That is, each node can use the common knowledge that every sent tuple is received at most after a certain bounded delay to reason about the state of the system.
The bounded delay guarantee can be modeled as an imaginary $\texttt{NULL}$ fact, as in \cite{Lamport:1984:UTI:2993.2994}.
Under this perspective, indirect dependencies appear the same as the direct ones, although, instead of a flow generated by ``informative'' facts, with the indirect relationship we model the flow of ``non-informative'', $\texttt{NULL}$ facts.
The interesting thing about the bounded delay guarantee is that it can be employed to specify when negation can be safely applied to a predicate.
In general, negation can be applied to a literal $R(\bar{u})$ when the content of $R$ is \emph{sealed} for what concerns the current round. %and therefore negation can be applied.
In local settings, we have that such condition holds for a predicate at round $t^\prime$ if its content has been completely generated at round $t$, with $t^\prime > t$.
In distributed settings we have that, if $R$ is a $communication$ relation, being in a new round $t^\prime$ is not enough, in general, for establishing that its content is sealed.
%This means that for sealing $R$, just the passing of time is not enough.
This is because tuples can still be floating, and therefore, until we are assured that every tuple has been delivered, the above condition does not hold. %since we are not guaranteed, locally, to have the total knowledge about the domain we are representing.
The result is that negation cannot be applied safely~\footnote{
Note that we can reason in the same way also for every other negative literal depending on $R$.
For this reason here we consider only the case in which negation is only applied over \emph{communication} relations.}. %the PCWA cannot be assumed to hold since even their content is not stable, at least until $R$ is not.

We will model the fact that, in synchronous systems, the content of a $communication$ relation $R$ is stable because of the bounded delay guarantee $\mathbf{S3}$, by having every node $i$ emit a fact $\texttt{NULL}^i_R$ at round $t$, for every $communication$ relation $R$; then each \texttt{NULL}$^i_R$ fact will be delivered to node $j$ exactly by the next round. 
We thus have that the content of $R$ is stable once $j$ has received a $\texttt{NULL}^i_R$ fact from every $i \in N$. 
The sealing of a $communication$ relation at a certain round is then ascertained only when $|N|$ $\texttt{NULL}_R$ facts have been counted.
We call this \emph{Snapshot Closed World Assumption} (\emph{SCWA}): negation on relation $R$ can be applied by a node just when it has surely received a consistent snapshot \cite{Babaoglu:1993:CGS:302430.302434} of the global instance $I_R$. 

We can assume that the program generating \texttt{NULL} facts is adjoined to the transducer program. 
Recall however that not necessarily the $\texttt{NULL}^i_R$ facts must be physically sent: %but they are only used to model the bounded delay guarantee of the channel connecting $i$ to $j$.
this in particular is true under the deterministic delivery semantics of \emph{bsp}, where the strike of a new round automatically seals all the $communication$ relations. In other words, under \emph{bsp} the program generating \texttt{NULL} facts is \emph{virtual}: no rule is fired and no fact is actually sent because the system definition already implicitly assume SCWA. Still these virtual facts will help us in reasoning about indirect flows of information, therefore we will still assume that such \texttt{NULL} facts are ``virtually" derived and sent.
%because of the $\textbf{S3}^{\prime\prime}$ condition.
Example \ref{ex-immaginary-NULL} below shows a concrete situation showing that although no \texttt{NULL} fact is sent, still SCWA holds.
The reader however should not believe that this is always the case.
In fact, in Section \ref{sec:async_del}, we will see how by simply dropping the deterministic delivery property \textbf{S3}$^{\prime}$ the situation becomes more complicated, and the virtual program generating \texttt{NULL} facts must actually be evaluated.
As a final remark, note that no $\texttt{NULL}^i_R$ fact need be issued if no SCWA must be enforced over $R$ or over a dependent relation.

\begin{definition}
\label{df:indirect}
%Let $(\mathcal{N}_{N\!, t\!, P\!, H}, \textbf{I})$ be a network context. %$\mathcal{N}$ be a specification, \textbf{I} an initial instance, and $(N, t, P, H)$ a configuration.
Given a run $\rho$, we say that two points $(\rho^{i}, t)$, $(\rho^{j}, t^\prime)$ %, with $\rho \in \mathcal{S}^{rsync}$, 
are
% let $\rho$ be a run in $\mathbfcal{S}_{\mathcal{N}(\textbf{I})}$. We then say that two points $(\rho^{i}, t)$, $(\rho^{j}, t^\prime)$ are 
related by an \emph{indirect potential causality} relation $\dashrightarrow$, if $i \neq j$, $t^\prime \geq t + 1$ %, $\theta(t^\prime) \geq \theta(t + 1) + \delta$ 
and a $\texttt{NULL}^i_R$ fact addressed to node $j$ has been (virtually) sent by node $i$ at round $t$.
%and the PCWA holds for a relation $R \in sch(\mathcal{N})$ on point $(\rho^{j}, t^\prime)$, while on the same run, for the same node, it was not holding for any time $\theta(t^{\prime\prime})$, where $\theta(t^{\prime\prime}) < \theta(t^\prime)$. % and $(\rho^{i}, t) \not \rightarrow (\rho^{j}, t^\prime)$.
\end{definition} 

\begin{example}
\label{ex-immaginary-NULL}
Consider the program of Example \ref{ex:false_calm}, a proper configuration and an initial instance.
%We have that at round $s$ no rule needs the SCWA to hold (\texttt{ready} is false).
At round $t+1$ we have that the SCWA does hold for relation $S$, and hence negation can be applied.
Note that if $R$ is empty in the initial instance, no fact is sent. %in the network.
Despite this, every node can still conclude at round $t+1$ that the content of $S$ is stable.
In this situation we clearly have an indirect potential causality relation. %between points $(\rho^{i}, t)$, $(\rho^{j}, t + 1)$, for all $i, j \in N$, with $\rho$ the run describing the execution of the system. %$\mathcal{S}^{\emph{rsync}}_{\mathcal{N}}(N, P, H, \textbf{I})$.
\end{example}

\begin{corollary}
\label{cor:necessary_indirect}
A necessary condition for an indirect potential causality relation to exist is the presence of a negated literal.
\end{corollary}

\eat{\noindent Note that both the direct and indirect causal relationships are defined as ``potential": the former because not necessarily every derived instance at round $t$ will causally effected the state of node $j$ at round $s^\prime$; the latter because is not automatic that the content of a relation is stable just because of the bounded delivery, as, for instance, can be deduced from Example \ref{ex:ct_parallel}.}
%From an instructive perspective, we can see indirect dependency as the derivation of a special $\texttt{NULL}$ fact -- similarly to \cite{Lamport:1984:UTI:2993.2994} -- appearing 
%Defined in this way, an indirect potential causality relation exists among every node in every round.
%Note, however, that this is not actually the case.
%A requirement for an indirect relation to exists, in fact, is that a node $j$ must have to indirectly infer some non-trivial conclusion by just observing the passing of time.
%We postpone to the next section the discussion on how actually an indirect relation exists between two points.
%to rule must contain a predicate which cannot be used as it is to derive new facts, but instead a not pre-defined amount of time must be waited before using it, otherwise wrong result could be inferred.

%Note, however, that actually not every time %$(\rho^{i}, t) \not \rightarrow (\rho^{j}, t^\prime)$ and $i \neq j$, and 
%$t^\prime = t +1$ we have an indirect relation between $(\rho^{i}, t)$ $(\rho^{j}, t^\prime)$.
%We instantiates $\texttt{NULL}$ facts as a surrogate for \emph{waiting units}. 
%A requirement for an indirect relation to exists, in fact, is that a rule must contain a predicate which cannot be used as it is to derive new facts, but instead a not pre-defined amount of time must be waited before using it, otherwise wrong result could be inferred.
%to $\texttt{NULL}$ tuple can be consumed by the receiving point.
%We will make this statement more clear in the next section.

We are now able to introduce the definition of syncausality: a generalization of Lamport's happen-before relation which considers not only the direct information flow, but also the flow generated by indirect dependencies. %\ie information flow consisting on normal and $\texttt{NULL}$ facts.

\begin{definition}
Let %$(\mathcal{N}_{N\!, t\!, P\!, H}, \textbf{I})$ be a network context, %specification, \textbf{I} an initial instance, $(N, t,$ $P, H)$ a configuration, 
$\rho$ be a run in the synchronous system $\mathcal{S}^{\emph{rsync}}$.
The \emph{syncausality} relation $\leadsto$ is the smallest relation such that:
\begin{enumerate}
\item if $(\rho^{i}, t) \rightarrow (\rho^{j}, t^\prime)$, then $(\rho^{i}, t) \leadsto (\rho^{j}, t^\prime)$; 
\item if $(\rho^{i}, t) \dashrightarrow (\rho^{j}, t^\prime)$, then $(\rho^{i}, t) \leadsto (\rho^{j}, t^\prime)$; and
\item if $(\rho^{i}, t) {\leadsto} (\rho^{j}, t^\prime)$ and $(\rho^{j}, t^\prime) {\leadsto} (\rho^{k}, t^{\prime\prime})$, then $(\rho^{i}, t) {\leadsto} (\rho^{k}, t^{\prime\prime})$.
\end{enumerate}
\end{definition}
%\noindent The definition of concurrent events under the Lamport's happen-before naturally generalizes also for concurrent points under the syncausal relation.

\subsection{From Syncausality to Coordination}
\label{sec:coordination_sync}

%At the level of granularity we set so far, no clear pattern emerges from the syncausality structure of runs that can be linked to a proper notion of coordination.
We next propose the \emph{predicate-level syncausality} relationship, modeling causal relations at the predicate level. 
That is, instead of considering how (direct and indirect) information flows between nodes, we introduce a more fine-grained relationship modelling the flows between \emph{predicates and nodes}. 

\begin{definition}
Given a %network context $(\mathcal{N}_{N\!, t\!, P\!, H}, \textbf{I})$, and a 
run $\rho \in \mathcal{S}^{\emph{rsync}}$, 
we say that two points $(\rho^{i}, t)$, $(\rho^{j}, t^\prime)$ are linked by a relation of \emph{predicate-level syncausality} $\overset{R}{\leadsto}$%over $R \in \Upsilon$, and write $(\rho^{i}, s) \overset{R}{\leadsto} (\rho^{j}, s^\prime)$
, if any of the following holds:
\begin{enumerate}
\item $i = j$, $t^\prime = t +1$ and a tuple over $R \in \Upsilon_{mem} \cup \Upsilon_{out}$ has been derived by a query in $Q_{ins} \cup Q_{out}$; 
\item $R \in \Upsilon_{com}$ and node $i$ sends a tuple over $R$ at time $t$ addressed to node $j$, with $t^\prime \geq t + 1$;
\item $R \in \Upsilon_{com}$ and node $i$ (virtually) sends a $\texttt{NULL}^i_R$ fact at time $t$ addressed to node $j$, with $t^\prime \geq t + 1$;
\item there is a point $(\rho^{k}, t^{\prime\prime})$ s.t. $(\rho^{i}, t) \overset{R}\leadsto (\rho^{k}, t^{\prime\prime})$ and $(\rho^{k}, t^{\prime\prime})\overset{R}\leadsto (\rho^{j}, t^\prime)$.
 %$R, R^\prime$ defines a \emph{predicate chain}. %with $R \overset{(\neg)}\dhmapsto R^\prime$, this is $R^\prime$ depends positively (negatively) on $R$.
\end{enumerate}
\end{definition}

\noindent 
\eat{Point 4 basically establishes that using the predicate-level syncausality relation we are able to infer the dependency graph of the $memory$,  $synchronous$ and $output$ predicates \footnote{Here we are indeed not considering the environment which stay in between the tuple emission and the tuple reception}.
The definition of \emph{predicate-level direct} and \emph{indirect potential causality} can be extracted from the predicate-level syncausality definition intuitively. }% condition $(1)$, $(2)$, and $(4)$ define the direct causality, while the indirect relation is expressed by condition $(3)$.
We are now able to specify a %when coordination is achieved. %necessary but not sufficient 
condition for achieving coordination.
Informally, we have that \emph{coordination exists when all the nodes of a network reach a common agreement that some event has happened}.  
But the only way to reach such agreement is that a (direct or indirect) information flow exists between the node in which the event actually occurred, and every other node.
This is a sufficient and necessary condition because of the reliability and bounded delay guarantee of \emph{rsync} system.
We formalize this intuition using the (predicate level) syncausality relationship: 

\begin{definition}
%Let $\mathbfcal{CN}$ be a correct class. %such that a specification $\mathcal{N} \in \mathbfcal{N}$ exists which is independent.
We say that a correct specification class $\mathbfcal{CN}$ \emph{manifests the coordination pattern} if, for all possible initial instances $\textbf{I} \in inst(\Upsilon_{db})$, whichever run $\rho \in \mathbfcal{S}_{\mathbfcal{CN}}(\textbf{I})$ we select where $N$ is not trivial, a point $(\rho^{i}, t)$ and a $communication$ relation $R$ exist so that $\forall j \in N$ there is a predicate-level syncausality relation with $(\rho^{i}, t) \overset{R}{\leadsto} (\rho^{j}, t^\prime)$ and $t^\prime \leq *$.
\end{definition}

\noindent We call node $i$ the \emph{coordination master}.
A  pattern with a similar role has been named \emph{broom} in \cite{DBLP:journals/jancl/Ben-ZviM11}.
Note that the condition $t^\prime \leq *$ is needed since only those points which actually contribute to the definition of the final state are of interest, while all the others can be ignored.

\stitle{Remark}: The reader can now appreciate to which extent coordination was already ``baked''  inside the broadcasting synchronous specifications of Section \ref{sec:synchronous_transducer_network}.
Note that broadcasting, in \emph{bsp}, brings coordination.
This is not necessarily true in asynchronous systems.
\vspace{1mm}

\noindent Intuitively, the coordination master is where the event occurs.
If a broadcasting of (informative or non-informative) facts occurs, then such event will become common knowledge among the nodes.
On the contrary, if broadcasting is not occurring, common knowledge cannot be obtained, therefore if the correct final outcome is still reached, this is obtained without coordination. % \ie without having to agree that something happen during the computation.
\eat{
%because they do not make any contribution toward the definition of the final state.
\noindent In practice, %we have that a class is generated from a specification. The latter 
a coordination protocol is required if either the class is not correct, %\ie in the class of systems specified by the network, two configured networks exist that are not consistent 
or the class is correct but all possible strategies will rise the coordination pattern. %or they are consistent, the network is communication-free, but the $time$ relation is used -- the program contains a coordination protocol exploiting the global clock. 
In the former case a coordination protocol is needed in order to make the specification independent.
In the latter a coordination mechanism was already injected into the program.
%Note that the condition $s+1 \leq *$ is needed since only points which actually contribute to the definition of the final state are of interest, while all the others can be ignored.
}
\noindent That is, if at least a non-trivial configuration exists s.t. the coordination pattern doesn't manifest itself, we have coordination-freedom:
%Following this intuition we have that coordination-freedom is define as follows:

\begin{definition}
\label{df:specification_coordination_free}
Given a correct class $\mathbfcal{CN}$ and an initial instance \textbf{I}, %let us denote with $\mathbfcal{S}_{\mathbfcal{CN}}(\textbf{I})$ the class of systems specified by $\mathbfcal{CN}$. 
we say that $\mathbfcal{CN}$ is \emph{coordination-free} if 
%\begin{description}
%\item[CF1] $\mathbfcal{S}^{rsync}_{\mathbfcal{N}}$ is convergent for at least a specification in $\mathbfcal{N}$; and
%\item[CF2] 
%a parallelization strategy $(\mathcal{K}, H)$ exists such that 
a non-trivial configuration $(N, D, H)$ can be selected for which $\mathcal{S}_{\mathcal{N}}(N, D, H, \textbf{I})$ does not manifest the coordination pattern, where $\mathcal{N} \in \mathbfcal{CN}$. % is the specification having $\mathcal{K}$ as key-set,  %-- \ie no coordination master is needed.
%\item $\mathcal{N}$ is time oblivious.
%\end{description}
\end{definition}
\noindent \kWlog, we will also dub the specifications belonging to $\mathbfcal{CN}$ as \emph{coordination-free}. 
From Definition \ref{df:specification_coordination_free} we can deduce the following proposition:
\begin{proposition}
\label{pr:coordination-bounded}
Every coordination-free specification parallelly computes a restricted query.
\end{proposition}
\begin{proof}
If a specification is coordination-free, the only flow of information is the direct flow.
In addition, the direct flow must be such that a master node does not exist, \ie no $communication$ relation is allowed to have a key set to $0$ because no broadcasting must occur.
\end{proof}

\noindent The reverse result clearly does not hold: a restricted query might require coordination, \eg non-monotonic, connected, recursion-delimited queries are restricted and not coordination-free.

Note that our definition of coordination-freedom based on syncausality relation makes it rather intuitive, in contrast with the original, declarative definition of~\cite{AmelootNB13}.

\subsection{From Coordination-freedom to Communication-freedom}

\cite{AmelootNB13} relates coordination-freedom with absence of communication under ideal distribution. It would then also be interesting to explore which relationship exists between our definition of coordination-freedom and communication.
%The definition of coordination-freedom introduced in \cite{AmelootNB13} 
%In addition, 
%follows that any coordination-free specification parallelly computes an embarrassingly parallel query and can be made communication-free.
%In fact, if a transducer is coordination-free, the only flow of information is the direct flow.
Since in a coordination-free specification broadcasting queries are not strictly needed, we can deduce that every coordination-free specification can be made communication-free. That is, at least a configuration exists for which the correct result is computed without emitting any fact: the trivial case is the configuration in which the partition function installs the full initial instance on one node only, and $H$ addresses every constant to the same node. 
\eat{\begin{proposition}
\label{pr:communication_free}
Let $\mathcal{N}$ be a coordination-free specification. %parallel computing an embarrassingly parallel query.
For every initial instance $\textbf{I} \in \Upsilon_{db}$ a not trivial configuration $(N, t,$ $P, H)$ exists so that $\mathcal{N}_{N\!, t\!, P\!, H}(\textbf{I})$ is communication-free. %version of $\mathcal{N}$, derived by making each emission query in $\mathcal{N}$ an insertion query, $\mathcal{N}_{N, t, \mathcal{H}}(\textbf{I})$ and $\mathcal{N}^{\prime}_{N, t, \mathcal{H}}(\textbf{I})$ are (eventually) consistent;
\end{proposition}
}

\begin{example}
As an example of a coordination- and communication-free specification consider the following {\sc ucq}-network $\mathcal{N}$ computing the transitive closure of the relation $R^{(2)}$: each node computes the closure of $R$ on its local data and then emits the derived atoms. %to the respective node. 

\noindent\hrulefill
\vspace{-1mm}
\begin{equation*}
\begin{split}
\text{Schema: }&\Upsilon_{db} = \{R^{(2)}\}, \Upsilon_{com} = \{S^{(1, 2)}\}, \Upsilon_{out} = \{T^{(2)}\}\\
\text{Program:}\text{ }&T_{out}(u, v) \leftarrow R(u, v).\\
&T_{out}(v, w) \leftarrow S(u, v), T(u, w).\\
&S_{snd}(u, v) \leftarrow T(u, v).
\end{split}
\end{equation*}

\noindent\hrulefill

\noindent $\mathcal{N}$ is oblivious and $\mathbfcal{S}^{\emph{bsp}}_{\mathcal{N}}(\textbf{I})$ is convergent for every initial instance \textbf{I}.
Since there is no negation, we have only to show that for every initial instance, a configuration exists such that the sending queries are not broadcasting.
Consider $D$ and $H$ such that the full instance is installed on one node, and every constant is hashed to that the same node.
$\mathcal{N}_{N\!, D\!, H}$ is communication-free.
%the communication-free specification $\mathcal{N}^\prime$ created from $\mathcal{N}$ by making $S^{(2)}$ a memory relation and hence the third rule an insertion rule.
%Whichever not trivial configurations we select having $P$ such that the initial instance is installed on every node, we have that $\mathcal{N}_{N, t, P}$ and $\mathcal{N}^{\prime}_{N, t, P}$ are consistent.
%Since no communication is used, no broadcasting pattern arise, therefore, by definition, we have that $\mathcal{N}$ is coordination-free.
\end{example}

\noindent We can therefore deduce that coordination-freedom might be a sufficient condition for a specification to be communication-free; however it is not a necessary condition, as shown by the next example.

\begin{example}
\label{ex:false_calm_unary}
Let $\mathcal{Q}$ be the following non-monotonic query: 
\begin{equation*}
%\begin{split}
Q(v) \leftarrow R(u, v), \neg T(u).
%\end{split}
\end{equation*}

\noindent The following {\sc fo}-transducer network parallelly computes $\mathcal{Q}$: %can be specified as follows:
\vspace{10ex}

\noindent\hrulefill
\vspace{-1mm}
\small
\begin{equation*}
\begin{split}
\text{Schema: }&\Upsilon_{db} = \{R^{(2)}, T^{(1)}\}, \Upsilon_{mem} = \{\texttt{Ready}^{(0)}\}, \phantom{aa} \\&\Upsilon_{com} = \{S^{(1, 2)}, U^{(1, 1)}\}, \Upsilon_{out} = \{Q^{(1)}\}.\\
\text{Program: }&S_{snd}(u, v) \leftarrow R(u, v).\\
&U_{snd}(u) \leftarrow T(u).\\
&\texttt{Ready}_{ins}() \leftarrow \neg \texttt{Ready}().\\
&Q_{out}(v) \leftarrow S(u, v), \neg U(u), \texttt{Ready}().
\end{split}
\end{equation*}
\normalsize
\noindent\hrulefill
%\vspace{2mm}

\noindent The specification is non-monotonic (and thus requires coordination), restricted, and can be made com\-munica\-tion-free.
Consider in fact a distribution function $D$ which installs the entire instance on a node $i$.
Assume then $H$ such that a hash function exists by which every emitted tuple is addressed to $i$, and non-trivial $N$.
Whichever initial instance \textbf{I} is given, we clearly have that $\mathcal{N}_{N\!, D\!, H}$ is communica\-tion-free.
\end{example}

\noindent Unfortunately, coordination-freedom is undecidable in general~\footnote{Recall that for a specification to be coordination-free, it must first of all be independent. Independence is undecidable (\cf Appendix A).}.
However, from the above intuitions we can draw the following: %\emph{Embarrassingly Parallel Is Communication-free} (\emph{EPIC}) 
%Proposition:
 %which generalize Proposition \ref{pr:communication_free} by properly defining the essence of bounded queries:
\begin{proposition}
\label{te:epic}
Every restricted query is parallelly computable
by a hashing specification which can be made communication-free.
%Every specification parallely computing an embarrassingly parallell query can be made communication-free.
%Let $\mathcal{N}$ be a specification parallelly computing an embarrassingly parallel query.
%For every initial instance $\textbf{I} \in \Upsilon_{db}$ a not trivial configuration $(N, t, P, \mathcal{H})$ exists so that $\mathcal{N}_{N, t, P, \mathcal{H}}(\textbf{I})$ is communication-free. %version of $\mathcal{N}$, derived by making each emission query in $\mathcal{N}$ an insertion query, $\mathcal{N}_{N, t, \mathcal{H}}(\textbf{I})$ and $\mathcal{N}^{\prime}_{N, t, \mathcal{H}}(\textbf{I})$ are (eventually) consistent;
\end{proposition}
\begin{proof}
%We start from the if direction. 
By definition every restricted query is parallelly computable by a specification $\mathcal{N}$.
We have to show that $\mathcal{N}$ can be made communication-free.
%If the query is monotonic we fall under the case of Proposition \ref{pr:communication_free}.
%We then consider here the case in which the query is non-monotonic.
Consider a configuration where $N$ is non-trivial and $D$, $H$ are such that the full initial instance is installed on node $i$, and a hash function exists so that every constant is hashed to the same node $i$.
We have that $\mathcal{N}_{N\!, D\!, H}$ is independent by definition, and communication-free.
\end{proof}

\begin{corollary}
Every coordination-free specification is communication-free.
\end{corollary}

\noindent \eat{This prove why the previous definition of coordination-freedom was not appropriate for our context: coordination-freedom was related with communication-freedom, and certain non-monotonic queries exists which are communication-free.}
With Proposition \ref{te:epic} we have described one of the characteristics of restricted queries: they are communica\-tion-free.
In the next section we will see that a larger class of queries can be computed in a communication-free way. These will be called \emph{embarrassingly parallel queries}~\footnote{The term embarrassingly parallel comes from the parallel computing field, where it refers to the class of problems parallelly solvable by a set of tasks, without resorting to communication \cite{Foster:1995:DBP:527029}.}.

\section{CALM in Rsync Systems}
\label{sec:sync_calm}
%!TEX root = TPLP.tex

As we have seen, the original version of the CALM principle as postulated in Conjecture \ref{cn:calm_new} is trivially not satisfiable in \emph{bsp} systems, because some monotonic queries exist -- \ie unconnected queries, \cf Section \ref{sec:parallel} and Example \ref{ex:not_connected} -- that are not coordination-free in the sense of Definition \ref{df:specification_coordination_free}.
We will then prove the CALM conjecture just for the remaining class of monotonic queries.
In the remainder of the section we will then close the circle by discussing how the notion of coordination introduced in \cite{AmelootNB13} collapses into the one we proposed, once the synchronization constraints are weakened.
 
%We will then try to prove the CALM conjecture just for the remaining class of monotonic queries.
%Then just a restricted version of the CALM conjecture holds for \textit{det} systems:

\subsection{The CALM Conjecture for \emph{bsp} Systems}

Let us  first introduce the following lemma:
\begin{lemma}
\label{lm:calm_only_if}
Let $\mathcal{L} \subseteq$ {\sc datalog}$^\neg$ be a query language. Every query that is parallelly computed by a coordination-free $\mathcal{L}$-transducer network is monotone and connected.
\end{lemma}
\begin{proof}
Let $\mathcal{N}$ be a coordination-free $\mathcal{L}$-transducer network parallelly computing a query $\mathcal{Q}$. 
By applying Proposition~\ref{pr:coordination-bounded} and Corollary~\ref{co:bounded-is-connected} we know that $\mathcal{Q}$ is connected, but might still be non-monotonic (\cf Figure~\ref{fig:hierarchy_queries}).
%\mi{vedi commento nel mio  messaggio di email}
It remains then to show that $\mathcal{Q}$ is indeed monotonic.
%First chained. 
%Assume Q as a  predicate in a rule r which is not chained. Since a non-chained predicate exists, the only way to make N parallely compute Q is that R is broadcasted. 
%In fact, assume in the contrary that R is not broadcasted but hashed over every term for example. Assume now we have an instance I where adom(I_R) \subseteq I. this means that  
Let \textbf{I} and \textbf{J} be two initial instances over $\mathcal{D}_{in}$ such that $\textbf{I} \subseteq \textbf{J}$.  
We have to show that $\mathcal{Q}(\textbf{I}) \subseteq \mathcal{Q}(\textbf{J})$. 
Assume first that $adom(\textbf{J} \setminus \textbf{I}) \cap adom(\textbf{I}) = \emptyset$. Our intuition is that, under this assumption, a configuration exists such that communication is not required to prove monotonicity. We will then use such configuration to show that indeed this holds also for any arbitrary pair of instances. 
Since $\mathcal{N}$ parallelly computes $\mathcal{Q}$, $\mathcal{N}$ is independent, therefore a non-trivial configuration exists such that $D$ installs $\textbf{I}$ on one single node, while $\textbf{J} \setminus \textbf{I}$ is completely installed in another node. 
In addition, assume a hash function in $H$ such that all the constant in \textbf{I} are hashed to the same node $i$ in which \textbf{I} is installed, while all the constant in $adom(\textbf{J} \setminus \textbf{I})$ are hashed to the node $j \neq i$ in which $\textbf{J} \setminus \textbf{I}$ resides.
Consider a fact $\textbf{f} \in \mathcal{Q}(\textbf{I})$. 
By contruction $\textbf{f}$ will appear in the output of $\mathcal{N}_{N, D, H}(\textbf{I})$ at node $i$ at a certain round $t$.
Consider now the case in which the input instance is $\textbf{J}$.
We have a node $j \ne i$ such that $I^j_{db} = \textbf{J} \setminus \textbf{I}$, while again $I^i_{db} = \textbf{I}$.
 Let us consider the point $(\rho, t)$ of the run $\rho \in \mathcal{S}^{\emph{bsp}}_{\mathcal{N}}(N, D, H, \textbf{I})$.
Because local transitions are deterministic, and no fact in $\textbf{J} \setminus \textbf{I}$ can be addressed from $j$ to $i$, \textbf{f} is output also in run $\rho$.
Again, by construction, being $\mathcal{N}$ independent, $\mathcal{N}_{N, D, H}(\textbf{J})$ parallelly computes the query $\mathcal{Q}(\textbf{J})$, therefore $\textbf{f}$ must also belong to $\mathcal{Q}(\textbf{J})$.

Consider now the communication-free specification $\mathcal{F}$ built from $\mathcal{N}$:
we can freely use this procedure since $\mathcal{N}$ is communication-free by Proposition \ref{te:epic}.
We then have that $\mathcal{N}_{N, D, H}(\textbf{I})$ = $\mathcal{F}_{N, D, H}(\textbf{I})$ and, similarly, $\mathcal{N}_{N, D, H}(\textbf{J})$ = $\mathcal{F}_{N, D, H}$ $(\textbf{J})$.
Consider now two generic input instances $\textbf{J}^\prime, \textbf{I}^\prime$ with $\textbf{J}^\prime \supseteq \textbf{I}^\prime$, and the same distribution function $D$ installing \textbf{I}$^\prime$ on node $i$ and $\textbf{J}^\prime \setminus \textbf{I}^\prime$ on node $j$.
Also in this case we have that $\mathcal{F}_{N, D, H}$ parallelly computes $\mathcal{Q}$ and,  reasoning as above, $\mathcal{F}_{N, D, H}(\textbf{I}^\prime) \subseteq \mathcal{F}_{N, D, H}(\textbf{J}^\prime)$, thus $\mathcal{F}$ is monotonic.
As a consequence, $\mathcal{N}$ is also monotonic.
\end{proof}

%The original version of the CALM conjecture as postulated in Conjecture \ref{} is trivially not satisfiable by an \emph{det} system 
\noindent We are now able to prove the restricted version of the CALM conjecture for \textit{bsp} systems (Contribution 4):

\begin{theorem}
A query can be parallelly computed by a coor\-dination-free transducer network iff it is monotone and connected.
\end{theorem}

\begin{proof}
Starting from the if direction, by Proposition \ref{lm:embarassing_monotonic} we know that a connected {\sc datalog} (\ie monotonic) query can be parallelly computed by an oblivious hashing transducer network $\mathcal{N} \in \mathbfcal{CN}$.
It remains to show that $\mathcal{N}$ is coordination-free.
We can notice that, because the transducer network is monotonic, no coordination pattern can occur because of indirect information flow (from Corollary~\ref{cor:necessary_indirect}).
On the other hand, a coordination pattern  might occur because of direct information flow caused by broadcasting rules. %in case the query is not embarrassingly parallel. % (the transducer is oblivious, therefore no broadcasting rule can exist).
Let $\mathcal{N}^\prime \in \mathbfcal{CN}$ be a specification (not necessarily different from $\mathcal{N}$) such that $\mathcal{K}$ is restricted.
Note that a restricted specification exists because every connected {\sc datalog} query is restricted by Proposition \ref{lm:embarassing_monotonic}. 
It remains to show that a non-trivial configuration $(N, D, H)$ exists, such that for every initial instance \textbf{I}, a run $\rho \in \mathcal{S}^{bsp}_{\mathcal{N}^\prime}(N, D, H, \textbf{I})$ exists where the coordination pattern does not appear.
Let us assume $H$ to contain a hash function such that every fact emitted is always addressed to the same node. 
Indeed the coordination pattern cannot exists in $\rho$ since all the tuples are not broadcasted but addressed to the same node.
Finally, $\mathcal{N} \in \mathbfcal{CN}$, therefore $\mathcal{N}_{N, D, H}$ correctly computes the query.

For what concerns the only-if direction, it is covered by Lemma \ref{lm:calm_only_if}.~\footnote{Note that the Lemma -- and thus the Theorem -- still holds  although it is well known that a set of monotonic queries exists which are not expressible in {\sc datalog} \cite{AfratiCY95}.}
%Starting from the if direction, by Lemma \ref{lm:parallel_monotonic} we know that a chained datalog query can be parallelly computed by an oblivious hashing transducer network $\mathcal{N} \in \mathbfcal{CN}$.
%$\mathbfcal{CN}$ is coordination-free since a specification $\mathcal{N}^\prime \in \mathbfcal{CN}$ exists that is bounded, and a configuration $(N, P, H)$ can be created where the coordination pattern doesn't appear.
%For the only-if direction, %-- although it is well known that a set of monotonic queries exists which are not expressible in datalog \cite{Afrati:1995:DVP:214597.215185} -- 
%note that every query that is parallelly computed by a coordi\-nation-free $\mathcal{L}$-transducer network is monotone and necessarily chained.
\end{proof}
%\vspace{-2mm}

\begin{corollary}
A {\sc datalog} query can be parallelly computed by a coor\-dination-free transducer network iff it is monotone and distributes over components.
\end{corollary}
\begin{proof}
The corollary directly follows from Theorem 1 and Propositions \ref{pr:connected1} - \ref{pr:connected2}.
\end{proof}
%
%\noindent The next theorem points out that a more general result can be obtained:
%We report here the theorem and we prove it also for the synchronous and reliable case.
%\begin{te}
%\label{te:calm_final}
%Let $\mathcal{L}$ be a query language containing Nrdatalog.
%For every query $\mathcal{Q}_{out}$ that is expressible in $\mathcal{L}$, the following are equivalent:
%\begin{enumerate}
%\item $\mathcal{Q}_{out}$ can be parallelly computed by an oblivious shuffling transducer network; 
%\item $\mathcal{Q}_{out}$ can be parallelly computed by a coordination-free transducer network; and 
%\item $\mathcal{Q}_{out}$ is monotone.
%\end{enumerate}
%\end{te}
%\begin{proof}
%$3 \Rightarrow 1$ from Theorem \ref{te:parallel_monotonic}; $1 \Rightarrow 2$ from the if-part of the CALM conjecture; and finally $2 \Rightarrow 3$ from Lemma \ref{lm:calm_only_if}.
%\end{proof}
%

%\noindent The ``C" letter in CALM refers to ``Consistency". 
%From the previous section we know that every coordination-free specification $\mathcal{N}$ is indeed consistent, \ie $\mathcal{N} \in \mathbfcal{CN}$.
%But what about communication-freedom?
\eat{As a preliminary answer to this question we can try to give a different reading of the CALM principle, by relating com\-munication-freedom (instead of coordination-freedom) with monotonicity.
%we want to relate the CALM theorem with the notion of communication-freedom.
\begin{pr}
\label{pr:communication_free_CALM}
Every oblivious specification parallely computing a datalog query can be made communication-free.
%Let $\mathcal{N}$ be a coordination-free specification.
%For every initial instance $\textbf{I} \in \Upsilon_{db}$ a not trivial configuration $(N, t, P, \mathcal{H})$ exists so that $\mathcal{N}_{N, t, P, \mathcal{H}}(\textbf{I})$ is communication-free. %version of $\mathcal{N}$, derived by making each emission query in $\mathcal{N}$ an insertion query, $\mathcal{N}_{N, t, \mathcal{H}}(\textbf{I})$ and $\mathcal{N}^{\prime}_{N, t, \mathcal{H}}(\textbf{I})$ are (eventually) consistent;
\end{pr}
}
 %From Theorem \ref{te:epic} and Proposition \ref{pr:communication_free_CALM} 
 
So far we have considered coordination-freedom. %The ``C" letter in CALM refers to ``Consistency". 
But what about communication-freedom?
As previously mentioned, we name the class of communication-free queries as \emph{embarrassingly parallel}.

\begin{definition}
\label{df:embarrassingly}
Let $\mathcal{L}$ be a language and $\mathcal{Q}$ a $\mathcal{L}$-query.
$\mathcal{Q}$ is \emph{embarrassingly parallel} if it is parallelly computable by a specification that can be made communication-free. 
\end{definition}

\noindent As a preliminary answer to the above question, we can try to give a different reading of the CALM principle, by relating communication-freeness (instead of coordination-freeness) with monotonicity.
%By leveraging Proposition \ref{te:epic} we have that a wide class of queries can be parallelly computed in a communication-free way.

%we want to relate the CALM theorem with the notion of communication-freeness.
\begin{lemma}
\label{pr:communication_free_CALM}
Every oblivious specification parallelly computing a {\sc datalog} query can be made communication-free.
%Let $\mathcal{N}$ be a coordination-free specification.
%For every initial instance $\textbf{I} \in \Upsilon_{db}$ a not trivial configuration $(N, t, P, \mathcal{H})$ exists so that $\mathcal{N}_{N, t, P, \mathcal{H}}(\textbf{I})$ is communication-free. %version of $\mathcal{N}$, derived by making each emission query in $\mathcal{N}$ an insertion query, $\mathcal{N}_{N, t, \mathcal{H}}(\textbf{I})$ and $\mathcal{N}^{\prime}_{N, t, \mathcal{H}}(\textbf{I})$ are (eventually) consistent;
\end{lemma}
\begin{proof}
Assume that a proper initial instance \textbf{I} is given.
Consider first a coordination-free specification $\mathcal{N}$ computing a restricted monotonic query.
We have already seen in Preposition \ref{te:epic} that a configuration exists which makes $\mathcal{N}$ communication-free.
Consider now the case in which the monotonic specification $\mathcal{N}$ is computing a query $\mathcal{Q}$ which is not restricted.
Consider a configuration $(N, D, H)$ as described in Preposition \ref{te:epic}: $D$ installs the full instance on a unique node $i$, $H$ addresses all the emitted facts to $i$, and $N$, is arbitrary but not trivial.
$\mathcal{N}(N, D, H, \textbf{I})$ is communication-free.
%Let $\mathcal{N}^f$ be the the communication-free version of $\mathcal{N}$.
We have to show that $\mathcal{N}(N, D, H, \textbf{I})$ computes the query $\mathcal{Q}(\textbf{I})$.
Consider first all the nodes $j \neq i$ in $N$.
Such nodes will output nothing since their instance is empty and the query is monotone.
For what concern $i$, it exactly computes $\mathcal{Q}(\textbf{I})$ since it contains the full instance.
\end{proof}

%Now, by employing the CALM and EPIC theorems we are able to state a more general result:
\noindent We can now %employ the previous results to 
state the following %\emph{EPIC} (\emph{Embarrassingly Parallel Is Communication-free}) 
Theorem:
\begin{theorem}
\label{te:nice}
Let $\mathcal{L}$ be a query language containing {\sc ucq}.
For every query $\mathcal{Q}$ expressible in $\mathcal{L}$, the following are equivalent:
\begin{enumerate}
%\item $\mathcal{Q}_{out}$ can be parallelly computed by a communication-free transducer network;
%\vspace{-2mm}
\item $\mathcal{Q}$ can be parallelly computed by an oblivious, inflationary transducer network; and 
\item $\mathcal{Q}$ is embarrassingly parallel.
\end{enumerate}
\end{theorem}
\begin{proof}
$2 \Rightarrow 1$ follows from Proposition \ref{lm:embarassing_monotonic}.
It remains to prove that every oblivious and inflationary transducer can be made communica\-tion-free.
We will show that the only kind of queries which can be parallelly computed and are not communication-free are the non-recursion-delimited queries.
From Lemma~\ref{pr:communication_free_CALM} we already know that monotonic queries are communication-free.
From Proposition~\ref{te:epic} we instead know that non-monotonic restricted recursion delimited queries also communication-free. 
We now proceed by contradiction: assume $\mathcal{N}$ is a non-monotonic transducer network computing the query $\mathcal{Q}$, and $\mathcal{Q}$ is not recursion-delimited, \ie a recursive stratum $m$ exists, which is followed my another stratum $m+1$.
By definition of stratification the stratum $m+1$ cannot be evaluated before stratum $m$ has terminated, otherwise wrong facts could be derived.
The transducer, in order to correctly compute the query, must therefore be able to detect when the recursion is terminated and hence the evaluation of the $m+1$-th stratum can start.
Since each node composing the network could end up having different (overlapping) partitions of the initial instance, different nodes might terminate the recursive computation in different rounds.
Note that, although a partition might exist for which recursion terminates at the same round for all nodes, $\mathcal{N}$ is independent by definition, therefore it must be able to compute $\mathcal{Q}$ even in the case in which recursion terminates unevenly.
Every node can detect that every other node has terminated its local recursive computation only by a direct information flow.
In particular, a broadcasting communication must be executed since every node must communicate to every other node that it has finished its local computation.
To express that a node has finished its computation, \texttt{id} must be clearly read, otherwise some receiving node might not be able to identify which node has actually terminated the computation.
Every node, in addition, in order to deduce that every other node has terminated its local computation, must read the \texttt{All} relation to know which nodes in the network have communicated that their local computation is completed.
Clearly this is not an oblivious specification since the $system$ relations are employed.
%This concludes the proof, \mi{che bisogno hai di dire anche questo? la prova non e' gia' finita?}since every parallelly computable oblivious query is either monotonic or non-monotonic, recursion-delimited and connected, which are exactly the queries which can be made communication-free.
\end{proof}

\stitle{Discussion}:
Summarizing, we have seen that three different classes of coordination patterns can be identified under the \emph{bsp} semantics (Contribution 5), all of them requiring acquisition of common knowledge of a property: \eat{one which is required to correctly evaluate negated atoms in non-monotonic queries that are recursion bounded; it is com\-munication-free and it does not need to be implemented because enforced just by the semantics of the system.
One which is implemented through broadcasting and can be communication-free; and one which instead necessarily requires communication and the access to $system$ relations.}
%We name the former 
\emph{snapshot coordination}, which implements the SCWA, and require the common knowledge of a relation instance to be globally sealed; \emph{broadcasting coordination} is required for unconnected queries and necessitate each node to know that a relation instance is not empty; and \emph{synchronized coordination} requiring common knowledge of local termination of all the nodes.
Broadcasting coordination is simple to implement because it only requires a broadcasting query.
Snapshot coordination exploits the indirect information flow and hence is communication-free, and is used by any non-monotonic, recursion-delimited query.
Finally, synchronized coordination necessarily requires access to $system$ relations, since non-monotonic, non-recursion-delimited queries must be synchronized by a direct information flow in order to  maintain consistency (\cf Example \ref{ex:ct_parallel}).
%Clearly, only the latter needs to be injected in a specification when required: snapshot coordination is enforced by the \emph{det} semantics, as already explained, while broadcasting coordination is directly expressed into the program.
%
\eat{An interesting difference between the three coordination patterns is that the first two are inherent of the system's semantics, while synchronized coordination is tolerant to any changes over the distributed system model. %and therefore synchronized coordination code is required to be synthesized, in general, for every system type.
In the next section we will hence show %the rationale behind the label ``snapshot coordination" and 
how the semantics of snapshot and broadcasting coordination patterns changes if we weaken the constraint of the system definition.
%While the reason behind the name ``snapshot coordination" will be clarified in Section \ref{sec:async_del}, the rationale of the appellative ``synchronized coordination" is quite intuitive: all the nodes must be synchronized in order to make the content of a $synchronous$ predicate stable.
%Note that, because of this, simultaneous coordination is directly related with \emph{common knowledge} \cite{Fagin:2003:RK:995831}.
%\eat{Figure \ref{fig:hierarchy_queries_new} updates the depiction of parallelly computable queries of Figure \ref{fig:hierarchy_queries} with the new results just discussed.
%
%\begin{figure}[h]
%\centering
%\includegraphics[width=8cm,height=6cm]{img/chap_5_queries_new}
%\vspace{-3mm}
%\caption{Complete Taxonomy for Parallelly Computable Queries under \emph{det}}
%\label{fig:hierarchy_queries_new}
%\end{figure}
%}
So far, in fact, we have only considered the deterministic delivery model, \ie tuples arrive exactly at time $\theta(s + 1) + var$ if emitted at round $s$.
But what would happen if we assumed less constrained systems?
Below, we first take a look at systems with non-deterministic delivery but bounded variance (condition $\textbf{S3}^{\prime\prime}$ of Section \ref{sec:synchronous_systems} does not hold); and %anymore. %\footnote{Note that still conditions \textbf{S1} - \textbf{R2} are assumed to hold.}
%This delivery model is clearly a more accurate representation of the non-deterministic behavior of actual communication means.
then we will conclude with \emph{rsync} systems, \ie systems with non-deterministic delivery and arbitrary finite variance. %and finally we will briefly introduced coordination under reliable systems with unbounded delay, which we refer to as \emph{rsud} systems.
}
Figure \ref{fig:hierarchy_queries_new} updates Figure \ref{fig:hierarchy_queries} with the new results we have  just discussed.

\begin{figure}[h]
\centering
\includegraphics[width=0.6\columnwidth]{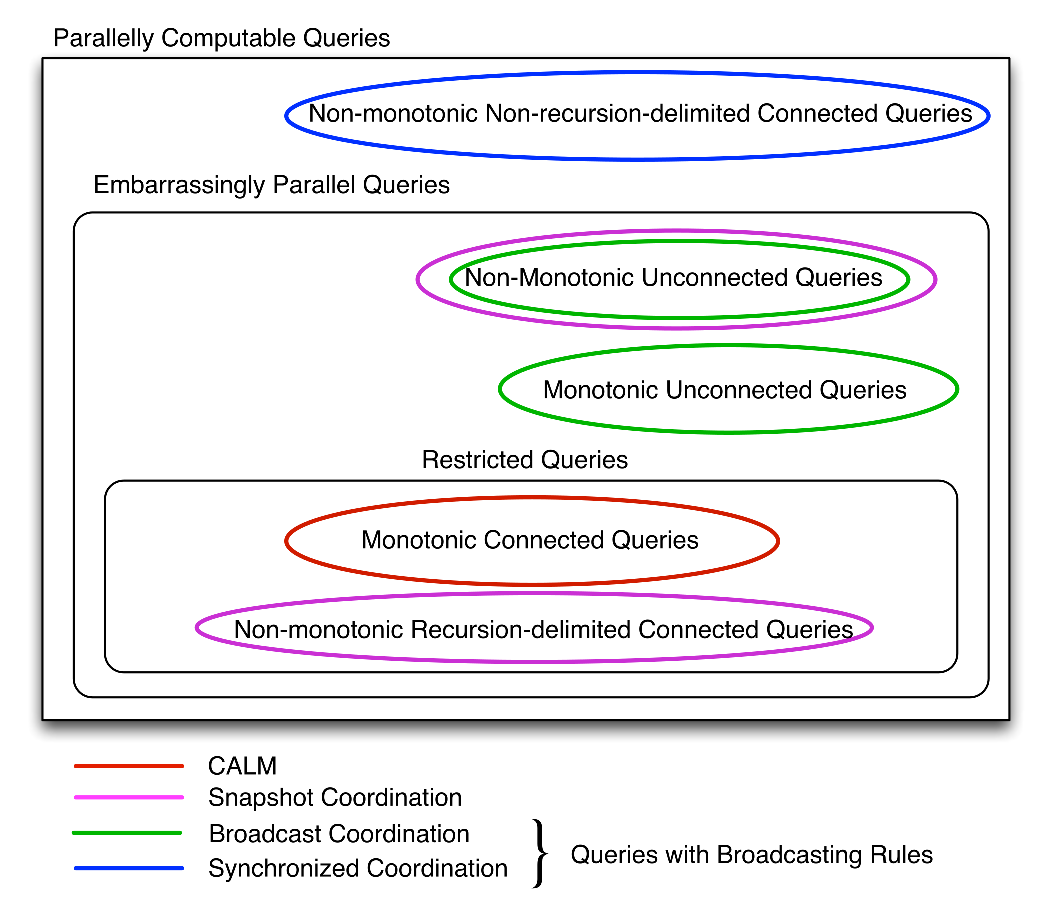}
\caption{Complete Analysis of Parallelly Computable Queries}
\label{fig:hierarchy_queries_new}
\end{figure}

An interesting difference among the three coordination patterns is that the first two depend on the system's semantics, while synchronized coordination is tolerant to any changes over the distributed system model. In the next section we will hence show how the semantics of snapshot and broadcasting coordination patterns changes if we weaken the constraint of the system definition. So far, in fact, we have only considered the deterministic delivery model (\cf Section~\ref{sec:synchronous_systems}), \ie tuples arrive exactly after $\Delta$ physical time once emitted. But what happens if we assume less constrained systems, \eg in MapReduce we start to pipeline Reducers with Mappers? Below we first take a look at systems with non-deterministic delivery but bounded delay (Section~\ref{sec:async_del}, Mappers and Reducers tasks can be pipelined but in a single MapReduce step), and then conclude with \emph{rsync} systems, \ie systems with non-deterministic delivery and arbitrary, finite delay (Section~\ref{sec:generic_snapshot_delay}, Mappers and Reducers are full pipelined).

\subsection{Coordination and Non-deterministic Delivery}
\label{sec:async_del}

%We start with defining the $\Delta$ quantity.
%In \emph{rsync} systems, $\Delta$ can be seen to be composed of two quantities: the \emph{network latency} $\delta$, and a \emph{synchronization delay} that expresses the time required by the system to detect that every node has reached the fixpoint.
In Section~\ref{sec:synchronous_systems} we have seen that \emph{bsp} systems assume that all emitted messages arrive exactly after $\Delta$ physical time (i.e., condition \textbf{S3}$^\prime$). In this Section we will instead assume that messages arrive non-deterministically within the $\Delta$ bound. 
Under this weaker condition we are not any longer certain about when successive rounds can start: if we let rounds start after $\Delta$ physical time (as under \emph{bsp} systems) we may spend unnecessary time waiting; conversely, if we start the next round right after all nodes have finished the current round, (i.e., before $\Delta$ time has elapsed, hence before all messages are received with certainty), we may receive late facts and eventually have to retract wrong deductions.  

%By construction, the network latency is considered twice in the definition of $\Delta$, once from the sending node to the environment, and once from the environment to the receiving node.
Consider now  an invertible function $\theta$  mapping  each round number to the \emph{physical time} in which it occurs, and two values $\Delta_{max}$, $\Delta_{min}$ representing respectively the maximum and the minimum network latency of the given physical system.
Let us simplify our model by substituting property \textbf{S3}$^\prime$ with the following constraint, named \emph{bounded delay}:
\begin{description}

\item[\textbf{S3}$^{\prime\prime}$] Let $del = \Delta_{max} - \Delta_{min}$, with  $del \ll \theta(t+1) - \theta(t)$ for every pair of rounds $t, t +1$. 
%\item[S3$^\prime$] for every pair of consecutive rounds $s, s^\prime$, $2\delta \ll \frac{\theta(s^\prime) - \theta(s)}{3}$.

\end{description} 				%\delta < s^\prime - \Delta - s + \epsilon
%\delta < s^\prime - 2 \delta - s ----> communication < computation time + synchronization time \noindent Informally, condition $\textbf{S3}^\prime$

\noindent \textbf{S3}$^{\prime\prime}$ specifies that, between two consecutive rounds, the variance of the communication delay is amply lower than the time spent for computation. %between two consecutive rounds. %end of one round and the start of the consecutive one.
From the above assumptions it follows that each tuple, derived by a send query at round $t$, will be available at the receiving site no later than the physical time $\theta(t+1) + del$, and that $\theta(t +1) \leq \theta(t+1) + del < \theta(t + 2)$, \ie facts are received during the successive round.
Note that although the delay is bounded, and hence we are assured that every fact is delivered during the successive round, the actual instant in which a fact is received falls non-deterministically in the range $[\theta(t+1)$, $\theta(t+1) + del]$. 
Henceforth, we will then use  $\mathcal{S}^{\emph{bsp-d}}$ to denote a \emph{rsync} system with \emph{bounded delay and non-deterministic delivery} (\emph{bsp-d}). Figure~\ref{fig:bsp-d} depicts how \emph{bsp-d} systems behave, which is in line with frameworks such as MapReduce online~\cite{CondieCAH10}, where key-value records are pipelined between Map and Reduce operations.

\begin{figure}[bt]
\centering
\includegraphics[width=0.45\columnwidth]{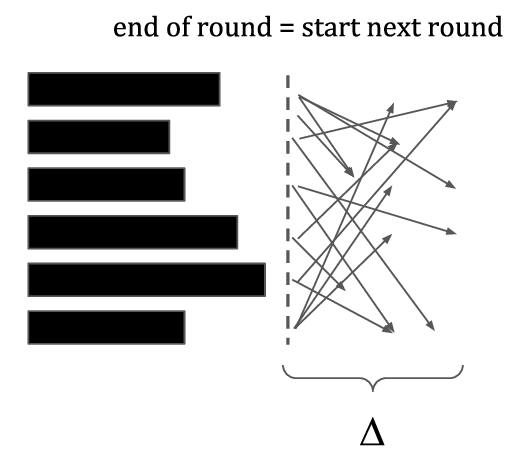}
\vspace{-3mm}
\caption{\emph{bsp-d} system computation model. Differently than \emph{bsp} systems, the next round starts right after the end of the current one. Additionally, data communication may take arbitrary (although bounded) time.}
\label{fig:bsp-d}
\end{figure}

Given a synchronous transducer network $\mathcal{N}$, if we assume a synchronous system with non-deterministic delivery and bounded delay$\mathcal{S}^{\emph{bsp-d}}$, we have that different behaviors arise based on the kind of transducer program. %defining $\mathcal{N}$.
In fact, let us consider first the case in which the program is monotonic: in this circumstance no wrong result can be derived by definition, even if an emitted fact is received after the round has already started. 
Therefore, monotonic transducer networks behave equivalently under \textit{bsp-d} and \textit{bsp} systems.
The same thing cannot be stated for non-monotonic programs, as the next example shows.

\begin{example}
Consider the transducer of Example \ref{ex:false_calm} computing the emptiness query. 
Under non-deterministic delivery, it is not clear when negation can be safely applied over $S$.
For instance, if negation is immediately applied, it may happen that a previously sent fact 
appears later, therefore invalidating the derived results.
%Intuitively, if every node is able to physically wait $\delta$ time, by construction we are assured that the content of $S$ is sealed. But what if this is not a viable solution? % can this be logically implemented in our framework? 
\end{example}

In synchronous systems with non-deterministic delivery we have that, in general, snapshot coordination is no more achievable ``indirectly'', without exchanging any message; under this model we can therefore appreciate more in detail the nature of snapshot coordination.

In order to explain how snapshot coordination can be implemented in \textit{bsp-d} systems, we first solve the problem under the constraint that communication\footnote{By communication here we mean both the nodes' local buffers and the actual communication medium.} implement \emph{First In First Out} (\emph{FIFO}) delivery\footnote{Although non-deterministic delivery may look impractical with FIFO communication, this is a simplifying assumption by which facts are communicated in sequential order but where a bounded delay may occur between consecutive facts.}, and  then consider the general case. %in which no assumption is considered on the communication medium.
As a first step, we introduce how \emph{monotonic and non-monotonic (stratified) aggregates} can be used in queries.

\subsubsection{Queries with Aggregates}
\emph{Aggregate relations} are usually employed in query languages to express \emph{aggregate queries}.
In the next subsections we will use aggregate relations in positive rule-heads and in the form $R(\Lambda <\bar{w}>)$, with $\Lambda$ one of the usual aggregate functions, and $\bar{w}$ a set of variables from the body %and $\bar{u}$ a list thereof called the grouping variables 
\cite{RamakrishnanU95}.
Aggregate relations appear in  the heads of \textit{aggregation rules}: %having the following form: %(here we show a deductive aggregation rule but the same applies for inductive aggregation rules):
\begin{align}
R(\Lambda <\bar{w}>) \leftarrow B_1(\bar{u}_1), \ldots, B_n(\bar{u}_n).
\end{align}
%
%where $\Lambda$ is an aggregate function, %$\bar{u}$ is a list of variables from the body, 
%and $\bar{w}$ is a list of variables from the body. %possibly belonging to the list $\bar{u}$.
If we denote with %$\bar{\texttt{u}}$ a ground assignment for $\bar{u}$, and with 
$\bar{\texttt{W}}$ the finite multi-set containing all the existing ground assignments of $\bar{w}$ which %with $\bar{\texttt{u}}$, 
satisfy the body of the rule, we have that $R(\texttt{a})$ is true, where $\texttt{a} = \Lambda <\bar{\texttt{W}}>$. That is, \texttt{a} is the result of the application of $\Lambda$ to the multi-set $\bar{\texttt{W}}$ \cite{BeeriNRS87}.

We consider two different types of aggregate predicates: usual \emph{stratified aggregates}, and \emph{monotonic aggregates} \cite{MazuranSZ13}.
For the former, they are stratified and hence the entire body must the stable (no holes are allowed in the local knowledge base) before the aggregate function can be applied. 
We then always assume the aggregate predicates to depend negatively on every predicate composing the body.
This assumption is quite natural since head-aggregation rules can be easily rewritten as body-aggregation rules, which, in turn, can be specified using the stratified semantics of {\sc datalog}$^\neg$ with built-in relations \cite{MumickS95}.
%The reader can then appreciate  the origin of the appellative ``stratified", and also why such type of aggregation is not allowed to appear in any recursive rule.

For what concerns the latter type of aggregates (the monotonic ones), since aggregation is monotonic, it is, for instance, allowed to appear in recursive rules.
In order to differentiate between monotonic and stratified aggregation functions, we label the former with the $m$ prefix. 
While we will not explain in detail the semantics behind monotonic aggregation -- we suggest the interested reader to refer to Mazuran's paper or~\cite{DBLP:conf/amw/ZanioloYDI16,DBLP:journals/tplp/ZanioloYDSCI17,DBLP:conf/amw/ZanioloYIDSC18} for more recent evolvements-- we nonetheless remark here the main operational differences between the two types of aggregation: 
stratified aggregation always returns a single value, which is the application of the function $\Lambda$ over the stable multi-set $\bar{\texttt{W}}$ once its computation is terminated; for monotonic aggregation, whenever the system is fed with new tuples, new values are returned forming a monotonically evolving distribution.
For instance, if $m\_max<w>$ returns the max value of the term $w$, every time a new tuple is generated defining a new max value for $w$, $m\_max<w>$ will return it. 
Conversely, $max<w>$ will return the single maximum value for the stable multi-set $\texttt{W}$.

\subsubsection{Snapshot Coordination under FIFO}
\label{sec:FIFO}

Under the FIFO assumption we have that tuples are received in the same order in which they are locally derived~\footnote{Note that this does not mean that tuples that are derived by different nodes in a certain global order are also received in the same order. FIFO can therefore be seen as enforcing a partial order.}.
Recall that snapshot coordination -- implementing the SCWA -- is used to ascertain that a relation is sealed. %and hence the PCWA assumption can be safely applied.
We can reduce the problem of detecting the sealed state of a relation to the problem of detecting a global stable property in a distributed system \cite{Babaoglu:1993:CGS:302430.302434}, and therefore apply one of the well-known snapshot protocols working under FIFO \cite{Chandy:1985:DSD:214451.214456}.

Let $\mathcal{T}$ be a transducer used to parallelly compute a non-monotonic, recursion-delimited query $\mathcal{Q}$. %as deFSRibed in Section \ref{sec:generation}.
In Section \ref{sec:coordination_sync} we have seen that $\texttt{NULL}$ messages are implicitly derived by send queries under the deterministic delivery assumption; by contrast, under non-deterministic delivery it might be necessary to explicitly send $\texttt{NULL}$ messages. %if we cannot wait $\delta$ time before evaluating non-monotonic rules.
%Assume we have a built-in relation \texttt{bot}$^{(1)}$ containing a special $\bot$ constant not otherwise appearing in the active domain of the initial instance. 
Consider a relation $R$ occurring negated in the body of a rule in $\mathcal{Q}$. %defining predicate $P$.
Let $\mathbf{B}$ be the body of the query $q_R \in Q_{snd}$ emitting $R$.
%Let then $\chi(R)$ and $\chi(P)$ be two strictly sequential stages.
%In $\mathcal{T}$, $R^\prime$ will be the head of an emission rule at stage $i$ \footnote{W.l.o.g. we can use $i = 0$ to express the case in which no stage is active, \ie when the extensional predicates are hashed.}, either if $R$ is an extensional or an intensional relation in $\mathcal{Q}_{out}$.
%Let $\textbf{B}$ the body of such rule. %and where the $\texttt{bot}$ predicate doesn't appear.
We can add to $Q_{snd}$ the following rules defining respectively a unary \emph{stratified aggregate} relation $\texttt{CntR}^{(1)}$, 
and a new unary $communication$ relation $\texttt{SealR}^{(0, 1)}$ emulating the $\texttt{NULL}$ message for $R$:
\begin{align}
\label{eq:NULL_FIFO_1}
&\texttt{CntR}(count <\bar{u}>) \leftarrow \mathbf{B}.\\
\label{eq:NULL_FIFO_2}
&\texttt{SealR}_{snd}(i) \leftarrow \texttt{CntR}(u), \texttt{Id}(i).
\end{align}

\noindent In this way, %$\mathcal{T}$ becomes a datalog$^{\neg s}$ program, and 
by exploiting the stratified semantics, each node $i$ can send the $\texttt{NULL}$ message for relation $R$ once the computation of the count of the number of tuples in $R$ is completed.
Since count is a stratified aggregate, \texttt{CntR} -- and therefore also \texttt{SealR} -- belongs to a higher stratum.  
In this way we are assured that the $\texttt{NULL}$ message is emitted after the instance over $R$ has been  completed.
%Remark that the $\bot$ tuple will never appears in $I_{R^{buf}}$ since \texttt{bot} is not allowed in \textbf{B}.
%Rule (\ref{eq:\texttt{NULL}_2}) is then always satisfied.
Under the FIFO semantics we are then guaranteed that once a node receives a \texttt{SealR} tuple, the content of $R$ is sealed for what concerns that emitting node.
This clearly does not mean that $R$ is globally sealed, since a tuple produced by a different node can still be floating.
To have the SCWA hold on $R$, a node must have received a number of $\texttt{NULL}$ messages equal to the number of nodes composing the network: \ie negation is applied on a stable snapshot of $R$.
To obtain this, we can add to $\mathcal{T}$ the rules:
\begin{align}
\label{eq:snapshot_1}
&\texttt{CntSlR}(m\_count<u>) \leftarrow \texttt{SealR}(u).\\
\label{eq:snapshot_2}
&\texttt{CntAll}(count < u >) \leftarrow \texttt{All}(u).\\
\label{eq:snapshot_3}
&\texttt{FSR}() \leftarrow \texttt{CntSlR}(u), \texttt{CntAll}(u).
\end{align}

\noindent and we attach the final seal $\texttt{FSR}()$ to the queries in which $R$ is negated.
%Here we have employed two different types of aggregate functions: stratified aggregates, with the usual semantics, and \emph{monotonic aggregates} as defined in \cite{MazuranSZ13}, labeled with the $fs$ (\emph{frequency support}) prefix.
%Differently from stratified aggregates always returning a single value -- \ie the application of a function over a stable multi-set of values -- for monotonic aggregation, a monotonically evolving distribution of values is instead returned, based on the evolution of the instances defining the body.
%In order to differentiate between monotonic and stratified aggregation functions, we label the former with the $fs$ (\emph{frequency support}) prefix as specified by Zaniolo \cite{DBLP:journals/vldb/MazuranSZ13}.
%Finally we bind every term $\texttt{u}_k$ of $R^\prime$ to the negative literal $\neg \texttt{\texttt{NULL}}(\texttt{u}_k)$ in order to note propagate further $\texttt{NULL}$ messages .
Queries (\ref{eq:snapshot_1}) - (\ref{eq:snapshot_3}) are used to define when \texttt{FSR} is true, \ie when $n$ $\texttt{NULL}$ messages have been received for relation $R$, with $n$ the number of nodes in the network.
Once \texttt{FSR} is true, negation can be safely applied over $R$, so the related query can be evaluated (if no other negative literal appears in the same query).
Note how we have employed monotonic and stratified aggregates: since we do not know when \texttt{SealR} is stable, we cannot apply to it stratified aggregates nor negation, while we can definitely use a stratified aggregate over \texttt{All}.
Interestingly, just by moving from a \emph{bsp} system to a \emph{bsp-d} system, both $system$ relations must be employed to implement snapshot coordination, and non-monotonic, connected, recursion-delimited queries are no longer embarrassingly parallel.
This is consistent with \cite{AmelootNB13}: non-monotonic queries are neither coordination- nor communication-free, and both \texttt{Id} and \texttt{All} relations are required.
In \emph{bsp-d} systems, we then have that syncausality degenerates into the Lamport's happens-before relation \cite{Lamport:1978:TCO:359545.359563}.
%We therefore have that the coordination pattern arises only when an instance exists which is delivered to every other node composing the network.
Figure \ref{fig:hierarchy_queries_ndet} depicts this new situation in which snapshot coordination code is injected into non-monotonic specifications.

\begin{figure}[h]
\centering
\includegraphics[width=0.6\columnwidth]{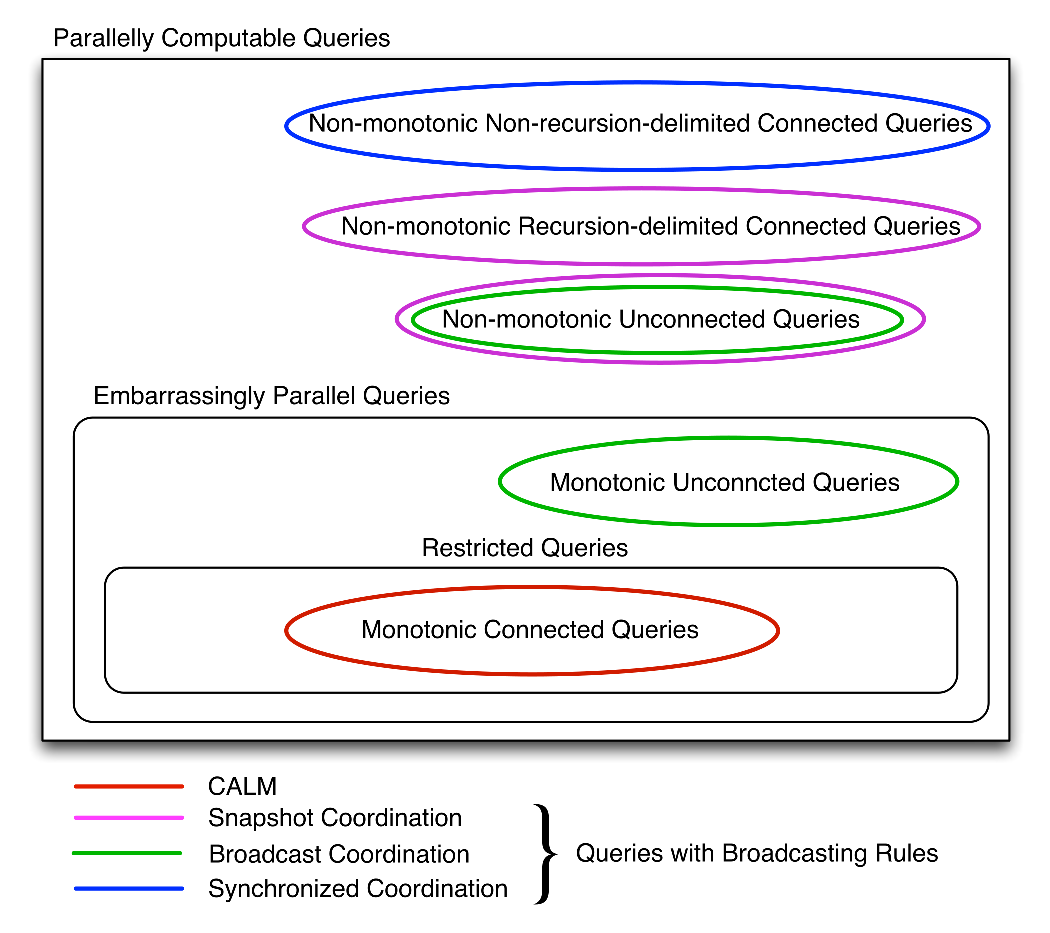}
\caption{Taxonomy for Parallelly Computable Queries under \emph{bsp-d}}
\label{fig:hierarchy_queries_ndet}
\end{figure}\vspace{-3ex}

%Finally, we want to highlight that the above, as well as the next algorithm for achieving snapshot coordination must be injected in a transducer program only if $\chi(R)$ and $\chi(P)$ are strictly sequential stages and $R$ appears negated in a rule defining $P$.
%If this is not the case, \ie $R$ is defined at stage $i$ and appears negated at stage $j$, with $j > i+1$, still the indirect sending of $\texttt{NULL}$ messages can be employed even if we are in asynchronous delivery settings.
%In fact, by assumption, we have that $\Delta$ is amply lower that the duration of a super-step, so by stage $j$, the content of $R^\prime$ is certainly stable.

\subsubsection{Generic Snapshot Coordination}
\label{sec:generic_snapshot}

If we drop the FIFO assumption, we can end up in a situation in which a $\texttt{NULL}$ message is received before a regular (informative) fact, therefore negation can end up being applied to a non-stable relation.
Indeed this problem is related to how a partial order of events can be enforced in distributed settings \cite{Lamport:1978:TCO:359545.359563}.
However, in our specific case, we are not interested in a complete ordering of the emitted tuples over $R$, but, instead, we just want to be able to state that \texttt{FSR} is true when $n$ $\texttt{NULL}$ messages have been received \emph{and} no other tuple over $R$ is still floating.
More concretely, we are interested in implementing just the \emph{gap-detection property} \cite{Babaoglu:1993:CGS:302430.302434} of ordered events, that is,  we want to be able to determine if, once an event (the $\texttt{NULL}$ message)  is received, there is some other event (sent tuples) happened before it, which has not been received yet.
Negation cannot, in fact, be applied until we are not guaranteed that our knowledge base has no gap. %since the CWA does not hold in such situation.
To implement this, query (\ref{eq:NULL_FIFO_2}) can be modified as follows:
\begin{align}
\label{eq:NULL_1}
%&R^\prime\_cnt(count <\bar{u}>) \leftarrow \textbf{B}.\\
%\label{eq:\texttt{NULL}_2}
&\texttt{SealR}_{snd}(i, u) \leftarrow \texttt{CntR}(u), \texttt{Id}(i).
\end{align}

\noindent where $\texttt{SealR}$ is now a binary relation containing also the number of tuples originally emitted for $R$.
Before applying (\ref{eq:snapshot_1}) - (\ref{eq:snapshot_3}) we have to ensure that the number of tuples over $R$ is equal to the number of tuples originally sent.
We then add to $\mathcal{T}$ the clauses:
\begin{align}
\label{eq:stable_1}
&\texttt{CntR}(m\_count <\bar{u}>) \leftarrow R(\bar{u}).\\
\label{eq:stable_2}
&\texttt{SmNR}(m\_sum<u>) \leftarrow \texttt{SealR}(i, u).
%\label{eq:stable_3}
%&\texttt{SmNR}(\texttt{0}).
\end{align}

\noindent counting the number of tuples in $R$ and the total number of tuples over $R$ derived globally, and finally we modify eq. (\ref{eq:snapshot_3}) as follows:
%\vspace{-1mm}
\begin{align}
\label{eq:snapshot_4}
&\texttt{FSR}() \leftarrow \texttt{CntSlR}(u), \texttt{CntAll}(u), \texttt{SmNR}(v), \texttt{CntR}(v).
%\tag{6.a}
\end{align}
%\vspace{-4.5mm}

\noindent In this way we are ensured that negation can be applied over $R$ only if the proper number of $\texttt{NULL}$ messages is received and, at the same time, all the emitted $R$-facts have also been received.

%\textbf{Remark:} Note how our implementation of snapshot coordination makes precise the intuition that (non-monotonic) coordination requires waiting and that ``waiting requires counting" \cite{Hellerstein:2010:DIE:1860702.1860704}.

\subsection{Coordination under Arbitrary Delay}
\label{sec:generic_snapshot_delay}

If now we assume that condition $\textbf{S3}^{\prime\prime}$ does not hold, we are into the initial \emph{rsync} semantics whereby delays can be arbitrary long although finite.
%\eat{Also in this case, as for \textit{ndet} systems, the possibility of just waiting $\delta$ time before computing negative literals could be still a practical viable solution. 
%Thus the injection of code for enforcing snapshot coordination can still be avoided if we are able to physically constraint the system to wait $\delta$ time before evaluating negative predicates. }
%In fact, under unbounded delay, a node can ends up in waiting forever before being sure that applying negation to a predicate is safe.
%This because a fact emitted at round $t$ is still delivered at time $\theta(t+1) + \delta$, but, since $\delta$ is now unbounded, we are no more assured that $\theta(t+1) \leq \theta(t+1) + \delta < \theta(t+2)$.
%This of course does not intact the correct evaluation of monotonic queries, while we are now forced to inject the coordination code deFSRibed in the previous section into every non-monotonic specification.
Surprisingly, in this situation we have that %while we are still forced to inject the coordination code deFSRibed in the previous section into every non-monotonic specification, 
monotonic unconnected queries become coordination-free.
To see why this is the case, first remark that a fact emitted at round $t$ is still delivered at most at time $\theta(t+1) + del$, but, since $del$ is now arbitrary, we are not  assured that $\theta(t+1) \leq \theta(t+1) + del < \theta(t+2)$ any more.
Despite this, the notion of coordination %(and also of direct/indirect potential causality relation) 
still maintains its semantics, even if, in this case, the coordination pattern may span multiple rounds. \eat{\footnote{To clarify this point in respect to the definitions of Section \ref{sec:refine_coordination}, note that every time we write $s+1$ we actually meant $\theta^{-1}(\theta(s+1) + var)$ when communication is involved ($s+1$ is merely a simplification in the notation, as mentioned in Section \ref{sec:synchronous_systems}).}.}
Let $\mathcal{N}$ be a specification parallelly computing a monotonic unconnected query, \textbf{I} an instance, and $(N, D, H)$ a non-trivial configuration.
Under the \emph{rsync} semantics %\footnote{Note that a similar argument also applies for \emph{ndet} systems.} 
we have that the system $\mathcal{S}^{\emph{rsync}}_{\mathcal{N}}(N, D, H, \textbf{I})$ is composed by multiple convergent runs, modeling the fact that sent tuples can be non-deterministically received in different rounds.
%one for each possible round $\theta^{-1}(\theta(s+1) + var)$ in which a fact emitted at round $s$ is actually received.
A configuration $(N, D, H)$ can then be chosen -- \eg the one where $D$ installs the entire initial instance \textbf{I} on every node -- so that no coordination pattern arises because the final state is already reached without having any broadcasted fact been received. %beside the local one.

\begin{example}
Consider the monotonic unconnected query of Example \ref{ex:not_connected}.
Assume a \emph{rsync} specification $\mathcal{N}$ defined as in Example \ref{ex:not_connected} but where now $S$ is hashed on both attributes, $U$ is broadcasted, and $\mathcal{Q}$ contains also the query:
\small
$Q_{out}(u, v) \leftarrow S(u, v), T(\_)$.
\normalsize
Consider the non-trivial configuration $(N, D, H)$ in which $D$ installs the full instance \textbf{I} %(that we leave defined as in the previous example)
on every node, while $\mathcal{H}(I_S) \subset N$ -- \ie the instance over $S$ is not hashed to every node.
We then have that a run $\rho \in \mathcal{S}^{\emph{rsync}}_{\mathcal{N}}(N, D, H, \textbf{I})$ exists such that every fact emitted over $S$ at round $t$ is received by round $t^\prime \leq *$, while every fact over $U$ that should be sent to a node in $N \setminus \mathcal{H}(I_S)$ is received in a round $t^{\prime\prime} > *$.
Clearly, we still have that the correct output is returned since $I_T$ consists of at least one fact, and every sent $S$-tuple has been correctly received.  The class defined by $\mathcal{N}$ is %coordination pattern does not apply, \ie $\mathbfcal{NC}$ is 
coordination-free.
\end{example}

We can therefore conclude that the \emph{if} direction of the original CALM principle is fully satisfied under the \emph{rsync} semantics since every monotonic ({\sc datalog}) query can now be computed in a coordination-free way.
One can also show that indeed also the \emph{only-if} direction is satisfied.
The reader can now completely appreciate how the notion of coordination we introduced perfectly aligns with the one of \cite{AmelootNB13} when arbitrary delay comes into play (Contribution 6): embarrassingly parallel queries are all coordi\-nation-free. %and hence CALM and EPIC coincide.
%This is exactly the definition of coordination-freedom introduced by Ameloot.
Nevertheless, our definition is more general since it can be seamlessly used in both synchronous and asynchronous systems.

% if a run exists in which no fact is ever delivered, and, despite that, the specification is still convergent, the specification is communication-free.
%Clearly, under the \emph{rsud} semantics, if a specification is communication-free, it cannot have coordination.
%Unchained monotonic queries are intuitively still communication-free in unbounded delay systems, therefore they are coordination-free.
%We then have that every embarrassingly parallel query is coordination-free, and the CALM and EPIC Theorems coincide.
%Note that this is exactly the definition of coordination-freedom introduced by Ameloot.
%We can then state that our definition of coordination is more general then the one introduced by Ameloot since it can be seamlessly also used both in the synchronous and asynchronous case.

\eat{\subsection{Coordination under Unbounded Delay}

To complete the picture, if now we consider \emph{rsud} systems, \ie system when just conditions \textbf{S1} - \textbf{S2} and \textbf{R1} - \textbf{R2} are satisfied, we are forced to inject code to enforce snapshot coordination.
In fact, a node can ends up in waiting forever before being sure that applying negation to a predicate is safe.
This because a fact emitted at round $s$ is still delivered at time $\theta(s+1) + \delta$, but, $\delta$ this is now unbounded. %we are no more assured that $\theta(t+1) \leq \theta(t+1) + \delta < \theta(t+2)$.
%: communication is a necessary condition to achieve coordination, therefore the two notions merges when unbounded delivery comes into play.
In systems with unbounded delay, we have that syncausality degenerate into the happen-before relation defined by Lamport \cite{Lamport:1978:TCO:359545.359563}.
We therefore have that the coordination pattern arises only when an instance exists which is delivered to every other node composing the network.
Conversely, the coordination pattern doesn't arise when a run exists such that the related specification is still convergent even if all broadcasted fact are not yet received.
}

Figure \ref{fig:hierarchy_queries_rsync} shows the new taxonomy when \emph{rsync} systems are considered. %\footnote{Note that the same conclusions (and chart) can be drown also for \emph{rsync} systems if the injection of coordination code is preferred to waiting $\delta$ time.}.

\begin{figure}[h]
\centering
\includegraphics[width=0.6\columnwidth]{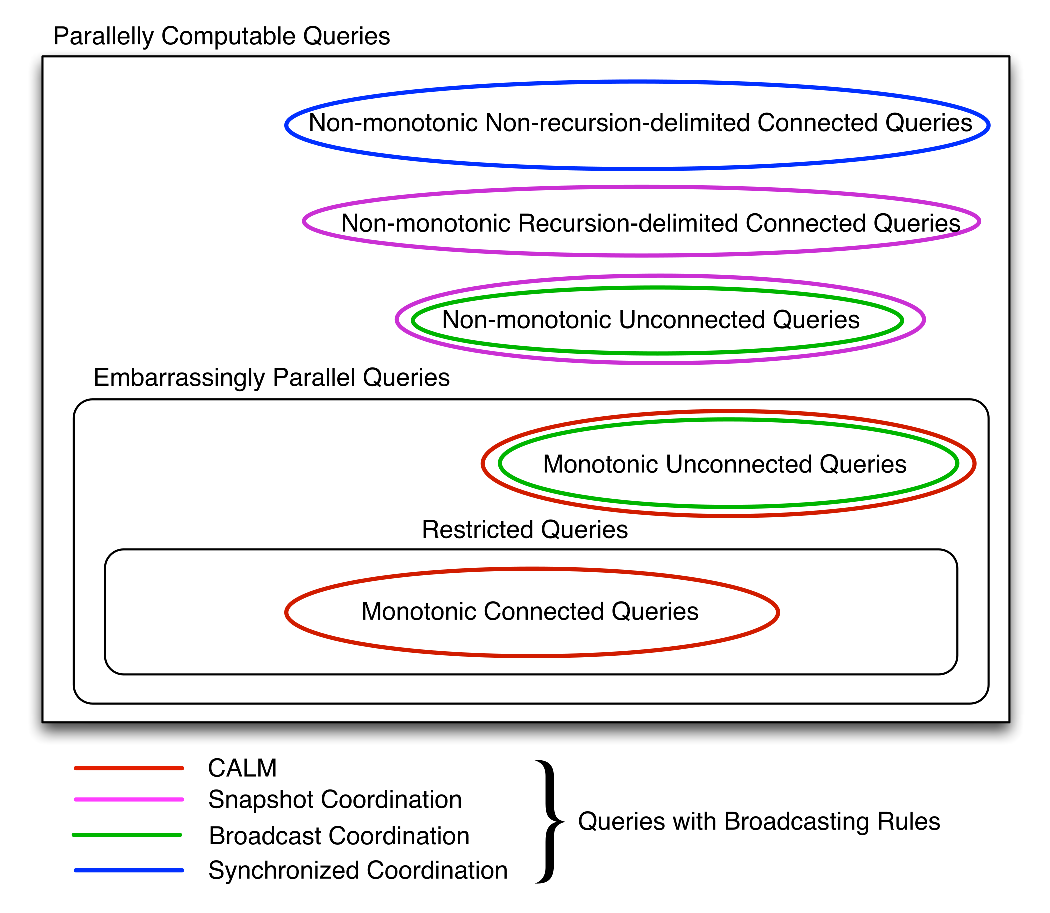}\vspace{-1ex}
\caption{Taxonomy for Parallelly Computable Queries under \emph{rsync}}
\vspace{-2ex}
\label{fig:hierarchy_queries_rsync}
\vspace{-1ex}
\end{figure}

\stitle{Remark}: From a states of knowledge perspective, in \emph{bsp} systems (respectively \emph{bsp-d} - \emph{rsync}), \emph{common knowledge} (\emph{$\delta$-com\-mon knowledge}) can be obtained by simply using broadcasting \cite{FaginHMV03}.
%If we remove the bounded delay constraint, common knowledge can be obtained through broadcasting just eventually.
However, if the final outcome is returned before $\delta$-common knowledge is reached, the former was computed without coordination. %This is possible only under the \emph{rsync} semantics.
For what concerns non-monotonic queries, the result of a query can be correctly computed only if the stability of the negated predicates is common knowledge among the nodes in the network. 
This highlights the main difference between broadcasting and snapshot/synchronized coordination: the former exists in \emph{bsp} - \emph{bsp-d} just because of the tight requirements imposed on the system; the latter are required by the actual semantics of the query.
\vspace{-2mm}

%The difference between $\delta$-common knowledge and eventual common knowledge highlights the essence of our definition of coordination: since in the latter common knowledge will be reached eventually without a strict timing requirement

%Finally note that from systems with arbitrary but bounded delay, there is no more a net difference, semantically, in employing broadcasting rules or not.

\section{Comparison with Other Work}
\label{sec:related}
%!TEX root = TPLP.tex

In this period we are witnessing new trends such as \emph{cloud computing} and \emph{multicore processing} becoming popular. %Hardware speed is not increasing any more, therefore distributed and parallel architectures must be exploited in order to obtain performance improvements. 
It is well-known that programming such architectures is  very difficult, thus %, and specialized techniques and experts are needed.
%ramming languages.
\emph{declarative networking} has been proposed to simplify such task. %of programming distributed systems. 
The idea  of declarative networking is to  use  high-level, declarative languages%instead of imperative ones
,  leaving to the system the burden of organizing an efficient execution plan \cite{Ameloot14}.
In this paper we propose to apply the same techniques to synchronous systems, in order to set forth the theoretical basis also for parallel datalflow optimizations.
This application was first identified by Hellerstein \cite{Hellerstein10}, who
also pointed out that a tradeoff exists between efficiency of pipelining and fault-tolerance provided by full materialization and that,
however,  the run-time should decide which of the two strategies must be  selected.
Similarly, decomposable plans were identified in the 80's to speed up the evaluation of {\sc datalog} programs through parallel execution~\cite{DBLP:conf/sigmod/WolfsonS88}.
The general pivoting technique~\cite{generalizedpivoting} implemented in BigDatalog~\cite{Shkapsky:2016:BDA:2882903.2915229} provides only a sufficient condition to determine if a program is decomposable; decomposability is in fact undecidable in general~\cite{DBLP:conf/sigmod/WolfsonO90}.
Clearly, a relation exists between decomposable programs and programs that distribute over components: every decomposable program distributes over components, while the opposite is not true. For instance, the following program distributes over components, but is not decomposable~\cite{DBLP:conf/sigmod/WolfsonO90}:\vspace{-1ex}
\begin{equation*}
\begin{split}
&Q(u, v) \leftarrow R(u, w), E(w, x), F(x, v).\\
&R(u, v) \leftarrow G(u, v).
\end{split}
\end{equation*}
%Clearly, a relation exists between decomposable programs and programs that distribute over components: every 
%Interestingly enough,  {\sc datalog}$ programs exists which distributes over components while not decomposable, e.g., the following~\cite{}:  
%
%while 
%A similar concern, namely when coordination must be used for achieving correct transactions in parallel data management systems, has been investigated in \cite{BailisDFG13,BailisFFG14}.
%How our model can be augmented with (coordination-free) updates is among our envisioned future works.
\cite{AmelootGKN15} consider a superset of decomposable plans called \emph{parallel-correct}, where queries are parametrized by a \emph{distribution policy} and are allowed to generate not-unique facts. 

In \cite{AfratiBCP11}, the authors study how the MapReduce model can be extended with recursion. Additionally, a model is proposed suggesting that the optimal computation time can be obtained by minimizing the volume of data passed as input to each task. This work is orthogonal to our contribution: in fact, while we focus on how a property of queries (coordination freeness) could be used to optimize queries, \cite{AfratiBCP11} mainly focus on adding support for recursion to a MapReduce framework. In our {\sc datalog} implementation of~\cite{Shkapsky:2016:BDA:2882903.2915229} we also proposed a better way to support recursion in Apache Spark. From our practical experience we found that the best way to implement transitive closure is through a decomposable plan which not necessary is optimal from a data-volume perspective since the full input dataset is passed to each task in each iteration.

The fact that CALM does not hold in general in \emph{rsync} systems was first suggested in~\cite{DBLP:conf/sigmod/WolfsonO90}\footnote{More precisely, they identified that a class of non-monotonic program exists that is communication-free.} and only recently, with the advent of parallel processing systems such as MapReduce and Spark, revamped by \cite{InterlandiT15} for distributed parallel settings.
From the latter work we borrow the basic techniques we used to build the hashing transducer network model and our notion of coordination-freedom.
%Recently, an increasing number of works have investigated how to push the boundary between batch and online processing \cite{DBLP:journals/pvldb/BoykinROL14,Murray13MII,Condie:2010:MO:1855711.1855732}.
%None of them, however, have identified which class of batch computation can actually be computed in streaming, without changes to the program semantics.
%Relational transducers have been first employed in \cite{Abiteboul:1998:RTE:275487.275507} as declarative model to express electronic commerce scenarios where a web-based application is assisted by a relational database.
%In relational transducers, the state of the application is maintained by the database, and the events generated by the interaction with the web interface and the output of the computation are specified by a set of input and output relations.
%In order to record the mappings between input and output sequences, a notion of \emph{log} is employed.
Our computational model merges the original transducer network model of \cite{AmelootNB13} -- representing how distributed computation is carried out by an asynchronous system -- with the BSP model of \cite{Valiant90}.
%We chose to embed BSP into the transducer network model in order to be able to formally represent computation in data-parallel systems, since many Big Data frameworks can be considered directly \cite{Malewicz:2010:PSL:1807167.1807184} or indirectly \cite{Pace2012246} as  implementations of the BSP model. %and hence, in this way, we are able to formally represent computation in data-parallel systems. 
%We are hence going to define how we differentiate from this two models.
%Until 2011, relational transducers were mainly used as a declarative tool for specification and verification of web-based business processes and services \cite{Spielmann:2003:VRT:846156.846161,Deutsch:2004:SVD:1055558.1055571,Deutsch:2009:AVD:1514894.1514924}.
%In \cite{Ameloot:2011:RTD:1989284.1989321} the notion of transducer network was coined. 
%Ameloot employed such declarative notion to represent how distributed computation is carried out by an asynchronous system.
We differ from the original transducer network model both semantically -- our definition of global transition implements a synchronous and reliable communication model -- and structurally -- we have $(i)$ an input clock driving the computation% and acting as a global clock
, and $(ii)$ a special environment transducer modeling everything not functionally related with the system. 
Additionally, our bounded-delay and asynchronous BSP models are somehow related to the \emph{Stale Synchronous Parallel} (SSP) and A-BSP models introduced in \cite{Cui:2014:EBS:2643634.2643639}.
We borrow the concept of environment from the multi-agent systems domain \cite{FaginHMV03}.
Although in contexts different than ours, synchronous transducer networks were also employed in \cite{FurcheGGG14,InterlandiTB13}.
Another interesting model related to ours is the \emph{Massive Parallel Model} (MP) of \cite{Koutris:2011:PEC:1989284.1989310}. 
In MP, each round is divided into three phases: the usual computation and communication phases, and a broadcasting phase. %The difference between the communication and broadcasting phase is that the former emploies hashing functions to address data to the proper node. 
As we have demonstrated, in parallel systems broadcasting implements coordination, therefore MP expresses exactly those queries that require coordination in order to proceed.
%Without the broadcasting phase, MP expresses instead coordination-free queries. 
Koutris et al. showed that by employing their model, a specific class of chained conjunctive queries, denoted  \emph{tall-flat}, can be computed in one round by a load-balanced algorithm. 
Conversely, if a query is not tall-flat, then every algorithm consisting of one round is not load-balanced. 
This work as well as~\cite{AmelootGKN15} focus on how to efficiently execute queries in parallel. Conversely, our main focus is on how to extend CALM over \emph{rsync} systems in order to unlock asynchronous plans. A similar investigation on efficiency is among our future plans.
The reader could be induced to believe that CALM indeed would not hold in synchronous settings, 
since systems of this kind already embed some notion of coordination. 
Indeed, \cite{DBLP:conf/wdag/Ben-ZviM10,DBLP:journals/jancl/Ben-ZviM11} prove that this is not true if a formal definition of coordination is taken into consideration, \ie coordination viewed as a particular state of knowledge required to obtain a shared agreement in a group of nodes. 
%Predicate-level syncausality is somehow related to the emerging field of queries explanation and causality~\cite{RoyS14,MeliouGHK10} in databases. 

Our work is addressing a complementary domain with respect to %the \emph{Boom} group is trying to address with 
\emph{Bloom} \cite{DBLP:conf/cidr/AlvaroCHM11,Alvaro:EECS-2013-133}. %declarative language for easing the implementation of eventual consistent distributed systems.
In Bloom programs, \emph{points of order} are identified: \ie code positions defining a  non-monotonic behavior that could bring inconsistent outcomes \cite{DBLP:conf/cidr/AlvaroCHM11}.
From our perspective, points of order identify where an indirect information flow exists.  
%we favor \emph{points of stability} as a better notation for identifying these circumstances, since \emph{stability} of the input facts is what is actually required. 
In \cite{Alvaro:EECS-2013-133}, two different coordination strategies have been identified at the basis of the cause of inconsistency: \emph{sealing} and \emph{ordering}. They are both comparable to our snapshot coordination. %the former in the general case; the latter in a global version of the FIFO case, where global ordering is obtained by exploiting external coordinations systems such as Zookeper \cite{Hunt:2010:ZWC:1855840.1855851}.
In addition, we have identified broadcasting and synchronized coordination.
%Note that the SCWA is a distributed version of the \emph{Progressive Closed World Assumption} of \cite{DBLP:conf/datalog/Zaniolo12}.
%Coordination points have been called \emph{pipeline-breakers} in \cite{Neumann11}. These in fact are the points where full materialization is strictly needed.
Finally, note that the CALM principle in its original form is satisfied only if no node is granted access to any information on how data was originally distributed;
in this case, in fact, certain weaker forms of monotonic programs can be evaluated in a coordination-free way \cite{ZinnGL12,AmelootKNZ14}.
From our viewpoint, this is possible because  the way data is distributed is already common knowledge before the computation starts, \ie nodes already embed a notion of coordination.
In practice, using synchronous specifications, nodes are able to compute non-monotonic queries in a coordination-free way ``by construction'', without any awareness of how data was initially partitioned.
%We deem this as a realistic assumption in certain specific scenarios. 
We are planning to investigate how the weaker forms of monotonicity identified in \cite{AmelootKNZ14} are related to our work, and whether a tradeoff exists between ``distribution awareness'' and ``synchronization''.
%\mi{Another difference between our study and the one of Ameloot, is that they look for the type of computations which, if the partitioning is right, it can be executed independently on each node while still obtaining a unique consistent output. In our case, we are more interesting in finding which set of queries can be computed in a pipelined way while maintaining consistent results.  }
%\vspace{-4mm}
%if a relationship exists between \emph{partition strategy awareness} awarnconsider also weaker forms of monotonicity.
%For instance, in Section \ref{sec:sync_calm} we used value invention (counting of tuples) to implement the coordination protocols, and 
%in \cite{Ameloot:2014:KNZ} Ameloot et al. show that a specific fragment of Datalog with value invention computes precisely all domain-disjoint-monotone queries:
%this may suggest that all the queries requiring snapshot coordination are domain-disjoint-monotone.
%However, whether the same protocols can be implemented without resorting value invention remains an open question.

\section{Conclusions}
In this paper the CALM principle is analyzed under synchronous and reliable settings.
By exploiting CALM, in fact, we would be able to break the synchronous cage of modern parallel computation models, and provide optimizations such as pipelining and decomposability when allowed by the program logic.
This topic has recently acquired much attention because, in spite  of the increasing number of applications showing better performance (and accuracy) for asynchronous execution over synchronous one~\cite{XieCGZ15,Cui:2014:EBS:2643634.2643639,Niu:2011:HLA:2986459.2986537}, only few practical systems provide this feature as optimization~\cite{Shkapsky:2016:BDA:2882903.2915229,HanD15}.

To reach our goal we have introduced a new abstract model emulating BSP computation, and a novel interpretation of coordination with sound logical foundations in distributed knowledge reasoning.
By exploiting such techniques, we have shown that the CALM principle indeed holds also in \emph{rsync} settings, but in general only for the subclass of monotonic queries defined as connected.
Finally, we have drawn attention to a hierarchy of queries  with related coordination-patterns and we showed how our definition of coordination-freedom is related to the assumptions imposed on the behavior of the system:  our formalization generalizes  the one employed by Ameloot et al. because applicable in synchronous as well as asynchronous settings.

Our next step will be to investigate to which extent the CALM principle is satisfied when queries with \emph{aggregates} are considered; \emph{monotonic aggregation} \cite{RossS97,MazuranSZ13} has been a hot topic in databases for many years:
does a relationship between monotone computation and coordination-freedom exist also for aggregate queries? 

Finally, consider that all the queries in this paper  are formulated in some sub-language of {\sc datalog}$^\neg$.
In the last few years   {\sc datalog}$^{+-}$~\cite{CGLP10} was defined: a family of rule-based languages that extends {\sc datalog} to capture the most common ontology languages for which query answering is tractable, and provides efficiently-checkable, syntactic conditions for decidability and tractability.  We plan to study extensions of our work to (sub-languages of) {\sc datalog}$^{+-}$, in order to apply our results to semantic web settings.

%\section*{Acknowledgements}
%
%We  would  like  to  thank  the  reviewers  for  the  comments  and  suggested  improvements.

\eat{
\todo{Theorem 6.6: in the proof you use the point that "Since each node composing
the network could end up having different (overlapping) partitions of the
initial instance, different nodes might terminate the recursive computation in
different rounds." Again, this is at the core of the CALM conjecture -- if
you don't know when other nodes will be "done", negation can be applied "too
early". Is there a way to take this nugget of reasoning and put it at the
center of the discussion, both in terms of prose and in terms of all the proof
techniques? I.e. the condition of interest is the common knowledge of a
property like local termination. We can use obliviousness. Also for null facts you need to access system relations to know that you are received all the NULL facts. With broadcast instead we can rewrite using body notation and use All for broadcasting instead of $N_H$.}
}

%\appendix
%\section*{APPENDIX}
%\setcounter{section}{0}
%\section{On Independent Specifications}
%\label{sc:independency}
%\input{independency}
%\section{Expressive Power}
%\label{sc:expressive-power}
%\input{expressiveness}

\bibliography{mi}

\end{document}